\documentclass[conference]{IEEEtran}
\IEEEoverridecommandlockouts
\usepackage{cite}
\usepackage{amsmath,amssymb,amsfonts}
\usepackage{algorithmic}
\usepackage{graphicx}
\usepackage{textcomp}
\usepackage{xcolor}
\usepackage{soul} 
\usepackage{color, xcolor} 
\usepackage{amsthm}
\usepackage{url}

\usepackage{listings}
\usepackage{balance}
\usepackage{array}
\usepackage{multirow}
\usepackage{setspace}
\usepackage[ruled,vlined]{algorithm2e}
\usepackage{subfigure} 
\usepackage{listings}
\usepackage{flushend}
\usepackage{enumitem} 
\usepackage{setspace} 
\usepackage{tcolorbox}
\usepackage{makecell}

\newtheorem{theorem}{Theorem}[section]
\newtheorem{lemma}[theorem]{Lemma}
\newtheorem{definition}[theorem]{Definition}

\def\BibTeX{{\rm B\kern-.05em{\sc i\kern-.025em b}\kern-.08em
    T\kern-.1667em\lower.7ex\hbox{E}\kern-.125emX}}
\begin{document}

\title{OneDB: A Distributed Multi-Metric Data Similarity Search System }

\author{%
  $^{\dagger}$Tang Qian, $^{\dagger}$Yifan Zhu, $^{\dagger}$Lu Chen, $^{\dagger}$Xiangyu Ke, $^{\sharp}$Jingwen Zhao,
  $^{\ddagger}$Tianyi Li, $^{\dagger}$Yunjun Gao, $^{\ddagger}$Christian S. Jensen\\
  \IEEEauthorblockN{$^{\dagger}$Zhejiang University \quad  $^\sharp$ Poisson Lab. of Huawei \quad $^{\ddagger}$Aalborg University}
  \IEEEauthorblockA{\textit{$^{\dagger}$\{qt.tang.qian, xtf\_z, luchen, xiangyu.ke, gaoyj\}@zju.edu.cn} $^{\sharp}$zhaojingwen5@huawei.com
    $^{\ddagger}$\{tianyi, csj\}@cs.aau.dk}%
}

\maketitle

\begin{abstract}
Increasingly massive volumes of multi-modal data are being accumulated in many {real world} settings, including in health care and e-commerce. This development calls for effective general-purpose data management solutions for multi-modal data. Such a solution must facilitate user-friendly and accurate retrieval of any multi-modal data according to diverse application requirements. Further, such a solution must be capable of efficient and scalable retrieval.

To address this need, we present OneDB, a distributed multi-metric data similarity retrieval system. This system exploits the fact that data of diverse modalities, such as text, images, and video, can be represented as metric data. The system thus affords each data modality its own metric space with its own distance function and then uses a multi-metric model to unify multi-modal data. The system features several innovations: (i) an extended Spart SQL query interface; (ii) lightweight means of learning appropriate weights of different modalities when retrieving multi-modal data to enable accurate retrieval; (iii) smart search-space pruning strategies that improve efficiency; (iv) two-layered indexing of data to ensure load-balancing during distributed processing; and (v) end-to-end system parameter autotuning. 

Experiments on three real-life datasets and two synthetic datasets offer evidence that the system is capable of state-of-the-art performance: (i) efficient and effective weight learning; (ii) retrieval accuracy improvements of 12.63\%--30.75\% over the state-of-the-art vector similarity search system at comparable efficiency; (iii) accelerated search by 2.5--5.75x over state-of-the-art single- or multi-metric solutions; (iv) demonstrated high scalability; and (v) parameter tuning that enables performance improvements of 15+\%.
\end{abstract}


\section{Introduction}
\label{sec:intro}

Diverse multi-modal data, when put together, hold the potential to offer a comprehensive and holistic perspective on the part of reality that an application concerns. By integrating diverse data from different information sources, it is possible to support a wide range of applications, such as multimedia search in healthcare and personalized recommendations in e-commerce \cite{jeong2024multimodal}.
However, effectively harnessing multi-modal data presents challenges, particularly when having to retrieve relevant data quickly and accurately. 
With the rapid proliferation of smartphones, IoT devices, and sensors, vast volumes of multi-modal data --- including text, images, and geospatial information --- are being generated at an accelerating pace~\cite{palacios2019wip}. 
For example, Instagram, a platform with 1+ billion monthly active users, sees more than 500 million stories shared daily, encompassing a variety of multimedia content such as photos, videos and geo-location data \cite{afyouni2022multi}. 
{The rapid growth of multi-modal data exposes the limitations of current specialized systems, which handle modalities in isolation and ignore any complementary information across modalities. These shortcomings motivate the development of a general-purpose data management system capable of handling each data modality while supporting similarity search across multiple, heterogeneous modalities.}

Consider an example of similarity search based on multi-modal data as shown in Fig.~\ref{fig:rent_example}, where each data entry includes multiple attributes such as X-rays, diagnostic reports, body temperature, and blood saturation levels, represented as vectors, strings, and numerical values. 
Utilizing similarity search can help identify individuals with potential infections by identifying patient entries that are similar to that of an infected individual. 
When relying solely on a single attribute for similarity search, results can include {\em numerous irrelevant candidates}. 
For instance, if we only retrieve entries based on blood saturation levels, entries of patients with similar $SO_2$ levels are returned (see the lower-left corner of Fig.~\ref{fig:rent_example}). However, the patients with these entries may be very different with it comes to other attributes, leading to insufficient medical diagnoses.
In contrast, incorporating multiple criteria --- body temperature, blood oxygen saturation, X-ray images, and diagnostic reports --- during retrieval enables more accurate entry matching that aligns better with the needs of healthcare professionals (see the lower-right corner of Fig.~\ref{fig:rent_example}). 
This example argues that similarity search based on multi-modal data can {\em yield superior results} compared to search based on single-modal data by {\em incorporating diverse information}. 
Furthermore, while the current dataset includes only four attributes with different modalities, additional modalities (locations, contact history, electrocardiograph, etc.) can be {\em integrated} to enhance diagnose quality. Hence, a {\em general-purpose} similarity search system that supports flexible data modalities is essential.

\begin{figure}[t]
	\center
	\includegraphics*[width=1\linewidth]{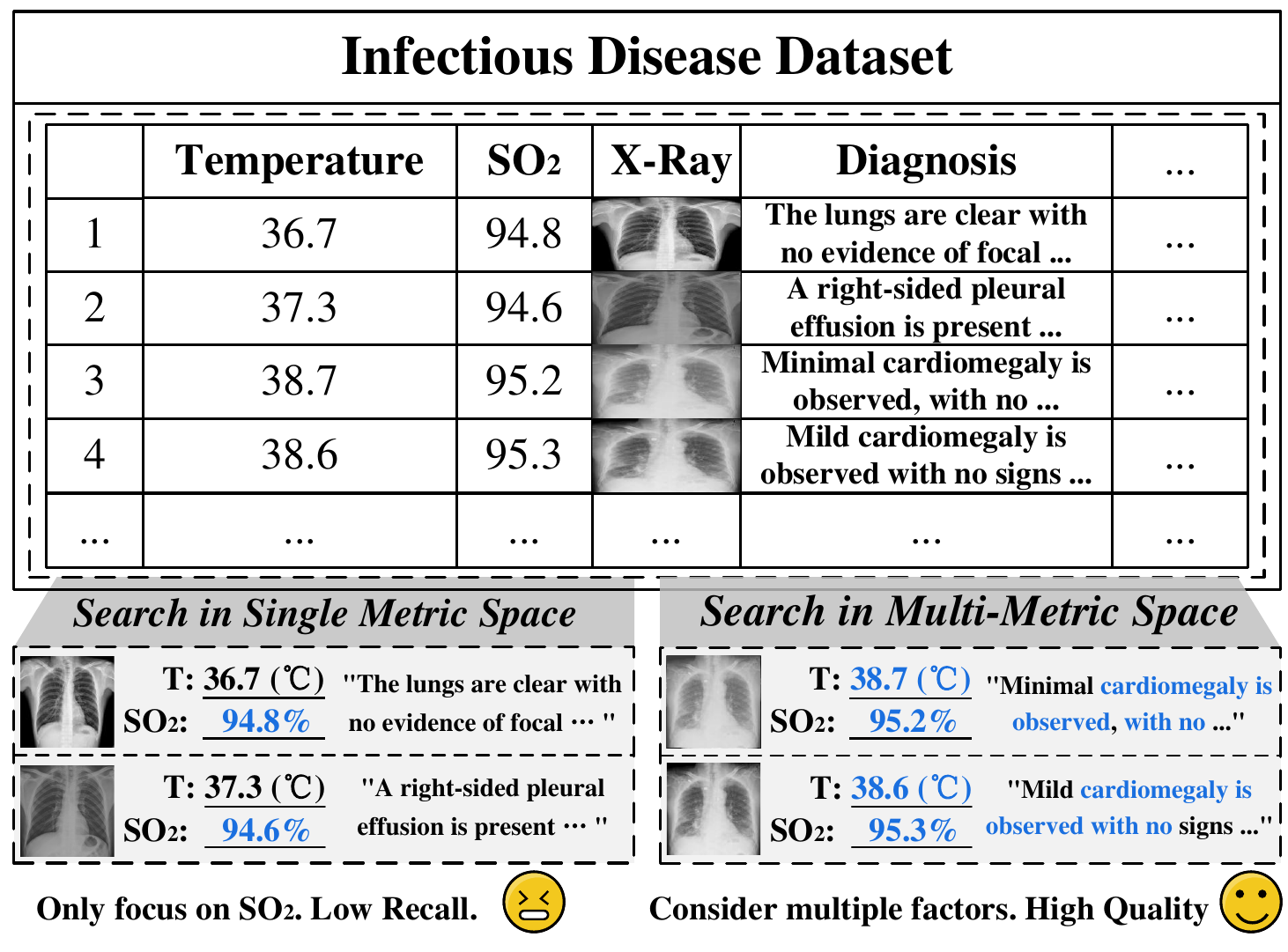}
	\centering
 \vspace{-8mm}
	\caption{Example of similarity search on multi-modal data}
        \vspace{-7mm}
  \label{fig:rent_example}
\end{figure}


While different types of data can be embedded as vectors, allowing for the use of multi-vectors to model multi-modal data. {Multi-vector search provides an approximation solution for the above scenarios, with existing methods mainly represent two approaches: vector fusion and isolated search~\cite{aguerrebere2023similarity}. Vector fusion combines multiple vectors into a single vector by concatenation. This approach increases storage and computation costs and requires decomposable similarity functions. Next, isolated search builds a separate index for each modality. This approach suffers from inefficiencies, such as difficulty in dynamically determining the optimal number of candidates to retrieve for each modality and the need to perform intersection and ranking only after all searches are completed, leading to redundant computations. These limitations impact the efficiency and accuracy of multi-modal data retrieval adversely.} 

However, {\em exact solutions are essential in some applications}. For instance, in epidemic prevention and control, missing potential cases (i.e., approximate solutions) may have severe consequences, 
while in healthcare, the accurate retrieval of patient records combining medical images and clinical notes is vital for diagnosis and treatment planning~\cite{huang2023one}.
Accurate retrieval of relevant data entries involving diverse modalities is critical, as even minor inaccuracies can result in relevant data entries not being retrieved, leading to catastrophic consequences. Based on this reasoning, we adopt the multi-metric space model, as it {\em naturally accommodates the complexities of multi-modal data}. Specifically, metric space modeling is applicable to data across different modalities and enables the use of different distance metrics for different data.
A multi-metric space integrates multiple distinct metric spaces, each tailored to a specific type of data with its particular distance metric. 
Thus, we adopt the multi-metric space model, which facilitates the development of efficient and accurate retrieval. By ensuring accuracy, our approach meets the demands of real-world critical applications.

Existing methods take the form of distributed methods in single or specific metric spaces \cite{shang2018dita, icde/ZhengWZZ0J21, waim/ZhuSKNY12, icde/ChenGLJC15} and standalone methods in multi-metric spaces \cite{zhu2024dimsdistributedindexsimilarity, franzke2016indexing}.
The former fail to support multi-modal data, while the latter cannot handle large volumes of data effectively. Consequently, neither approach facilitates efficient and scalable similarity search on large-volume multi-metric data.
Vector database system, such as {\sf AnalyticDB-V}~\cite{wei2020analyticdb}, {\sf Pase}~\cite{yang2020pase}, {\sf ElasticSearch}~\cite{elasticsearch2018elasticsearch}, and {\sf PGvector}~\cite{pgvector}, focus primarily on similarity search on single vectors. Although {\sf VBase}~\cite{zhang2023vbase} and {\sf Milvus}~\cite{wang2021milvus} support multi-vector similarity search, they have notable limitations.
{\sf VBase} operates as a standalone solution, while {\sf Milvus} employs a naive baseline that aggregates the top-$k$ results from each modality. This naive approach selects an appropriate $k$ value for each modality that is both low and enables an accurate final top-$k$ result, as each modality contributes differently to the overall outcome. A larger $k$ can reduce query efficiency, whereas a smaller $k$ can compromise result accuracy, as shown in Section~\ref{subsec:expsimilarity}.
Three key challenges must be addressed to support efficient and scalable similarity search on large-volume multi-metric data.


\textbf{\textit{Challenge I: How to provide proper modality weights for multi-modal data?}} 
During similarity search on multi-modal data, a common approach involves computing similarities for each modality separately and then perform a weighted summation to obtain an overall similarity. 
The choice of weights significantly influences the results. 
Most previous studies \cite{franzke2016indexing, zhu2024dimsdistributedindexsimilarity} assume that modality weights are {\em fixed} and provided by users during index construction.
However, this assumption does not reflect real-world scenarios where 
(1) different users have {\em varying preferences}, necessitating adaptable weights, 
and (2) users are not aware of their precise preferences, often expressing their needs {\em in the forms of desired outcomes} (e.g., exemplar top-$k$ results) rather than specific weight values.

To facilitate flexible similarity queries with variable weights, we propose a two-tiered indexing approach. At the global level, we apply uniform weights across all modalities during index construction on a central master node, allowing for efficient filtering of modalities not involved in a query. Concurrently, local index construction on worker nodes employs separate indexes for each modality, enabling targeted filtering for improved efficiency.

We propose a modality weight learning model for use when explicit weights are not provided. 
Recognizing the challenge of acquiring numerous query cases (i.e., query objects together with $k$ similar result objects), we employ two strategies: 
(i) To effectively learn proper weights given limited query cases, we design an effective positive and negative sample generation strategy for use during the training. 
In experiments, we find that just 30 query cases are sufficient to obtain high-quality weights.  
(ii) Our model is lightweight, optimizing a loss function that minimizes the distance from the query point to its corresponding positive samples while maximizing the distance to the negative ones.
Our experiments show that training for less than 100 seconds is sufficient to achieve weights that enable queries with approximately 90\% recall 
 {(Exp.~\ref{fig:weight})}.

\textbf{\textit{Challenge II: How to obtain high query efficiency in a distributed environment?}} 
It is crucial to utilize computational resources in a distributed environment fully through effective load balancing to support efficient similarity search on large volumes of multi-modal data. However, distributing data evenly across modalities is challenging.
Two main partitioning strategies exist: data granularity partitioning (horizontal partitioning)~\cite{zhu2024dimsdistributedindexsimilarity} and modality granularity partitioning (vertical partitioning)~\cite{zhu2022desire}.
Data granularity partitioning operates at the level of multi-modal entries, or records, partitioning an entire dataset into subsets with an equal number of records. While this approach ensures that all modalities are present on a single node, it {\em complicates the retrieval of similar data under varying weights}. 
In contrast, modality granularity partitioning operates at the level of multi-modal attributes and partitions the dataset according to different modalities, clustering data more effectively but {\em incurring high communication costs} to access all modalities across distributed nodes.
To address the second challenge, we propose a dual-layer indexing strategy that combines these two partitioning approaches. 
On the master node, we perform coarse-grained data granularity partitioning, assuming uniform modality weights since these cannot be predetermined. We cluster the data using all modalities with equal weights. Specifically, we first select a high-quality pivot from each metric space, mapping each space into a one-dimensional vector space. Given $m$ modalities, this results in vectors in an $m$-dimensional vector space. 
We then employ the {\sf RR*-tree} \cite{franzke2016indexing} to partition the dataset evenly, ensuring that the data in its leaf nodes are similar by considering all mapped dimensions. 
We implement fine-grained modality granularity partitioning on the worker nodes. 
Each modality is clustered using different index structures (e.g., using inverted indexes for strings and R-trees for low-dimensional vector data) to enhance clustering performance and facilitate effective pruning. Consequently, a multi-metric space index forest is constructed at each worker node. 
Furthermore, we develop effective pruning strategies for both the global and local levels to optimize query efficiency for high-similarity search.

Traditional relational database systems, while effective for structured data, struggle to efficiently support multi-modal data and scalable execution over complex similarity queries. In contrast, OneDB is built on Apache Spark, which provides 1) unified abstractions (DataFrame, Dataset) for heterogeneous data, 2) a distributed in-memory architecture with high scalability and fault tolerance, and 3) an extensible SQL optimizer that enables seamless integration of custom index structures and similarity operators. This allows OneDB to integrate diverse index structures (e.g., R-tree, MVP-tree, inverted index) into a unified system with scalable multi-modal similarity search.  Prior studies~\cite{sun2017dima, xie2016simba} have also demonstrated Spark's effectiveness in distributed similarity analytics. 



\textbf{\textit{Challenge III: How to achieve end-to-end system parameter auto-tuning?}} 
While the proposed dual-layer indexing strategy aims to utilize computational resources fully by distributing multi-modal data evenly, real distributed environments present challenges such as communication overhead and variable computational power among nodes, influenced by factors like network latency and additional workloads (e.g., Google found that I/O activity can vary by up to 90\% in their distributed systems, leading to substantial differences in processing times among nodes~\cite{53317}). 
Existing studies \cite{zhu2024dimsdistributedindexsimilarity} often rely on idealized cost models based on assumptions (e.g., uniform data distribution) that do not hold true in pratice. 
Moreover, studies often overlook the significance of system parameters (e.g., Spark configuration settings, sampling rates), which often require the expertise of skilled database administrators to tune manually.
To address these challenges, we propose an end-to-end auto-tuning module that leverages deep reinforcement learning. In this approach, a reward function evaluates tuning performance and provides feedback to the learning model, facilitating continuous improvement. Additionally, we employ the deep deterministic policy gradient method, enabling the tuning of parameters in continuous space. This adaptive tuning mechanism improves system performance by aligning resource allocation with the fluctuating demands of multi-modal data processing (Exp.~\ref{fig:reward_function}).

In summary, our main contributions are as follows.
\begin{itemize}
[topsep=2pt,itemsep=2pt,parsep=0.5pt,partopsep=0.5pt,leftmargin=*]
\item{} \textit{Distributed Multi-Metric Similarity Search Framework.} We propose OneDB, a distributed framework for multi-metric space similarity search, supporting flexible and efficient multi-metric exact similarity search. To the best of our knowledge, this is the first general distributed system to support similarity search in multi-metric space.
\item{} \textit{Multi-Metric Weight Learning.} We propose a lightweight metric weight learning model that efficiently captures the inter-modality weight relationships for user-targeted queries.
\item{} \textit{Effective Dual-Layer Indexing Strategy.} We design a hybrid modality granularity indexing strategy that supports the flexible combination of multi-metric spaces while avoiding high index construction and search costs.

\item{} \textit{End-to-End Parameter Auto-Tuning Module.} We design an end-to-end auto-tuning module that leverages deep reinforcement learning to optimize parameters through a reward function evaluation and deep deterministic policy gradient.

\item{} \textit{Experimental Study.} Experiments on  real-life and synthetic 
 datasets offer insight into OneDB: it i) achieves speedups of {2.5--5.75x} compared to state-of-the-art single- or multi-metric space solutions; ii) improves accuracy by {12.63\%--30.75\%} with comparable efficiency compared to state-of-the-art vector database system; iii) {exhibits high scalability}; iv) more efficient and effective weight learning compared to the random strategy; and v) achieves performance improvements of 15\%--17\% after parameter tuning.
\end{itemize}

The rest of the paper is organized as follows. Section~\ref{sec:related} reviews related work. Section~\ref{sec:problem} presents problem statement. Section~\ref{sec:framework} provides an overview of our system. 
Sections~\ref{sec:index&search}--\ref{sec:tuning} present the detailed techniques of our system.
Section~\ref{sec:experiments} reports on a comprehensive experimental evaluations. Finally, Section~\ref{sec:conclusions} concludes the paper.
\section{Related Work}
\label{sec:related}

\vspace{-1mm}
\subsection{Metric Similarity Search}
Similarity search in a single metric space aims to identify objects that are most similar to a given query object based on a single metric or distance measure~\cite{chen2022indexing}. Indexing solutions have been developed to accelerate similarity search: compact partitioning, pivot-based, or hybrid methods.

\noindent\textbf{Compact Partitioning Methods} partition a metric space into compact sub-regions. They are designed to reduce a search space by discarding irrelevant regions through pruning techniques. Notable examples include the Generalized Hyperplane Tree~\cite{ipl/Uhlmann91}, which employs hyperplane-based partitioning, and the {\sf M-Tree}~\cite{vldb/CiacciaPZ97}, which utilizes ball partitioning.

\textbf{Pivot-based Methods} transform a metric space into a vector space using a set of reference points, or pivots. By computing distances between objects and selected pivots, geometric properties can be applied to filter out non-matching objects efficiently. Techniques like the Linear Approximate Similarity Search Algorithm ({\sf LAESA})~\cite{mico1994new} and the {\sf MVP-Tree}~\cite{sigmod/BozkayaO97,tods/BozkayaO99} exemplify this approach, where pre-computed distances facilitate fast similarity search.
\textbf{Hybrid Methods} integrate the advantages of compact partitioning and pivot-based techniques to achieve better performance. For example, the Geometric Near-neighbor Access Tree ({\sf GNAT})~\cite{vldb/Brin95} integrates hyperplane partitioning with the use of pivots to enhance search efficiency, while other methods like the Bisector Tree~\cite{kalantari1983data} utilize a combination of ball partitioning and pivots to improve the accuracy and speed of similarity searches.

Although different solutions have been developed for similarity search in single-metric spaces, they cannot be efficiently applied to multi-metric spaces. Motivated by it, several indexing methods proposed for similarity search in multi-metric spaces can be classified as combined or separate methods.

\noindent\textbf{Combined Methods.} Combined methods combine multiple distance metrics from different metric spaces linearly into a single metric. {\sf QIC-M-tree}~\cite{tods/CiacciaP02} applies a user-defined distance that can be regarded as a combination of multiple metrics to build the index. The {\sf $M^{2}$-Tree}~\cite{ciaccia2000m2} and {\sf $M^{3}$-Tree}~\cite{bustos2012adapting} adapt {\sf M-tree} indexes to handle multiple metrics in a unified manner. Bustos et al. \cite{bustos2012adapting} also propose a general methodology to adapt existing single metric space indexes, such as the M-Tree and List of Clusters, into multi-metric space indexes. 
 The {\sf RR*-tree}~\cite{7498318} employs reference-object embedding to map multi-metric data into a single vector space, utilizing an R-tree to index the embedded objects. Combined methods simplify index construction. However, whenever a new metric is added, or weights are adjusted, they typically require rebuilding of the entire index, leading to high reconstruction costs.


\noindent\textbf{Separate Methods.} In contrast, separate methods maintain individual indexes for each metric space, allowing for the preservation of distinct metric characteristics. {As pivot-based indexing schemes achieve high efficiency due to the pruning power of pivots, {\sf C-forest}~\cite{celik2006new} and a pivot-based index~\cite{sac/BustosKS05} select high-quality pivots to index objects in each metric space. {\sf Spectra}~\cite{zabot2019efficient} utilizes pivots to embed and index each metric space, and metric spaces with low correlations are indexed together. {\sf DESIRE}~\cite{zhu2022desire} selects high-quality pivots to cluster objects into compact regions in each metric space and uses $B^{+}$-trees to index the distances between pivots and objects. 
Separate methods support flexible combinations of any sub-set of multi-metric spaces and variable weights of each metric. However, they incur significant overhead due to the necessity of constructing and searching multiple indexes.}

Finally, we note that the above multi-metric index solutions are standalone and cannot support large-volume multi-modal data. Hence, we aim to develop distributed indexing techniques for efficient similarity search in multi-metric space.


\vspace {-2mm}
\subsection{Distributed Similarity Search}

The increasing data volumes calls for distributed similarity search methods. Two main categories of methods exist: distributed implementations of existing metric indexes and distributed indexes for specific data types.

\noindent\textbf{Distributed Implementations of Metric Indexes.} 
Methods in this category distribute objects across multiple nodes and utilize existing metric indexes for local indexing. 
Batko et al.~\cite{batko2005similarity} use an embedding method to map objects from metric spaces into a distributed tree structure (Address Search Tree, AST) and utilize GHT for indexing data at worker nodes. Similarly, {\sf M-Chord}, {\sf MT-Chord}, and the {\sf M-index}~\cite{is/NovakBZ11,dpd/DoulkeridisVKV09,waim/ZhuSKNY12} apply iDistance~\cite{tods/JagadishOTYZ05} for global partitioning and employ the B$^+$-tree for local indexing. The asynchronous metric distributed system~\cite{dase/YangDZCZG19} proposes partitioning objects via pivot mapping into minimum bounding boxes and utilizes publish/subscribe communication to support asynchronous processing. Moreover, Sun et al.~\cite{sun2017dima} allocate objects using a global mapping function and store them using an inverted index. Still, this proposal only supports similarity search based on set similarity (e.g., Jaccard similarity) and character similarity.

\begin{table}\label{table:notations}
	\centering   
	\caption{{Frequently used notations}}
    \vspace{-3mm}
	\label{table:notations}
	\small
	\setlength{\tabcolsep}{3pt}
	\begin{tabular}{p{2.8cm}p{5.4cm}}
		\hline
		\textbf{Notation} & \textbf{Description} \\
		\hline
        $q$, $o$ & A query, an object in a metric space \\
        $q^M$, $o^M$ & A query, an object in a multi-metric space \\
        $o^i$ & The component of $o^M$ in the $i$-th metric space \\
        $m$ & The number of metric spaces \\
        $\delta_i(\cdot,\cdot)$ & A distance metric in $i$-th metric space \\
        $\Delta$ & A set of distance metrics $\delta_i$ ($1 \leq i \leq m$) \\
        $W$, $\delta^W(\cdot,\cdot)$ & A weight vector, a multi-metric distance  \\
        $\text{MMRQ}(q^M, W, r)$ & A multi-metric range query for a query object $q^M$ with weight $W$ and radius $r$ \\
        $\text{MMkNNQ}(q^M, W, k)$ & A multi-metric $k$-nearest neighbor query for a query object $q^M$, a weight vector $W$, and an integer $k$ \\
		\hline
	\end{tabular}
	\vspace{-0.5cm}
\end{table}

\noindent\textbf{Specialized Distributed Indexes.} These indexes are tailored to specific data types (such as trajectories or time series) and often incorporate domain-specific optimizations. For instance, {\sf DITA}~\cite{shang2018dita} and {\sf Simba}~\cite{xie2016simba} employ trie-based indexing for trajectories, while {\sf PS2Stream}~\cite{chen2017distributed} uses the Grid-Inverted-Index for distributed spatio-textual data streams.

We note that the above distributed indexes are for single-metric space and are not well-suited for generalized multi-metric space. Hence, we propose a distributed indexing method enabling efficient similarity queries in multi-metric space. Our method builds a combined global index structure and  separated local indexes to achieve both efficiency and flexibility.


\vspace {-2mm}
\section{Preliminaries}
\label{sec:problem}

We proceed to introduce the concept of a {\em multi-metric space} along with key definitions relevant to similarity search. Table~\ref{table:notations} summarizes frequently used notation. 

A {\em metric space} is a pair $(M, \delta)$, where $M$ denotes a set of objects (or a domain) and $\delta$ is a distance function (metric) that quantifies the similarity between any two objects $q$ and $o$ in $M$. 
The distance metric $\delta$ adheres to four properties, which are fundamental to metric spaces: 
\textbf{(i) \textit{symmetry}:} $\delta(q, o) = \delta(o, q)$; 
\textbf{(ii) \textit{non-negativity}:} $\delta(q, o) \geq 0$; 
\textbf{(iii) \textit{identity}:} $\delta(q, o) = 0$ if and only if $q = o$; 
and \textbf{(iv) \textit{triangle inequality}:} $\delta(q, o) \leq \delta(q, o') + \delta(o, o')$. 
It is essential to note that in a standard metric space, the set $M$ contains data objects of a single type; hence, all objects are compared using the same metric $\delta$.

Next, a {\em multi-metric space} $(\mathcal{M}, \Delta)$ integrates multiple distinct metric spaces. 
In this context, $\mathcal{M}$ is a collection of domains $M_i$ (where $1 \leq i \leq {m}$), and $\Delta$ is a corresponding set of distance metrics $\delta_i$, each associated with its respective domain $M_i$. 
This structure enables the simultaneous modeling of multiple types of data governed by its own distance metric. 
In a multi-metric space, an object is represented as {$ o^M= \{o^i \mid 1 \leq i \leq m\}$}, where each $o^i$ resides in the domain $M_i$ of the corresponding metric space. 

\begin{figure}[tb]
\vspace{-1cm}
    \centering    \includegraphics[width=0.9\linewidth,keepaspectratio]{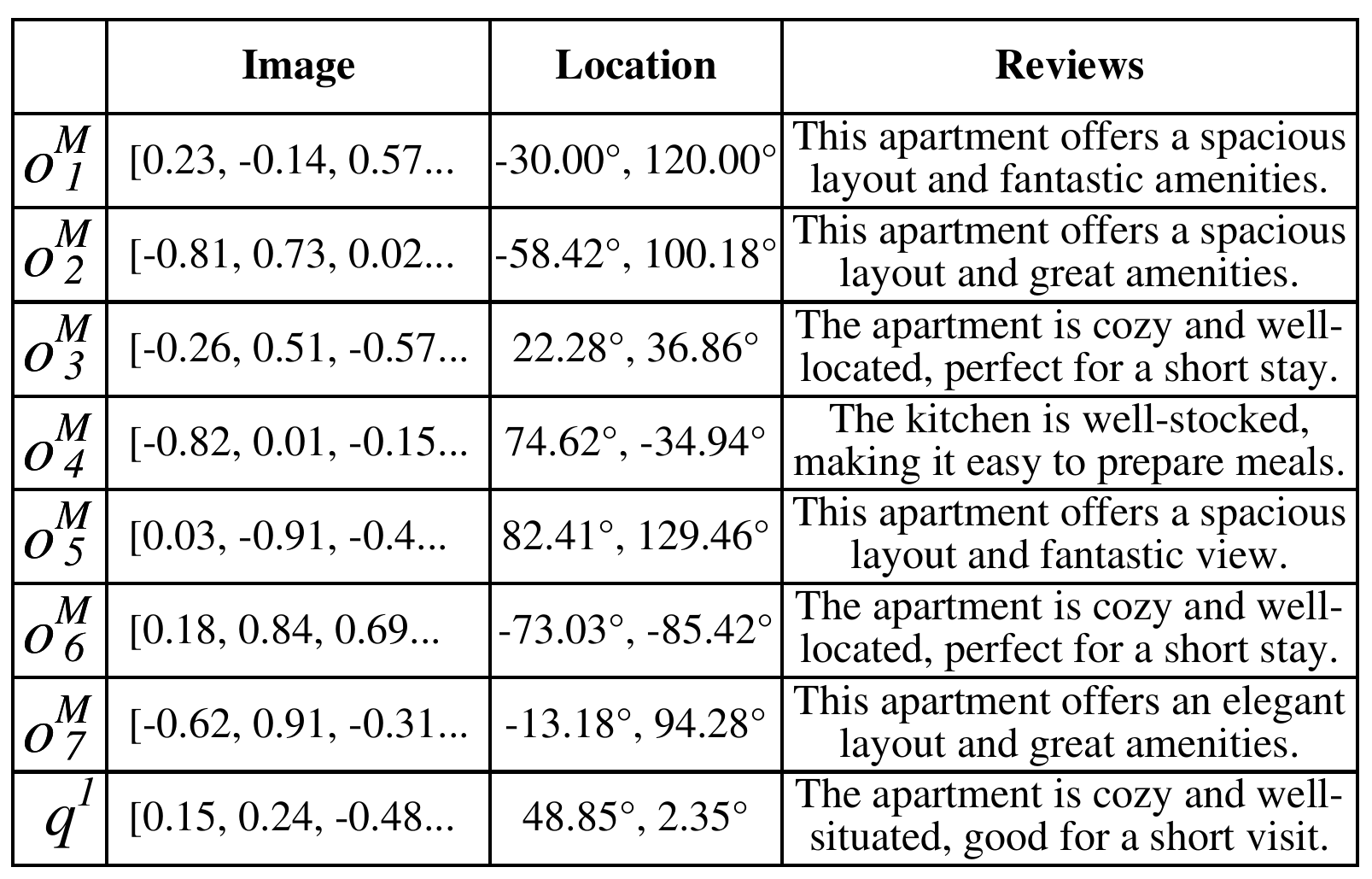} 
    \vspace{-2mm}
   \caption{A multi-metric space example}
\vspace{-6mm}
\label{fig:multi_metric_example}
\end{figure}

Figure~\ref{fig:multi_metric_example} provides an example of a multi-metric space object set $S$,
where each object represents an apartment described from multiple perspectives. These attributes include: 
{(i) image, (ii) location, and (iii) customer reviews.} 
{Specifically, the multi-metric space aggregates three individual metric spaces $(\mathcal{M}, \Delta)$ (i.e., $m = 3$): \(M_1\) is a high-dimensional vector space for images, where the $L_1$-norm is used as the distance metric; \(M_2\) is a two-dimensional vector space representing the location (latitude and longitude), where the $L_2$-norm serves as the distance metric; finally, \(M_3\) represents the string domain for text reviews, where edit distance is used.}


The overall similarity between objects in the multi-metric space integrates distances from each metric space:

\begin{definition}  
\label{defn:MultiMetricDistance}
{\bf (Multi-Metric Distance.)} 
Given a weight vector \( W = (\omega_1, \ldots, \omega_M) \), where \( \omega_i \in \mathbb{R} \) and \( \omega_i \in [0, 1] \) represent the importance of each metric space, the multi-metric distance \( \delta_W(\cdot, \cdot) \) between two objects \( \mathbf{q}^M \) and \( \mathbf{o}^M \) in a multi-metric space is defined as: \(\delta_W(\mathbf{q}^M, \mathbf{o}^M) = \sum_{\delta_i \in \Delta} \omega_i \cdot \delta_i(q^i, o^i)\). where \( \delta_i(\cdot, \cdot) \) is the distance metric associated with the \( i \)-th metric space, and \( q^i \) and \( o^i \) are the components of \( \mathbf{q}^M \) and \( \mathbf{o}^M \) in the \( i \)-th metric space.
\end{definition}

{In the definition of \textit{multi-metric distance}, \(\omega_i\) controls the contribution of the metric \(\delta_i\) to the overall similarity. In Fig.~\ref{fig:multi_metric_example}, the weight vector \(W = (1, 0, 1)\) indicates that, when performing a multi-metric similarity search, only the image and review modalities are considered. In different search scenarios, the user's intent underlying a similarity search may align with specific metric weights. For instance, some users may prioritize image similarity in the search results, corresponding to an increased weight for the image metric space. The assignment of metric weights can be provided either by the user or determined through the weight learning module in our system. To ensure balanced data ranges across different modalities, distances \(\delta_i(q^i, o^i)\) in metric space \((M_i, \delta_i)\) are normalized by dividing them by twice the median of all observed distances.}

Based on the multi-metric distance \(\delta^W(\cdot, \cdot)\), we can formally define two typical types of similarity search in multi-metric spaces: the multi-metric range query and the multi-metric \(k\)-nearest neighbor (\(k\)NN) query.

\begin{figure}[t]
\vspace{-1cm}
    \centering    
    \includegraphics[width=1\linewidth,keepaspectratio]{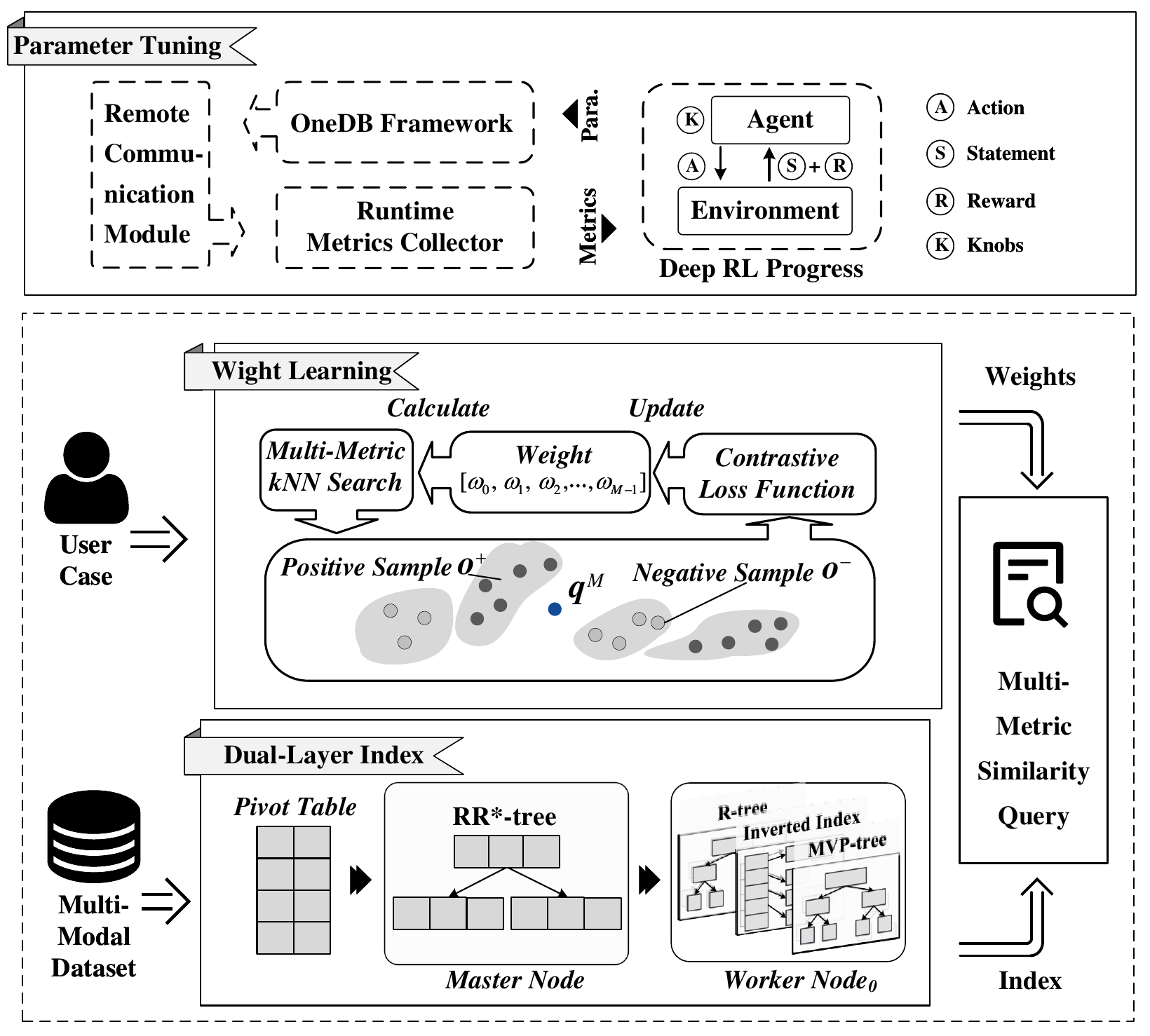} 
    \vspace{-6mm}
   \caption{OneDB framework}
   \label{fig:onedb-framework}
\vspace{-0.6cm}
\end{figure}

\begin{definition}  
\label{defn:MMRQ}
{\bf (Multi-Metric Range Query)}  
Given a set of objects \(S\), a query object \(q^M\), a weight vector \(W\), and a search radius \(r\) in a multi-metric space, a multi-metric range query (MMRQ) identifies all objects in \(S\) such that their distance to \(q^M\) is within \(r\): 
\(
\text{MMRQ}(q^M, W, r) = \{o^M \mid o^M \in S \text{ and } \delta^W(q^M, o^M) \leq r\}.
\)
\end{definition}

\begin{definition}  
\label{defn:MMkNNQ}
{\bf (Multi-Metric \(k\)NN Query)}  
Given a set of objects \(S\), a query object \(q^M\), a weight vector \(W\), and an integer \(k\) in a multi-metric space, a multi-metric \(k\)-nearest neighbor query (MM\textit{k}NNQ) seeks to find the \(k\) objects in \(S\) that are closest to \(q^M\) in terms of similarity, i.e.,
\(
\text{MMkNNQ}(q^M, W, k) = \{S' \mid S' \subseteq S, \ |S'| = k, \ \text{and } \forall s^M \in S'( \ \forall o^M \in (S - S')( \ \delta^W(q^M, s^M) \leq \delta^W(q^M, o^M)))\}.
\)
\end{definition}

In the example in Fig.~\ref{fig:multi_metric_example}, query point \(q^M\) represents an apartment with the following attribute values: an image of the apartment (embedded as a vector [0.23, -0.14, ...] using a deep learning model), an address with geographic coordinates of 48°85' N latitude and 2°35' W longitude, and a customer review stating "The apartment is cozy and well-
situated, good for a short visit." The multi-metric range query $\text{MMRQ}(q^M, W, r)$ returns $o^M_2$ and $o^M_3$, while the multi-metric \(k\)NN query returns $o^M_2$ (with $k$ = 1).

\section{The OneDB Framework: Overview}
\label{sec:framework}

The proposed OneDB offers a comprehensive solution for multi-metric {exact} similarity search, supporting user-friendly SQL-based multi-metric range queries and \(k\)NN queries. To support a large volume of data, OneDB is built on Spark.

\subsection{Overview}

Figure~\ref{fig:onedb-framework} provides an overview of the OneDB framework. 
OneDB consists of four components: i) multi-metric weight learning module, ii) dual-layer indexing module, iii) similarity search module, and iv) end-to-end parameter tuning module.

\noindent\textbf{Multi-Metric Weight Learning.}
We propose an effective multi-metric weight learning model capable of identifying the relative importance of different modalities in multi-metric similarity computations. For any pair of objects \( q^M \) and \( o^M \), the model assigns specific weights to each modality's metric space, thereby adjusting the contribution of each modality to the similarity computation. According to Definition~\ref{defn:MultiMetricDistance}, the similarity computation \(\delta_W\) between \( q^M \) and \( o^M \) combines their distances \(\delta_i\)  in each single metric space with the corresponding weights \(\omega_{i}\), i.e., the weighted distance is calculated as \(\delta_W({q}^M, {o}^M) = \sum_{\delta_i \in \Delta} \omega_i \cdot \delta_i(q^i, o^i)\).
To effectively learn the weights, 
the \textit{k}NN search method is first employed to identify positive samples that exhibit high similarity to \( q^M \), while the left objects are all negative samples. Using a contrastive loss function, the model distinguishes the distances from the target sample to the negative and positive samples separately in the feature space. Next, the model adaptively adjusts the weights of each modality to reflect their relative importance in the similarity metric through iterative training. Finally, the learned weights are applied to similarity search tasks, improving retrieval performance.

\noindent\textbf{Dual-Layer Indexing.} We design a dual-layer index structure for a distributed environment. At the global level, a global $RR^*$-tree is constructed on the master node to distribute data to worker nodes evenly. At the local level, a multi-metric space index forest is constructed on each worker node, consisting of separate indexes built for each data modality. Our system supports three types of separate indexes, including R-tree (used to index low-dimensional vectors) and {mvp-tree} (used to index high-dimensional vectors), {for text data, we build inverted indexes for efficient candidate filtering before applying edit distance~\cite{sun2019balance}.} Although only three types of indexes are supported in our current implementation, our system is flexible enough to integrate any kind of index structure designed for various data types. Additionally, an automatic index selection method can be designed to select the proper index structure for each data modality.
This indexing scheme allows OneDB to maintain a low cost for global index construction while enabling efficient similarity search locally.

\noindent\textbf{Similarity Search.} 
Based on the dual-layer index structure, efficient multi-metric range query and multi-metric $k$NN query methods are proposed. Effective pruning lemmas are also designed at both global and local levels during the search.

\noindent\textbf{End-to-End Parameter Tuning.}
Due to the various parameters involved in our system, it is difficult for users to set proper values for all parameters, especially in a dynamic environment. OneDB integrates a reinforcement learning-based adaptive tuning method to adjust the tunable configuration knobs of the search framework dynamically. Specifically, OneDB involves an agent interacting with the environment to incrementally optimize the knob settings at each time step \( t \) through incremental training. At each state \( s_t \), the agent combines the current features with an externally computed reward \( r_t \) and selects an appropriate action \( a_t \) based on the current policy \( \tau \), which corresponds to adjusting the knobs. By comparing the performance differences between the previous state \( s_t \) and the subsequent state \( s_{t+1} \), the agent iteratively computes the reward function to effectively capture system performance changes, guiding parameter settings for high performance.

\vspace{-0.1cm}

\subsection{SQL Statements}
OneDB offers a flexible and user-friendly interface. We extend Spark SQL to support two types of typical similarity searches and offer specific SQL usage examples. {These operators can be freely combined with standard SQL features, inheriting full structured query support from Spark SQL.}

\noindent\textbf{Multi-Metric Range Query.} We use {\sf SELECT} statement with ODBRANGE to retrieve objects from table T that have a multi-metric distance less than \(r\) from the query point \(q^M\), where \textit{W} represents the query's weight vector. Note that users can either manually set the weight vector or obtain high-quality weights via the multi-metric weight learning module.


\lstset{
  mathescape=true,
  basicstyle=\bfseries\sf\small, 
  backgroundcolor=\color{yellow!10}, 
  columns=fullflexible, 
  frame=single, 
  rulecolor=\color{red}, 
  frameround=tttt 
}

\begin{lstlisting}
SELECT * FROM T WHERE T.col IN ODBRANGE($q^M$, W, r)
\end{lstlisting}


\noindent\textbf{Multi-Metric \(k\)NN Query.} We use {\sf SELECT} statement with ODBKNN to find the top \(k\) objects in table T that are closest to the query point \(q^M\) in terms of multi-metric distance.

\begin{lstlisting}[mathescape]
SELECT * FROM T WHERE T.col IN ODBKNN($q^M$, W, k)
\end{lstlisting}


\vspace{-3mm}
\section{Multi-Metric Weight Learning}
\label{sec:weight_learning}
We introduce a multi-metric weight learning model that provides proper modality weights for multi-metric data distance computation. The critical challenges of multi-metric weight learning model lie in two aspects: i) how to generate training data with limited query cases provided by users; and ii) how to efficiently and effectively learn the weights. Our lightweight learning model is based on contrastive learning. To begin, given an anchor object (i.e., a query point), we acquire its positive and negative examples. Subsequently, we construct and minimize a contrastive loss function to learn the relative weights. We first provide the training data generation method and loss function in the following.

\vspace {-2mm}
\subsection{Training Data Generation}

The users provide a set of \( N \) query points and their corresponding true \( k \)-nearest neighbors, denoted as \( G^+ \). However, in real-life applications, $N$ is small, as we cannot ask users to provide a large number of query cases. To generate positive and negative samples during each iteration of training, we first perform multi-metric \( k \)-NN queries on $N$ query points using the weight vector \( \bar{W} \) obtained from the previous iteration, resulting a multi-metric \( k \)-NN search result set $F$. Then in each iteration, positive and negative samples are generated.

\noindent\textbf{Positive Samples.} The intersection of true \( k \)NNs provided by users and \( k \)NNs obtained by $\text{MMkNNQ}(q^M, \bar{W}, k)$ are considered as positive samples \( o^{M+} \):
\(
o^{M+} \in G^+(q^M) \cup F(q^M)
\). 

\noindent\textbf{Negative Samples.}  The non-true \( k \)NNs in the result set \( F \) are treated as negative samples \( o^{M-} \):
\(
o^{M-} \in F(q^M) \setminus G^+(q^M)
\).


\vspace {-2mm}
\subsection{Loss Function}
Our loss function aims to make the query point more similar to its positive samples \( o^{M+} \) while increasing the difference from the negative samples \( o^{M-} \), which is defined below.
\begin{equation}
     L = \frac{1}{m} \sum_{q^{M} \in Q} -\log \frac{\sum_{o^{M+}} e^{\delta(q^M, o^{M+})}}{\sum_{o^{M+}} e^{ \delta(q^M, o^{M+})} + \sum_{o^{M-}} e^{\delta(q^M, o^{M-})}} 
\end{equation}
\noindent
where $m$ denotes the number of metric spaces, $\delta(q^M, o^{M+})$ denotes the multi-metric distance between the query point $q^M$ to its positive sample $o^{M+}$, while $\delta(q^M, o^{M-})$ denotes the multi-metric distance between the query point $q^M$ to its negative sample $o^{M-}$.


\begin{algorithm}[t]
\small
\SetNlSty{small}{}{:}
\SetCommentSty{footnotesize}
\LinesNumbered
\DontPrintSemicolon
\caption{Global Index Construction}
\label{algo:rrtreeconstruction}
\KwIn{Dataset $D$ in the multi-metric space $\mathcal{M} = \{M_1, M_2, \dots, M_{m}\}$, the number $m$ of metric spaces, Reference set $R = \{p_1, p_2, \dots, p_{m}\}$}
\KwOut{Constructed $RR^*$-tree with $R^*$-tree indexing}
    \ForEach{object $o^M \in D$}
    {
        $v_o \gets$ Initialize an empty feature vector\;
        \ForEach{ metric space $M^i \in \mathcal{M}$}
        {
            {
                $v_{o}[i] \gets \delta^i(p^i_{i}, o^i)$\; 
            }
        }
        Insert $v_o$ into the $R^*$-tree\;
    } 
    \KwRet the constructed $RR^*$-tree\;
\end{algorithm} 

The loss function is to optimize the weights in the multi-metric space, \( W = [\omega_0, \omega_1, \dots, \omega_{m-1}] \), by minimizing \( L \), thereby enhancing the similarity between the query object \( q^M \) and its positive sample set, while reducing the similarity between the query object and the negative sample set. The initial weights \( W \) are randomly generated.
By adjusting the weights \( \omega_i \), the model aims to increase the similarity distances of positive samples, $e^{\delta(q^M, o^{M+})}$, while decreasing the similarity distances of negative samples, \( e^{\delta(q^M, o^{M-})} \). During the optimization process, if the similarity of a negative sample exceeds that of a positive sample in the query results, the loss function \( L \) will increase significantly, reflecting the inaccuracy of the current weights. Consequently, gradient descent is applied to update \( W \), gradually reducing the value of \( L \). This ensures that the similarity between the query object and the positive samples increases while the similarity to negative samples decreases, leading to an optimized set of weights.


The proposed multi-metric weight learning approach effectively captures the relative weights between different modalities. Combining multi-metric distance computation with the modality weight learning module allows OneDB to seamlessly integrate information from various modalities into the similarity computation process. Additionally, different from deep learning methods that are costly, our weight learning model only constructs a contrastive loss function and minimizes it to learn the relative weights, which is lightweight. As a result, our model can converge to an optimal weight configuration with fewer iterations. In our experiments,  training under 100 seconds is able to converge to 90\% recall.



\vspace{-2mm}
\section{Multi-Metric Similarity Search}
\label{sec:index&search}
We proceed to introduce the dual-layer index structure and then present efficient multi-metric similarity search methods.
\begin{figure*}[tb]
\vspace{-12mm}
    \centering    \includegraphics[width=1\linewidth,keepaspectratio]{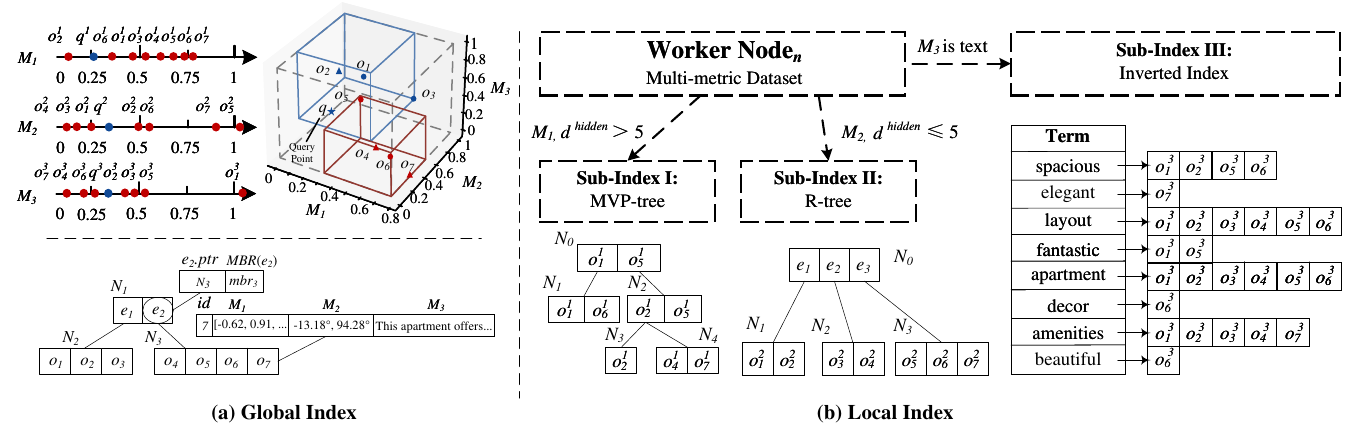} 
    \vspace{-6mm}
   \caption{An example of dual-layer index}
\vspace{-5mm}
\label{fig:two-layer-index}
\end{figure*}

\vspace{-2mm}
\subsection{Dual-Layer Index}
\label{subsec:index}
To ensure an efficient multi-modal similarity search, we design a dual-layer index. Specifically, we utilize an effective multi-metric space indexing structure, the $RR^*$-Tree ({Reference-$R^*$-Tree}) at the master node, which facilitates effective partitioning via pivot mapping with minimum bounding rectangles (MBR) at the global level. At the worker nodes, we construct a index forest for each partition, i.e., a specific index is constructed for each metric space.

\noindent\textbf{Global Index.} In the master node, we construct the global $RR^*$-Tree index to evenly partition the multi-modal data and then distribute the data partitions to worker nodes. 
Fig.~\ref{fig:two-layer-index} illustrates the global index structure built on the multi-metric space dataset in Fig.~\ref{fig:multi_metric_example}. First, we select a pivot $p$ for each metric space $M_i$ to map the metric space to one-dimensional vector space, where each object $o^i$ is denoted via its distance $\delta_i(p^i, o^i)$ to pivot $p^i$. In this example, we chose $o_2$ as the pivot for $M_1$, $o_4$ for $M_2$, and $o_7$ for $M3$. The tree mapped one-dimensional vector spaces are illustrated in Fig.~\ref{fig:two-layer-index}(a), resulting in a three-dimensional vector space shown in Fig.~\ref{fig:two-layer-index}(b). In the sequel, an $R^*$-tree is built on the data in the mapped three-dimensional vector space, as shown in Fig.~\ref{fig:two-layer-index}(c). 
Each entry $e$ in the non-leaf node of $R^*$-tree stores (i) a pointer $e.ptr$ to the root node of its sub-tree and (ii) the MBR $MBR(e)$ to contain all the data points in its sub-tree. For example, non-leaf entry $e_2$ in the root node stores the pointer to the node $N_3$ and the blue $mbr_3$ that contains all the data points (i.e., $o_4$, $o_5$, $o_6$, and $o_7$). Each entry in the leaf node of $R^*$-tree stores (i) the object identifier and (ii) the mapped vector $[\delta_1(p_1^1,o^1),\delta_2(p_2^2,o^2),\ldots,\delta_m(p_m^m,o^m)]$ of this object.  Additionally, each leaf node corresponds to a partition.

$R^*$-tree is a balanced tree, which recursively partitions the data to minimize the overlap between partitions, and thus, compact partitions are obtained (i.e., objects in each partition are similar to each other). 
The objects in each leaf node form a data partition. Note that we adopt the pivot selection method FFT (farthest-first-traversal)\cite{mao2012pivot}, as it is simple and effective. Additionally, we use only one pivot per metric space to prevent high dimensionality at the global level, thereby achieving better partitioning quality and enhanced search performance.

Algorithm~\ref{algo:rrtreeconstruction} presents the detailed pseudo-code for constructing the global index. For each data object $o^M \in D$ in multi-metric space (line 1), it computes its distance $\delta_i(o^i, p_i^i)$ to the corresponding pivot $p_i$ in each metric space $M_i$ to form a mapped vector $v_o$ (lines 2--4), and then inserts the computed vector $v_o$ into the $R^*$-tree (line 5). Finally, the constructed $R^*$-tree is returned (line 6). 

\noindent\textbf{Discussions.} $RR^*$-Tree construction ignores the weights of metric spaces. This is because different users may have various weights; if we partition the multi-metric space dataset considering weights, it is difficult to obtain satisfactory partitioning for every user. In our global index, we assume that each metric space has the same weight, so we consider all metric spaces equally in order to support various weights raised by different users. By doing this, we can cluster similar objects in the same partition, enabling good pruning ability in a general case.


\begin{algorithm}[t]
\small
\SetNlSty{small}{}{:}
\SetCommentSty{footnotesize}
\LinesNumbered
\DontPrintSemicolon
\caption{Local Index Construction}
\label{algo:localindexconstruction}
\KwIn{Dataset $S$ with $N$ points, the number $m$ of metric spacers}
\KwOut{Indexes $\{I_1, I_2, \dots, I_m\}$ for each metric space}
    \ForEach{$o^M_j \in S, \, j \in \{1, 2, \dots, N\}$}
    {
        $\{o_j^1, o_j^2, \dots, o_j^M\} \gets \text{Decompose}(o_j^S)$ \tcp{Decompose into separated metric spaces}
    }
    \ForEach{$i \in \{1, 2, \dots, m\}$}
    {
        $C_i \gets \{o_1^i, o_2^i, \dots, o_N^i\}$ \tcp{Dataset in metric space $M_i$}
        
        \If{modality of $M_i$ is text}
        {
            $I_i \gets \text{InvertedIndex}(C_i)$ \tcp{Inverted index for text}
        }
        \Else
        {
            
            $\mu \gets \text{Mean}(M_i)$, $\sigma^2 \gets \text{Variance}(M_i)$, $d_i^{\text{hidden}} \gets \frac{\mu^2}{2\sigma^2}$ 
            
            

            \If{$d_i^{\text{hidden}} > 5$}
            {
                $I_i \gets \text{MVP-tree}(C_i)$ \tcp{MVP-tree for high $d_i^{\text{hidden}}$}
            }
            \Else
            {
                $I_i \gets \text{R-tree}(C_i)$ \tcp{R-tree for low $d_i^{\text{hidden}}$}
            }
        }
        
        Construct $I_i$ for $C_i$\;
    }
    \KwRet $\{I_1, I_2, \dots, I_m\}$\;
    
\end{algorithm}

\noindent\textbf{Local Index.}
 Next, the multi-modal data objects partitioned by the global index are assigned to the corresponding worker nodes. Note that a worker node may be allocated multiple partitions, resulting in the construction of several relatively independent local indexes for the same worker. For the local index, we construct separated indexes for each metric space, enabling better pruning on each metric space by utilizing the specific characteristics of each data type. We determine the specific index type based on the hidden dimension of each modality. The hidden dimension of the data is computed as \(d^{\text{hidden}} = \frac{\mu^2}{2\sigma^2}\)~\cite{chen2022indexing}, where \(\mu\) represents the mean distance between data objects, and \(\sigma^2\) denotes the distance variance. In our system, we use R-tree (designed for low-dimensional vectors, $d^{\text{hidden}} \le 5$), {MVP-tree} (designed for high-dimensional vectors, $d^{\text{hidden}} \textgreater 5$), and inverted index (designed for text). However, our system is flexible to support various indexes. 
 

 {The detailed local index construction is presented in Algorithm \ref{algo:localindexconstruction}. Initially, the algorithm takes as input a dataset $S$ with $N$ data points and the number $m$ of metric spaces. It first literately decomposes each data point into $m$ separate components (lines 1--2). Subsequently, for each metric space $M_i$ $(1\le i \le m)$, the algorithm constructs a dataset $C_i$ forming by $o_j^i$ $(1\le j \le N)$ (line 4). If the modality of $C_i$ is textual, an inverted index is constructed for that modality (lines 5--6).
Otherwise, the algorithm computes the hidden dimension of the modality (line 8).
 Based on the hidden dimension, the algorithm determines whether to use an MVP-tree for high-dimensional data  (if the hidden dimension is greater than 5) or an R-tree for lower-dimensional data (lines 9--12). Finally, a local index forest is constructed and returned (lines 13--14).}

\vspace{-2mm}
\subsection{Multi-Metric Range Query}
Based on the constructed two-layer index, we present pruning lemmas and the corresponding algorithm for multi-metric range query.
For a multi-metric range query $\text{MMRQ}(q^M, W, r)$, the query region is first mapped to an MBR at the global level. The mapped region $R_q = [\delta_i(q^i, p_i^i)-r, \delta_i(q^i, p_i^i)+r | 1 \le i \le m]$. Here, $p_i$ are pivots selected to construct the global $RR^*$-tree as stated in Section~\ref{subsec:index}. The pruning lemma in the global index can be developed.

\begin{lemma}
\vspace{-2mm}
\label{lemmagloabl}
Given a non-leaf entry $e$ in the global $RR^*$-tree, for all $i$ $(w_i \neq 0)$, $R_q[i]$ does not intersect with $e.MBR[i]$ (i.e., $R_q[i].min$ $\textgreater$ $e.MBR[i].max$ or $R_q[i].max$ $\textless$ $e.MBR[i].min$), then $e$ can be pruned. Here, $R_q[i]$ = $ [\delta_i(q^i, p_i^i)-r, \delta_i(q^i, p_i^i)+r]$.
\end{lemma}

\begin{proof}
\vspace{-2mm}
For each data point $o$ in $e$, $e.MBR[i].min \le \delta_i(o^i, p_i^i) \le e.MBR[i].max$, where $e.MBR[i].min$ and $e.MBR[i].max$ denotes the lower bound and upper bound of $e.MBR$ in $i$-the dimension respectively. For a metric space $M_i$ with $(w_i \neq 0)$, if $R_q[i]$ does not intersect with $e.MBR[i]$, we can get (1) $R_q[i].min$ $\textgreater$ $e.MBR[i].max$ or (2) $R_q[i].max$ $\textless$ $e.MBR[i].min$. 

(1) If $R_q[i].min$ $\textgreater$ $e.MBR[i].max$, then $\delta_i(q^i, p_i^i)-r > e.MBR[i]$ $.max\ge \delta_i(o^i, p_i^i)$. According to the triangle-inequality, $\delta_i(o^i, q^i) \ge \delta_i(q^i,$ $ p_i^i)$ $- \delta_i(o^i, p_i^i) > r$. Thus,  $\delta(o, q) \ge \delta_i(o^i, q^i) > r$, implying $o$ cannot be the result.

(2) If  $R_q[i].max$ $\textless$ $e.MBR[i].min$, then $\delta_i(q^i, p_i^i)+r< e.MBR[i]$ $.min\le \delta_i(o^i, p_i^i)$. 
According to the triangle-inequality, $\delta_i(o^i, q^i) \ge \delta_i(q^i,$ $ p_i^i)$ $- \delta_i(o^i, p_i^i) > r$.  Thus,  $\delta(o, q) \ge \delta_i(o^i, q^i) > r$, implying $o$ cannot be the result.
\vspace{-3mm}
\end{proof}


At the global level, a combined index is constructed assuming that the weights of all metric spaces are the same; thus, we cannot use the weight vector $W$ for fine-grained pruning. In the following, we present a pruning lemma on operated indexes at the local level.

\begin{lemma}
\vspace{-2mm}
\label{lemmaweight}

Let \( \Delta \) represent the set of distance metrics. An object \( o^M \) is included in the result set of \( \text{MMRQ}(q^M, W, r) \) if there exists at least one metric \( \delta_i \in \Delta \) such that \( \omega_i > 0 \) and \( \delta_i(q^M, o^M) \leq \frac{r}{\sum_{\delta_i \in \Delta} \omega_i} \).
\vspace{-3mm}
\end{lemma}

The above lemma can be derived by the pigeon theorem~\cite{zhu2022desire}; thus, it is omitted. Lemma~\ref{lemmaweight} can also be easily extended to non-leaf entries in each local level index.
{Based on the above two lemmas, we design a multi-metric range query method tailored for distributed environments. First, we apply Lemma~\ref{lemmagloabl} to perform partition filtering at the global level, selecting qualified candidate partitions and reducing query costs. Next, we assign the query tasks to the worker nodes corresponding to the candidate partitions and use Lemma~\ref{lemmaweight} to find candidate data points in each metric space efficiently. During the verification phase, the candidate set is validated using the multi-metric distance, finding the final results.

\vspace{-2mm}
\subsection{Multi-Metric \textit{k}NN Query}
Directly aggregating \(\textit{k}\)NN answers from each queried metric space does not guarantee correctness. To address this, DESIRE~\cite{zhu2022desire} processes MM\(\textit{k}\)NN queries by first performing \(\textit{k}\)NN searches in one or more metric spaces to establish a distance boundary, which is then used to conduct range queries across all the queried spaces to obtain the final results. However, this approach struggles to balance efficiency and pruning capability, as leveraging more information from the queried spaces for pruning comes at the expense of increased search costs. To overcome this limitation, we propose a two-phase MM\(\textit{k}\)NN strategy. In the first phase, we determine an upper bound \(\bar{dis_k}\) for the \(k\)-th nearest neighbor across multiple metrics by searching the partition most similar to the query. In the second phase, we convert the MM\(\textit{k}\)NN query into a range query, using \(\bar{dis_k}\) as the radius for \text{MMRQ}.

Specifically, the proposed strategy begins by identifying the partition whose center is most similar to \(q^M\) in the global index, serving as a sample of the multi-modal objects. Next, the corresponding worker nodes search for \(k\) candidate results within the local index of each modality in that partition. After reassembling and verifying the modalities, an approximate \(k\)-th nearest neighbor, \(\bar{o^M_{k}}\), is obtained. 
The multi-metric distance between \(\bar{o^M_{k}}\) and the query point \(q^M\) is then calculated to derive the upper bound of the \(k\)-nearest neighbor distance, \(\bar{dis_k}\). This upper bound is subsequently used as the search radius for the multi-metric range query, ensuring the correctness of the \(k\)-th nearest neighbors by accurately pruning irrelevant candidates and refining the search process to determine the exact \(k\)-th nearest neighbors of the given query.

Notably, the sampling scale influences both the efficiency of the \(\textit{k}\)NN computation and the accuracy of the estimated upper bound \(\bar{dis_k}\), while the precision of \(\bar{dis_k}\) affects the number of distance calculations required during the verification phase. Therefore, the query cost is mainly composed of three parts: (1) the size of partitions in worker nodes, which determines the sampling scale, (2) the quality of the sampled top \(k\) candidate results, and (3) the multi-modal range query efficiency using the distance upper bound \(\bar{dis_k}\)

\vspace{-2mm}
\section{End-to-End Tuning}
\label{sec:tuning}
In real-world distributed environments, communication overhead and computational capacity between nodes are not only affected by network latency but also influenced by other complex factors, such as system workloads, bandwidth fluctuations, dynamic allocation of hardware resources, and even changes in network topology. These complex factors make it extremely difficult to establish a cost model for distributed similarity search systems. Additionally, the multi-modal data in our problem makes the cost model development much harder.
As a result, mathematical cost models struggle to effectively analyze the costs of distributed multi-modal data search systems. Motivated by it, we introduce an end-to-end reinforcement learning (RL) based tuning model.
\vspace{-2mm}
\subsection{RL for Tuning}

We incorporate reinforcement learning  (RL), which enables the optimization of continuous parameter spaces through trial-and-error learning. With reinforcement learning's feedback mechanism, the system can learn in an end-to-end manner, which is easy to use. 
{Compared to alternative methods such as grid search or supervised learning, RL is particularly suitable for our setting. First, no labeled data is available for learning optimal configurations. Second, the high-dimensional and discrete tuning space makes exhaustive search infeasible. RL enables the system to learn effective strategies through limited interactions, and has proven successful in database tuning scenarios~\cite{zhang2021cdbtune+}.}

During the tuning process, the RL agent serves as the optimization component, interacting with the environment (i.e., our search system, OneDB). The agent observes the current state of the OneDB, represented as a vector of internal metrics, and adjusts specific tuning knobs (actions) to maximize the reward. 
Specifically, at time \( t \), the state of the agent consists of the current values of all tuning knobs. After applying the recommended configuration, the agent evaluates the system's performance through multiple queries (e.g., \( k \)-nearest neighbor and range queries), with execution time as a key reference for the system's state.

The RL based tuning is iterative. In each iteration, 
the agent compares the system's performance with the previous state (the initial state or the state in the previous iteration) to calculate the reward. The agent continuously adjusts the knob configurations according to the rewards, eventually converging to an optimized configuration. The final policy is implemented through a deep neural network, which maps the current system state to the optimal tuning knob configuration, ensuring the highest possible reward.

As the tuning in the OneDB involves high-dimensional states and continuous actions, Deep Deterministic Policy Gradient (DDPG), a policy-based deep reinforcement learning algorithm, is an ideal choice for addressing such complex scenarios. DDPG combines features of Deep Q-Networks (DQN) and actor-critic algorithms. It is able to learn continuous action values directly from the current state without the need to compute and store Q-values for all discrete actions as in DQN.

The DDPG architecture consists of two main components: the actor and critic networks. The actor network generates actions (i.e., adjustments to the system parameters), while the critic network evaluates the actions produced by the actor and provides feedback to guide the actor's learning. By sampling from the replay buffer, which stores state, action, reward, and next-state tuples, the critic network updates its parameters using the Q-learning algorithm, while the actor network is updated via policy gradients.

\vspace{-4mm}
\subsection{Reward Function}

The reward function in OneDB's tuning is designed to capture the performance improvements resulting from each tuning adjustment. It measures performance changes based on two metrics: the difference from the initial configuration \( \Delta Q^{t}_{0} \) and the difference from the previous time step \( \Delta Q^{t}_{t-1} \). 
Following the principles of reinforcement learning, we have designed the following reward function:
\begin{equation}
r = \text{sign}(\Delta Q^{t}_{0}) \cdot \left( (1 + |\Delta Q^{t}_{0}|)^2 - 1 \right) \cdot \left| 1 + \text{sign}(\Delta Q^{t}_{0}) \cdot \Delta Q^{t}_{t-1} \right|
\end{equation}

The sign function, denoted as $\text{sign}(x)$, returns 1 when $x > 0$, -1 when $x < 0$, or 0 when $x = 0$.
We offer three strategies for tuning the reward function's learning speed.

(i) An exponentially weighted reward function amplifies the reward for significant performance improvements using an exponential function, magnifying larger performance gains. 
\begin{equation}
    r = \text{sign}(\Delta Q_{t \to 0}) \cdot \left( e^{|\Delta Q_{t \to 0}|} - 1 \right) \cdot \left| e^{\text{sign}(\Delta Q_{t \to 0}) \cdot \Delta Q_{t \to t-1}} \right|
\end{equation}

(ii) A logarithmically weighted reward function smooths penalties for performance degradation, thereby reducing the impact of minor fluctuations significantly. 
\begin{equation}
\begin{split}
r = \text{sign}(\Delta Q_{t \to 0}) \cdot \big( (1 + |\Delta Q_{t \to 0}|)^2 - 1 \big) \\
\cdot \left| 1 + \text{sign}(\Delta Q_{t \to 0}) \cdot \Delta Q_{t \to t-1} \right|
\end{split}
\end{equation}
\vspace{-0.5mm}
(iii) Lastly, a reward function with an additional penalty for performance drops introduces a penalty term \( \lambda \), enforcing stricter penalties for performance decreases to ensure that the agent prioritizes stability and avoids configurations that lead to significant performance degradation. 
\begin{equation}
r = - \lambda \cdot \max\left(0, -\text{sign}(\Delta Q_{t \to 0}) \cdot \Delta Q_{t \to t-1}\right)
\end{equation}
These methods employ different mathematical weighting mechanisms. They either emphasize amplifying performance improvements or impose stricter penalties for performance declines, promoting the agent to select configurations that stabilize and significantly enhance performance during tuning.

\vspace{-4mm}
\section{Experiments}
\label{sec:experiments}

In this section, we conduct extensive experiments to evaluate the performance of our proposed distributed multi-modal data similarity search framework, OneDB. 
{OneDB supports exact similarity search across all modalities. Once the modality weight vector $W$ is fixed, the similarity distance $\delta_W(\cdot,\cdot)$ is well-defined, and the returned result set is fully deterministic and complete under this metric. Thus, all reported results in the experiments are based on fully deterministic query evaluation.} The effectiveness of the Dual-Layer Indexing Strategy is demonstrated through experiments on index construction efficiency, similarity search performance, and scalability. Furthermore, the reliability of two independent modules—the weight learning component and the RL based tuning component—is confirmed through the ablation studies.

\vspace{-2mm}
\subsection{Experimental Settings}
\textbf{Datasets.} In our experiments, we use three real-world datasets: (i) \textit{Air}, that contains hourly measurements of key air pollutants collected from various monitoring stations across multiple cities in India, such as $\mathrm{PM2.5}$, $\mathrm{NO}$, $\mathrm{NO_2}$, $\mathrm{CO}$, $\mathrm{SO_2}$, $\mathrm{O_3}$, Benzene, Toluene, and Xylene ; (ii) \textit{Food}, which provides detailed information on food products, including the number of additives, nutritional facts (such as salt, energy, fat, proteins, and sugars), the main category, a set of descriptive labels, and the appearance of the products; and (iii) \textit{Rental}, that contains information on apartments in New York, including price, the number of bedrooms and bathrooms, location, publish date, and a brief review. Additionally, we generated (iv) \textit{Synthetic}, which includes randomly assigned geographic coordinates, textual data sourced from the Moby\footnote{http://icon.shef.ac.uk/Moby/}, figures gathered from Flickr\footnote{http://cophir.isti.cnr.it/}, along with multiple features consisting of randomly generated values. {Notably, we generate \textit{Synthetic II}, extended from \textit{Synthetic} dataset, includes million-scale features from the GoogLeNet model. \textit{Synthesis II} dataset cannot be evaluated on a single machine. It is generated to evaluate the impact of varying the number of workers on query performance.}

In our experiments, we apply the following metric distance functions: (i) $L_1$-norm distance is utilized for key attributes such as air pollutants, date, price, number of bedrooms and bathrooms, number of additives, nutritional facts, and vector representations for image data, including food product images and Flickr photos; (ii) $L_2$-norm distance is used for geographic coordinates; (iii) edit distance is applied to the main category of food products and words from the Moby dataset. To ensure comparability across metrics, we follow the prior study~\cite{zhu2022desire} and normalize each distance by dividing them by twice the median of all observed distances. Table~\ref{tab:datasets} outlines the dataset statistics, with cardinality denoted as ``Card." and the number of metric spaces represented by ``m". 

\noindent\textbf{Parameters and Performance Metrics.} We evaluate the performance of the multi-metric similarity search framework by varying parameters, including the \( k \) in multi-metric \( k \)-nearest neighbor (MM\textit{k}NN) queries, the search radius \( r \) in multi-metric range queries (MMRQ), the number of worker nodes and the cardinality (the proportion of the dataset used in the search relative to the entire data). Here, we control the proportion of result objects returned by MMRQ by adjusting the search radius \( r \), ensuring that \( r \) represents a consistent search scope across different datasets. The key performance metrics include index construction time, query execution time, and throughput. For OneDB, the index construction time consists of two components: the time to build the global index (which manages distributed partitioning) and the local index (which facilitates efficient filtering at worker nodes). Additionally, to mitigate the impact of network fluctuations in the distributed environment, we measure based on the average value over 100 queries.

\noindent\textbf{Baselines.} To validate the performance of the multi-metric similarity search, we compare the OneDB framework with three state-of-the-art systems: (1) the single-machine multi-metric similarity search framework DESIRE, (2) the distributed single-metric similarity search framework {\sf DIMS}, and (3) vector database system {\sf Milvus}. To ensure fairness, {\sf DESIRE} and {\sf DIMS} are integrated into the same Spark platform as OneDB, resulting in {\sf DIMS-M} and {\sf DESIR-D}. To further verify the effectiveness of OneDB's index selection, we also include two baselines derived from OneDB: (1) OneDB-R2M, which replaces the R-tree with an M-tree, and (2) OneDB-MVP2M, which replaces the MVP-tree with an M-tree. Experiments run on a 11 nodes cluster, with each machine configured with Ubuntu 14.04.3, 128GB of RAM, and equipped with two 12-core E5-2620 v3 processors running at 2.40 GHz. The system runs on Hadoop 2.6.5 and Spark 2.2.0.

 \vspace {-2mm}
\subsection{Construction Cost and Update Performance}



\begin{table}
 \vspace{-1cm}
\caption{Statistics of the datasets used}
\label{tab:datasets}
\small 
\vspace{-0.1cm}
\setlength{\tabcolsep}{4.2pt} 
\begin{tabular}{|p{2.0cm}<{\centering} |p{1.35cm}<{\centering} |p{1.1cm}<{\centering} |p{1.25cm}<{\centering} |p{1cm}<{\centering}|p{1.1cm}<{\centering}|p{1.35cm}<{\centering}|} \hline 
 \multicolumn{2}{|c|}{\textbf{Datasets}} & \textbf{Rental} & \textbf{Air} & \textbf{Food} & \textbf{Syn.} & \textbf{Syn. II} \\ \hline 
 \multicolumn{2}{|c|}{Card.}  & 113,176 & 1,150,000 & 38,757 & 200,000 & 10,000,000 \\ \hline 
 \multicolumn{2}{|c|}{m}  & 5 & 13 & 9 & 50 & 96 \\ \hline 
\end{tabular}
\vspace{-0.4cm}
\end{table}

\begin{table}[tb]
\caption{Construction Costs and Storage Sizes}
\vspace{-0.45cm}
\small{\textit{*Note: Syn. denotes Synthesis.}}
\label{tab:construction_cost}
\centering
\footnotesize 
\setlength{\tabcolsep}{5pt} 
\begin{tabular}{|c|c|c|c|c|c|c|}
\hline
\multicolumn{2}{|c|}{\textbf{Datasets}} & \textbf{Food} & \textbf{Air} & \textbf{Rental} & \textbf{Syn.} & \textbf{Syn. II} \\ \hline
\multirow{2}{*}{DIMS-M} & (s) & 29.33 & 1209.36 & 52.97 & 682.55 & 2359 \\ \cline{2-7} 
 & (MB) & 180.14 & 1608.24 & 272.36 & 1752.33 & 52309 \\ \hline
\multirow{2}{*}{DESIRE-D} & (s) & \textbf{20.43} & 893.52 & \textbf{32.21} & \textbf{551.52} & 2135 \\ \cline{2-7} 
 & (MB) & 147.73 & 1327.99 & 243.86 & \textbf{1602.14} & \textbf{41250} \\ \hline
\multirow{2}{*}{OneDB} & (s) & 22.77 & \textbf{743.84} & 36.65 & 611.47 & \textbf{2083} \\ \cline{2-7} 
 & (MB) & \textbf{157.2} & \textbf{1397} & \textbf{238} & 1681.92 & 42521 \\ \hline
\end{tabular}
\vspace{-6mm} 
\end{table}

\begin{table}[t]
\small
\vspace{-1cm}
\caption{{Effect of Different Update Ratios}}
\label{tab:update_cost}
\vspace{-0.1cm}
\setlength{\tabcolsep}{0.5mm}{
\begin{tabular}{|p{0.9cm}<{\centering} |p{3.2cm}<{\centering} |p{1cm}<{\centering} |p{1cm}<{\centering} |p{1cm}<{\centering}|p{1.05cm}<{\centering}|} \hline 
 \multicolumn{2}{|c|}{\textbf{Update Ratio}} & \textbf{Food} & \textbf{Air} &\textbf{Rental} &\textbf{Syn.}\\ \hline  
\multirow{2}{*}{\makecell{0.1\%}} & Avg. Update Cost (ms)  & 8.93 & 39.06 & 14.19 & 35.51\\ \cline{2-6}
 & Query Time $\Delta$ (ms) & +0.53 & +4.73 & +0.77 & +4.89 \\ \hline 
 \multirow{2}{*}{\makecell{1.0\%}} & Avg. Update Cost (ms)  & 9.52 & 42.22 & 16.37 & 41.93\\ \cline{2-6}
 & Query Time $\Delta$ (ms) & +0.92 & +6.32 & +0.97 & +6.88 \\ \hline 
\end{tabular}
}
\vspace{-4mm}
\end{table}

To compare the construction costs of OneDB with those of its two competitors, {\sf DIMS-M} and {\sf DESIRE-D}, we utilized construction time and storage capacity as performance evaluation metrics. The results presented in Table~\ref{tab:construction_cost} indicate that {\sf DIMS-M} incurs higher computational costs as it constructs indexes for each modality in global and local indexes. At the same time, OneDB exhibits comparable construction efficiency but less storage cost against {\sf DESIRE-D}. 
This is because OneDB constructs separated indexes for different modalities at the worker nodes, avoiding unnecessary redundant storage, while {\sf DESIRE-D} needs to store additional pre-computed distances. However, considering that OneDB adopts a coarse-grained indexing strategy at the primary node and selects modality-specific indexes at the worker nodes, OneDB efficiently manages large-scale, heterogeneous datasets while maintaining a balanced trade-off between storage and computational costs. These results highlight that OneDB demonstrates superior practicality in scenarios requiring scalable and efficient multi-modal data processing.

To evaluate the update performance of OneDB, we perform 0.1\%--1\% random updates (object deletions and insertions). As shown in Table~\ref{tab:update_cost}, OneDB achieves low update costs, and stable query performance with only minor increases in latency. 


\vspace{-0.2cm}
\subsection{Similarity Search Performance}
\label{subsec:expsimilarity}
\begin{figure}[t]
\vspace{0.1cm}
\begin{center}
\subfigtopskip=-7pt
\subfigcapskip=-3pt
\includegraphics[height=0.3cm]{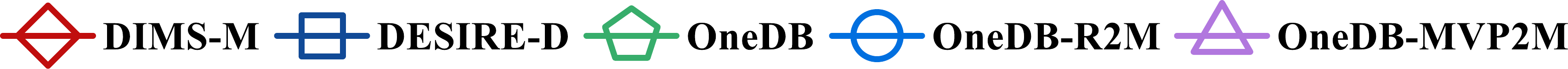}\vspace{0.05cm}
\subfigure[\textit{Air}]{
  \label{fig:rnn-02}
  \includegraphics[width=0.23\textwidth,height=0.14\textwidth]{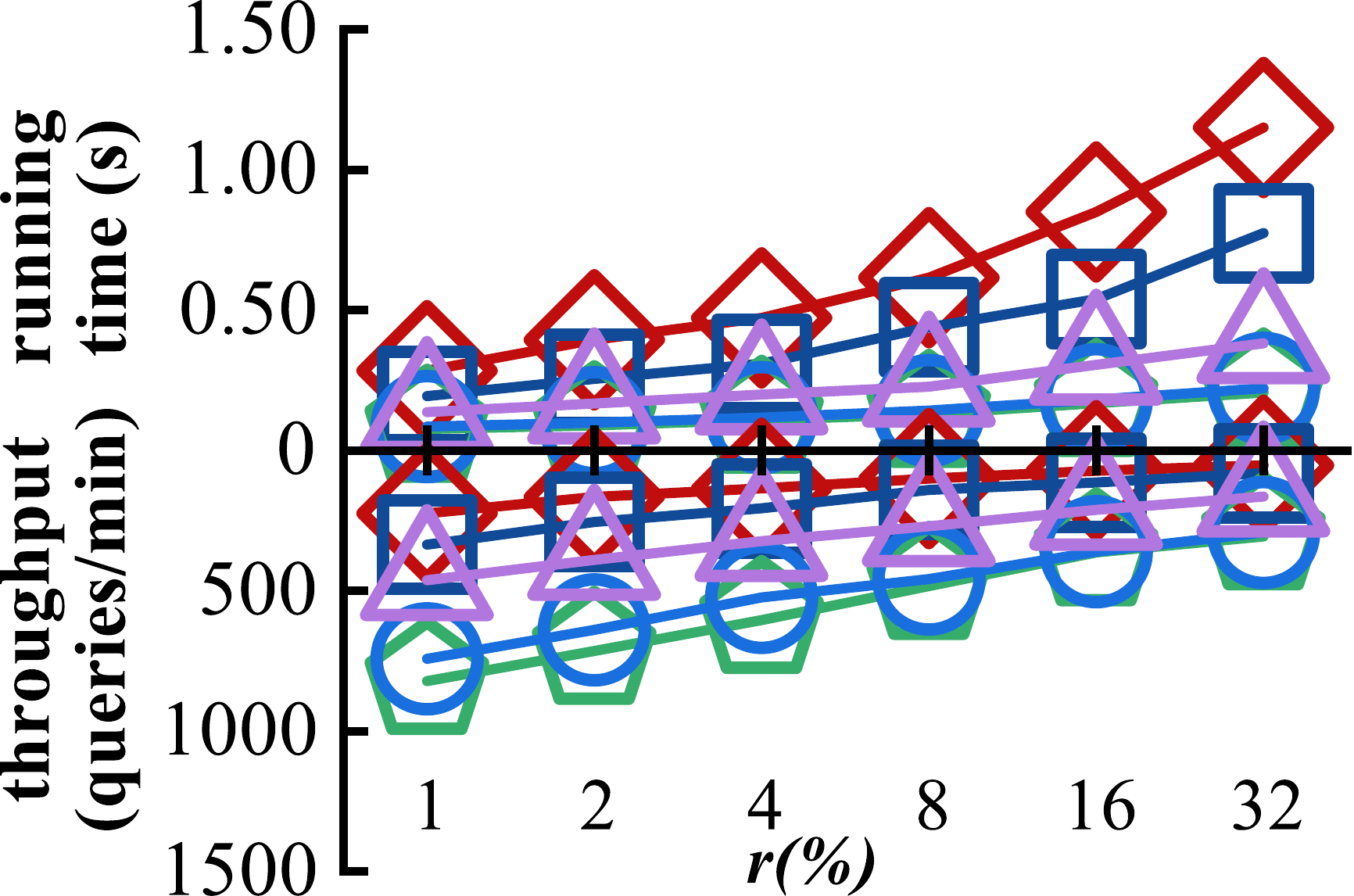}
}
\subfigure[\textit{Food}]{
  \label{fig:rnn-02}
  \includegraphics[width=0.23\textwidth,height=0.14\textwidth]{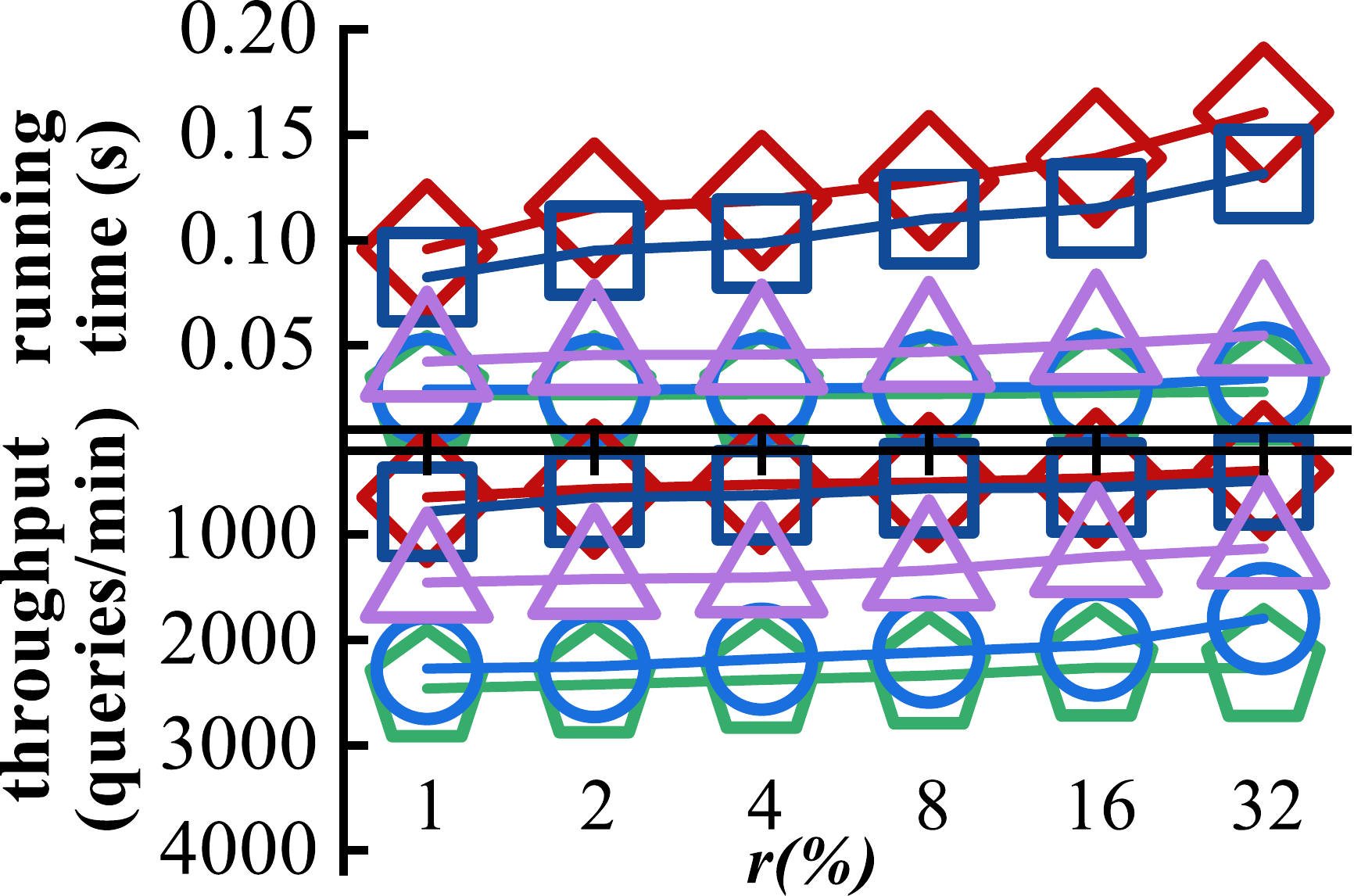}
}

\subfigure[\textit{Rental}]{
  \label{fig:rnn-02}
  \includegraphics[width=0.23\textwidth,height=0.14\textwidth]{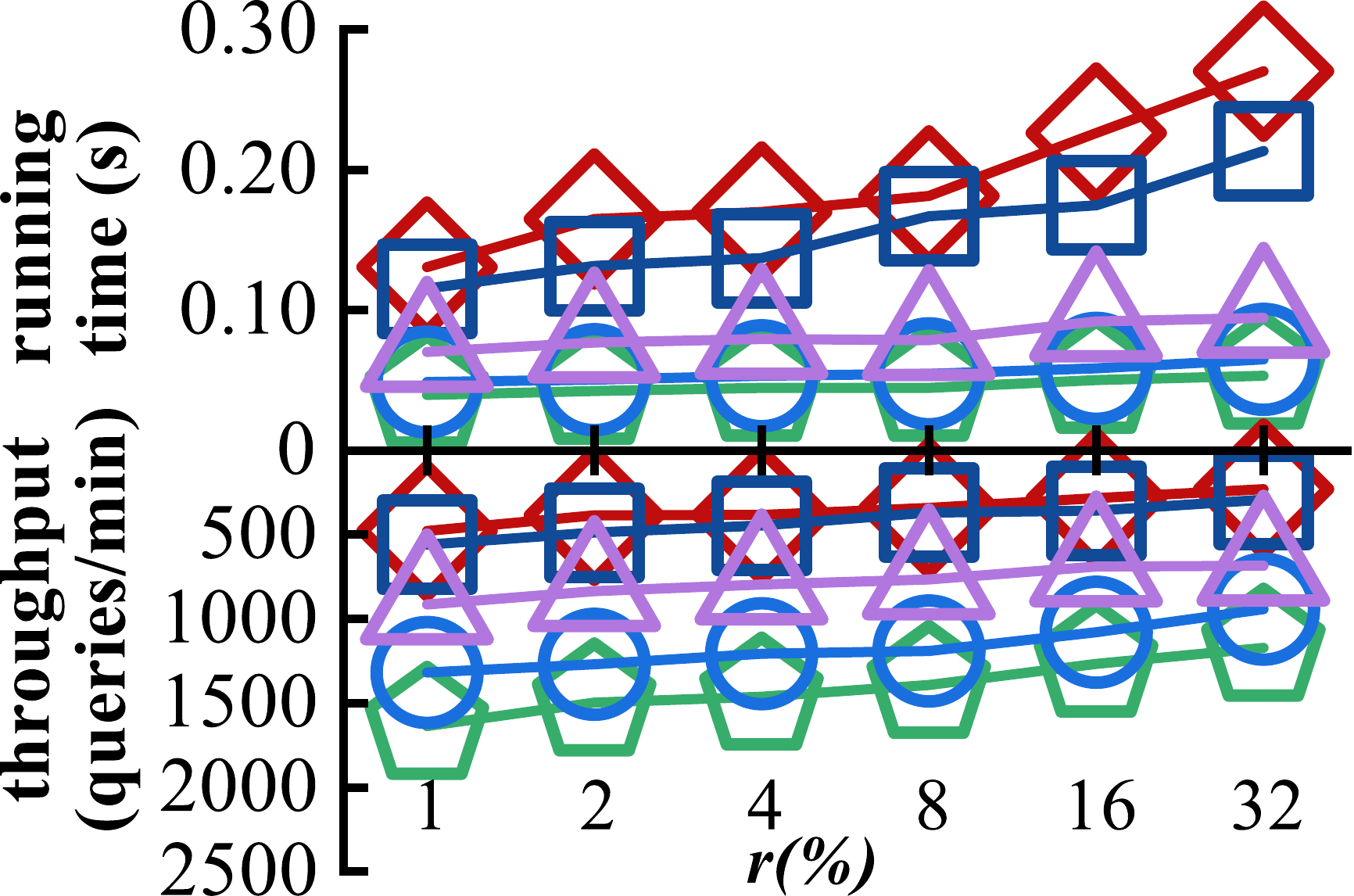}
}
\subfigure[\textit{Synthesis}]{
  \label{fig:rnn-02}
  \includegraphics[width=0.23\textwidth,height=0.14\textwidth]{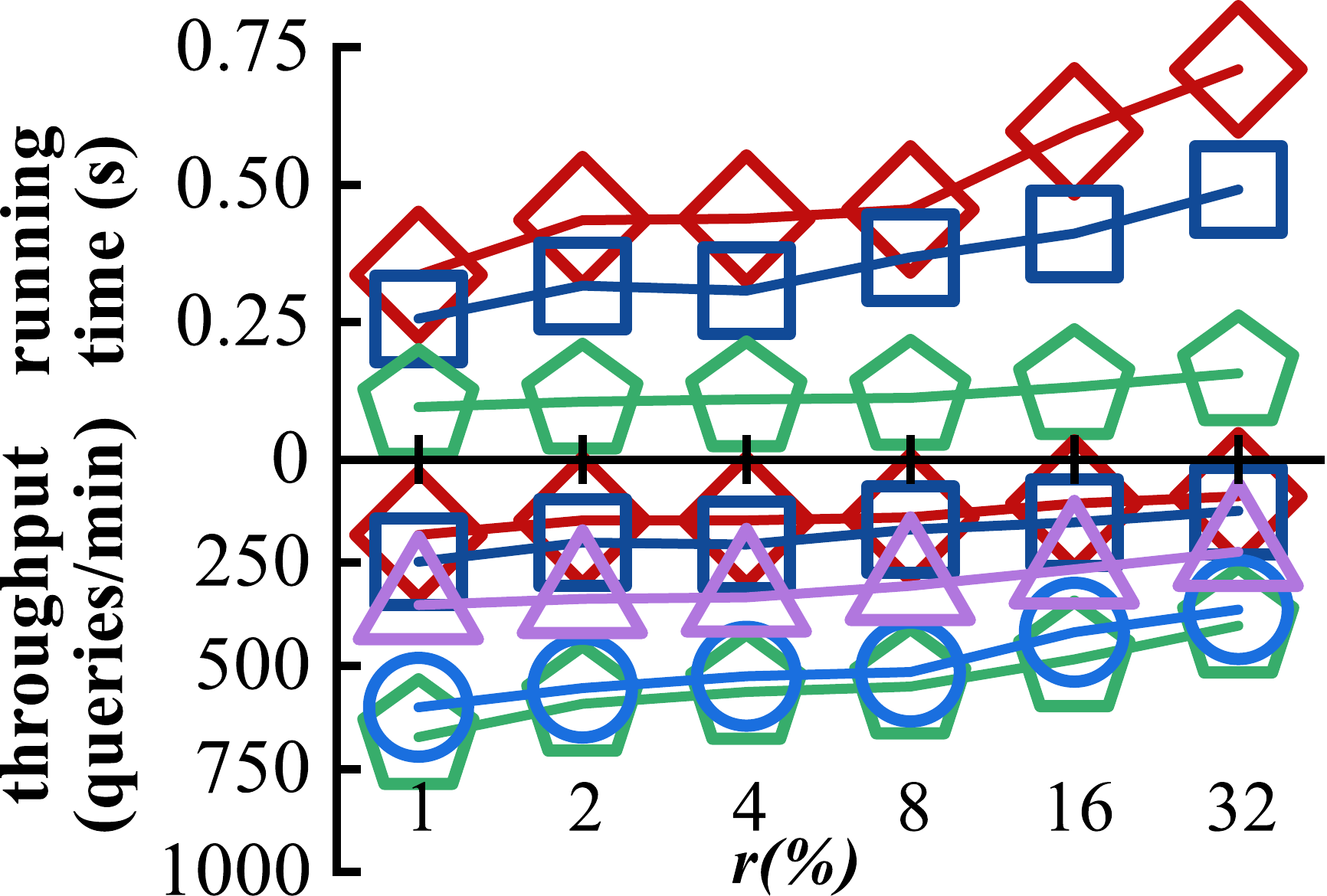}
}
\vspace{-0.25cm}
\caption{{MMRQ performance vs. $r$}}
\vspace{-0.6cm}
\label{fig:rnn}
\end{center}
\end{figure}

\begin{figure}[t]
\vspace{-1cm}
\begin{center}
\subfigtopskip=-7pt
\subfigcapskip=-3pt
\includegraphics[height=0.3cm]{ExpFigs/icon.pdf}\vspace{0.05cm}
\subfigure[\textit{Air}]{
  \label{fig:rnn-02}
  \includegraphics[width=0.23\textwidth,height=0.14\textwidth]{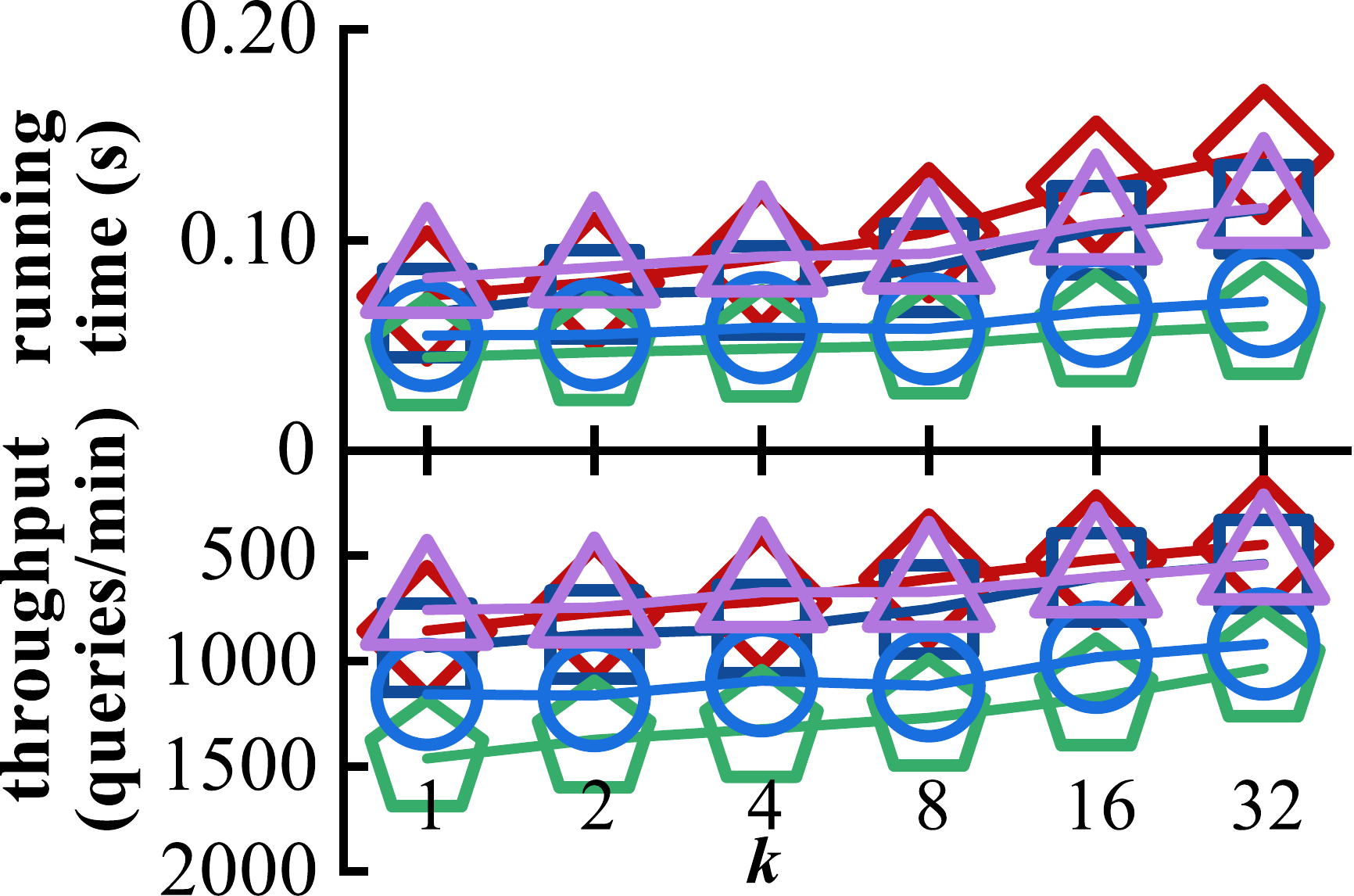}
}
\subfigure[\textit{Food}]{
  \label{fig:rnn-02}
  \includegraphics[width=0.23\textwidth,height=0.14\textwidth]{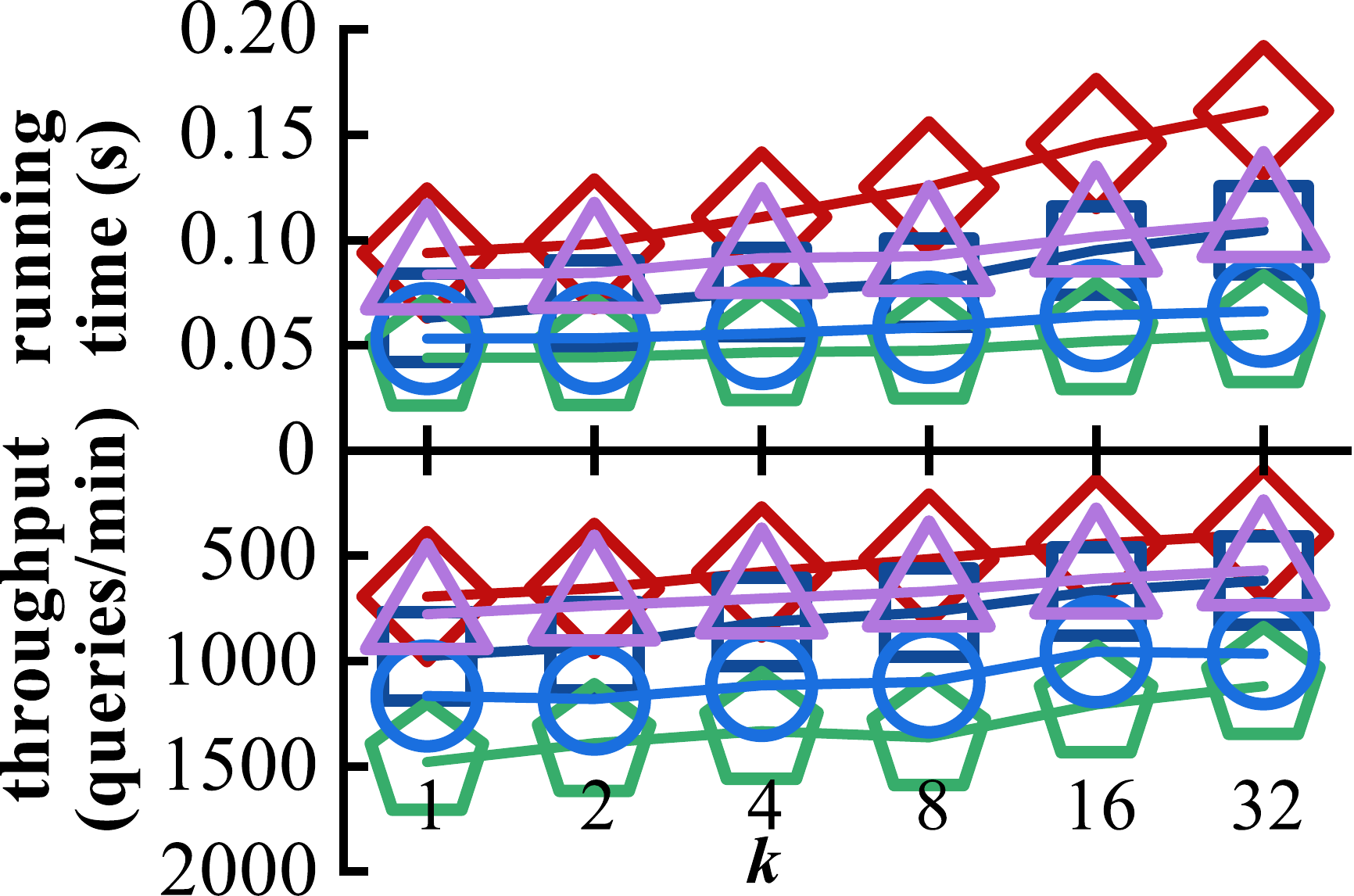}
}

\subfigure[\textit{Rental}]{
  \label{fig:rnn-02}
  \includegraphics[width=0.23\textwidth,height=0.14\textwidth]{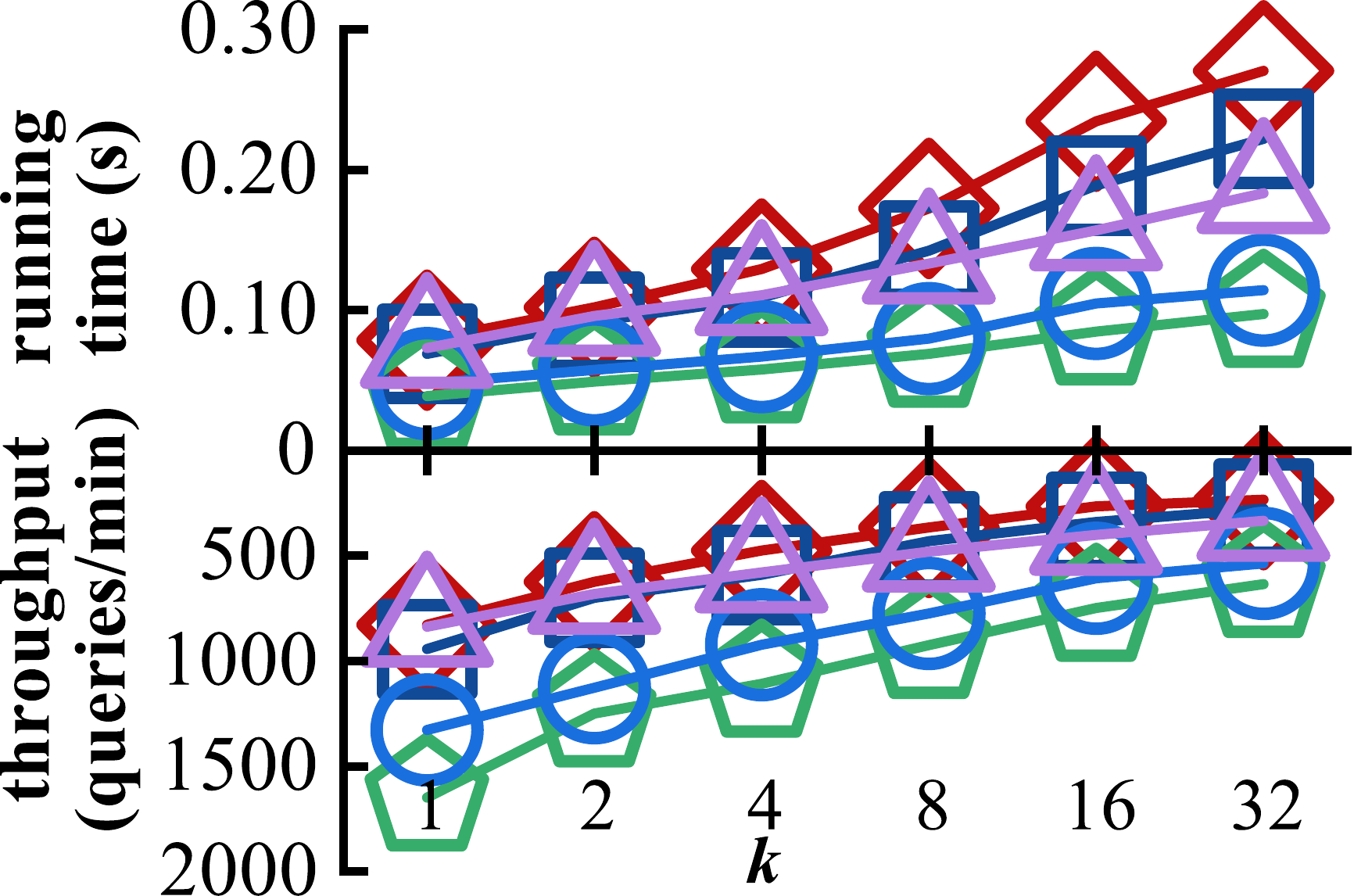}
}
\subfigure[\textit{Synthesis}]{
  \label{fig:rnn-02}
  \includegraphics[width=0.23\textwidth,height=0.14\textwidth]{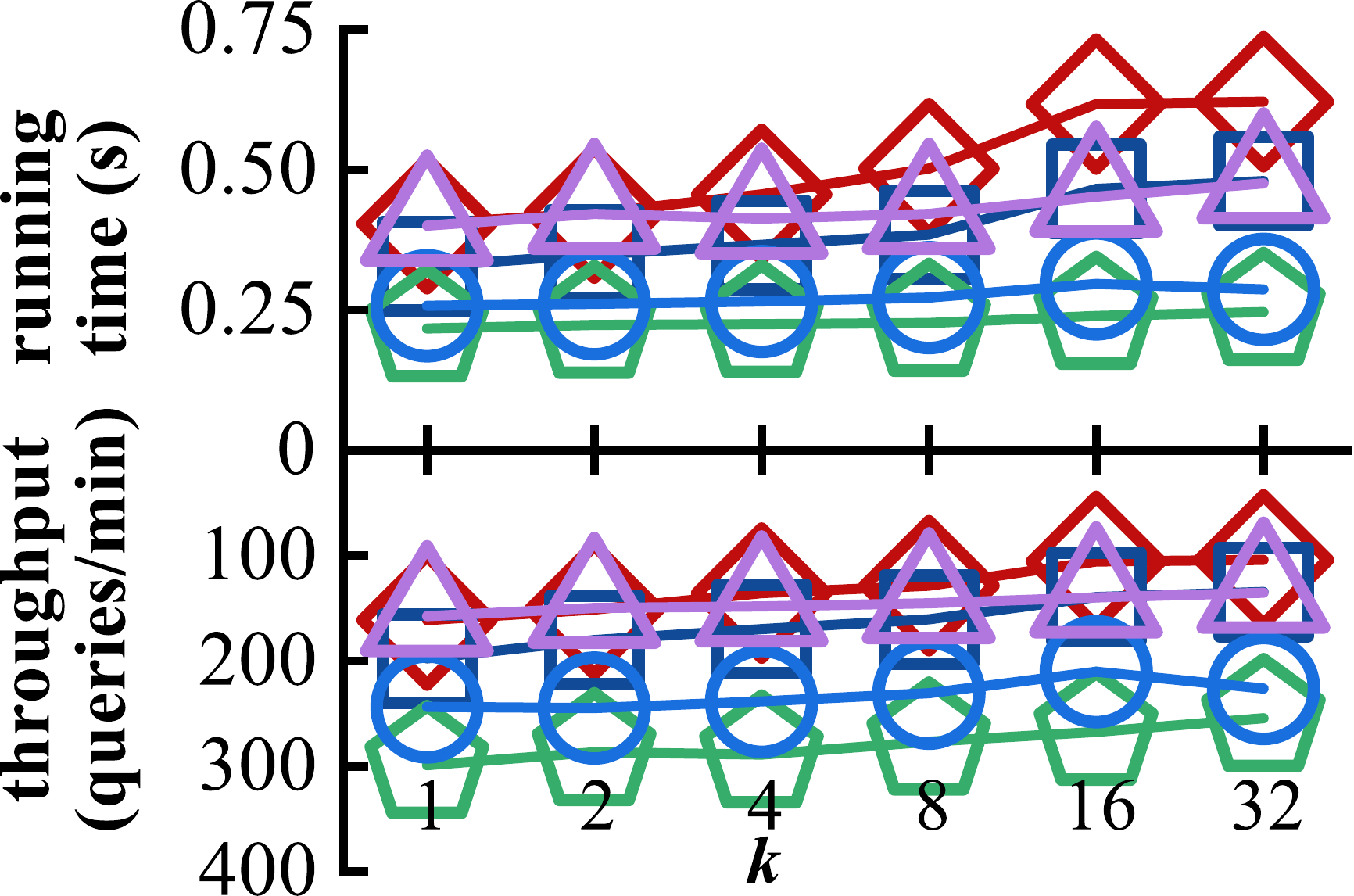}
}
\vspace{-0.35cm}
\caption{{MM$k$NNQ performance vs. $k$}}
\vspace{-0.3cm}
\label{fig:knnn}
\end{center}
\end{figure}

\textbf{Comparison with Metric Space Solutions.} Figs.~\ref{fig:rnn} and~\ref{fig:knnn} depict the MMRQ and MMkNNQ performance by varying $r$ and $k$ respectively. These results demonstrate that OneDB exhibits significant advantages in running time and throughput, particularly on larger datasets with higher values of $r$ and $k$, achieving 2.5 to 5.75 times speedup compared to other baseline methods. This can be attributed to two key factors. First, OneDB uses RR*-tree at the primary node to partition data, achieving balanced workload distribution and efficient resource utilization.
Additionally, by selecting different indexes based on modality type at the worker nodes, OneDB accelerates the similarity search process. To validate the choice of index, we evaluate against OneDB-R2M and OneDB-MVP2M. Experimental results demonstrate that these variants incur higher query latency (up to 25\% increase for OneDB-MVP2M), confirming the superiority of the current design in achieving optimal query performance. This observation is consistent with findings in previous survey~\cite{chen2022indexing}, as MVP-tree is a good chocie for high-dimensional data space in main-memory. Second, when handling MM\textit{k}NNQ, OneDB quickly locates the nearest neighbors of the query objects using global and local indexes and leverages multiple local workers to answer the queries simultaneously, reducing the computational overhead. Moreover, OneDB's indexing mechanism prunes unnecessary regions, narrowing the search space and significantly improving overall performance. Across various $r$ and $k$ values, OneDB consistently demonstrates excellent performance across all datasets, validating the effectiveness and efficiency of its partitioning strategy and indexing structure.

\begin{figure}[t]
\vspace{-0.3cm}
\begin{center}
\subfigtopskip=-7pt
\subfigcapskip=-2pt
\includegraphics[height=0.3cm]{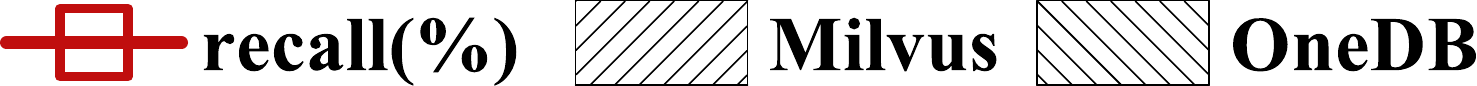}\vspace{0.15cm}
\subfigure[\textit{Air}]{
  \label{fig:rnn-02}
  \includegraphics[width=0.23\textwidth,height=0.14\textwidth]{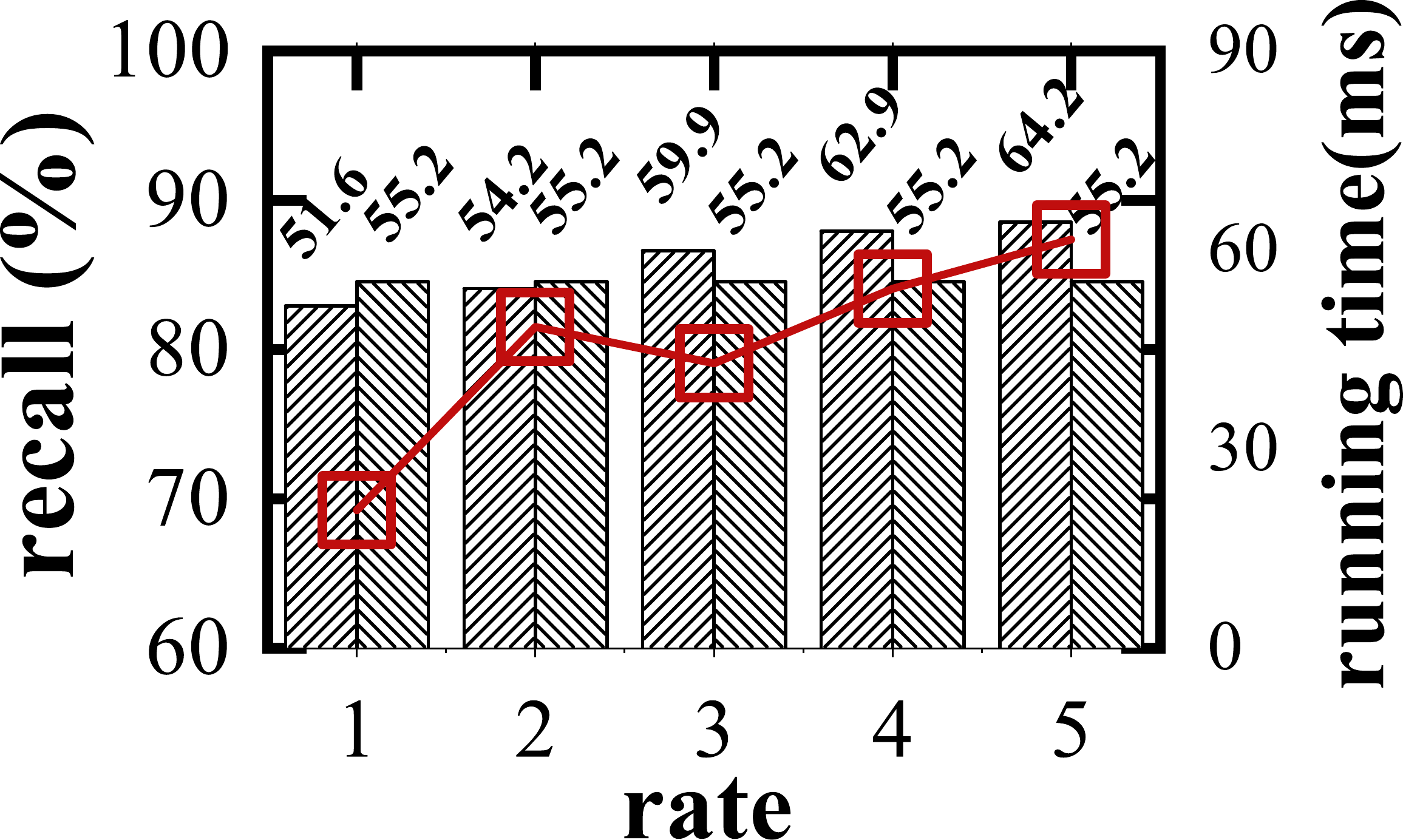}
}
\subfigure[\textit{Food}]{
  \label{fig:rnn-02}
  \includegraphics[width=0.23\textwidth,height=0.14\textwidth]{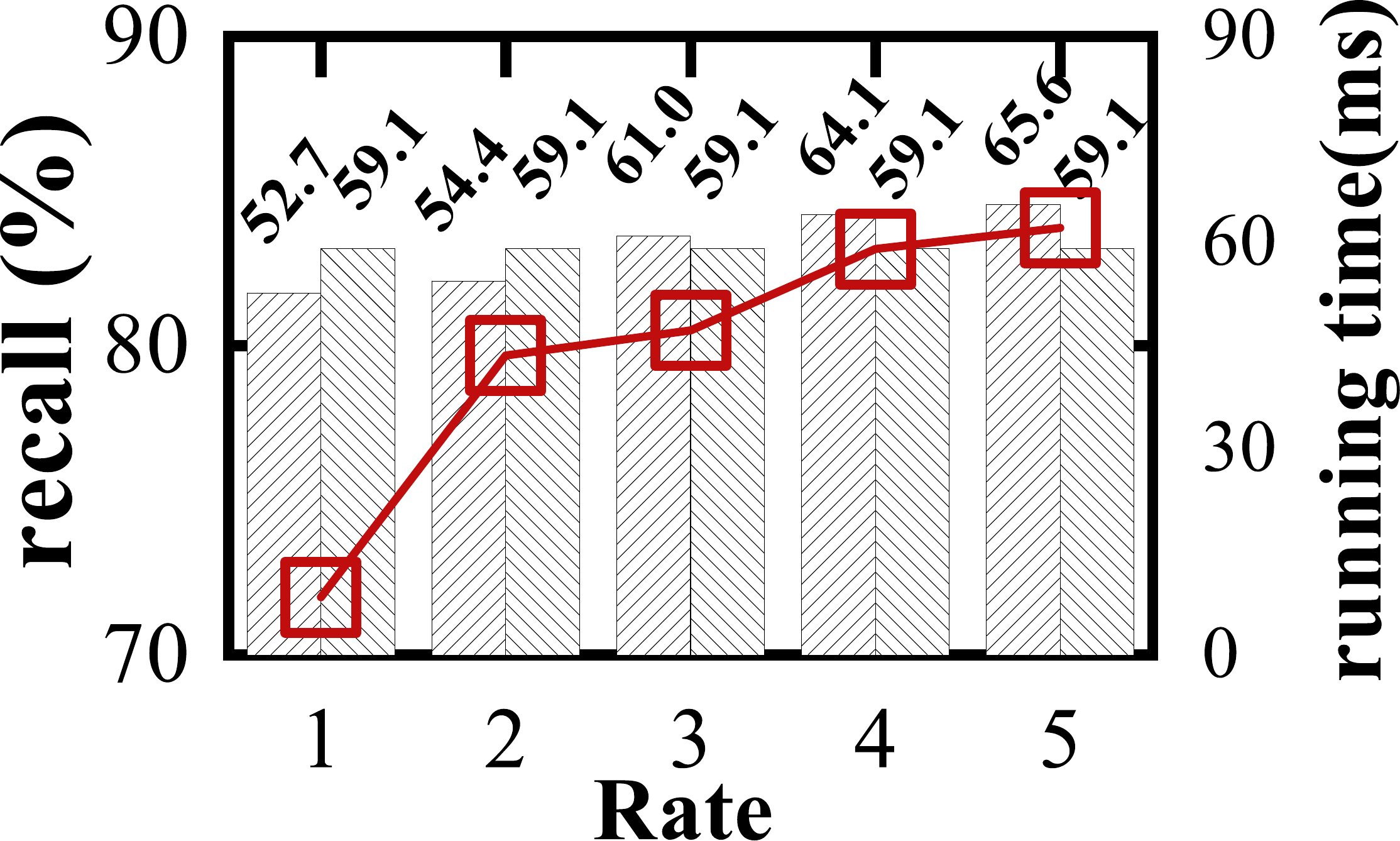}
}

\subfigure[\textit{Rental}]{
  \label{fig:rnn-02}
  \includegraphics[width=0.23\textwidth,height=0.14\textwidth]{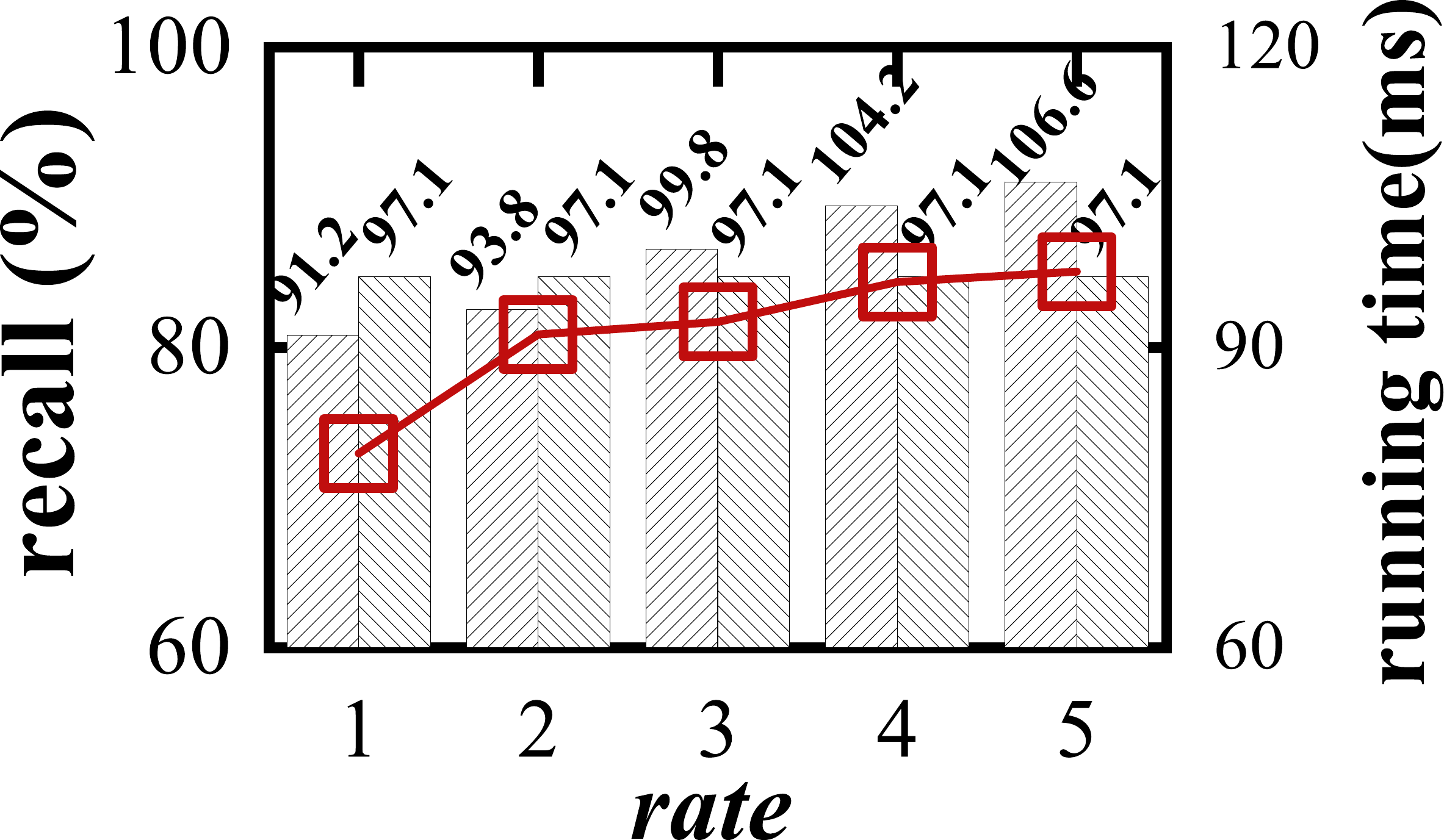}
}
\subfigure[\textit{Synthesis}]{
  \label{fig:rnn-02}
  \includegraphics[width=0.23\textwidth,height=0.14\textwidth]{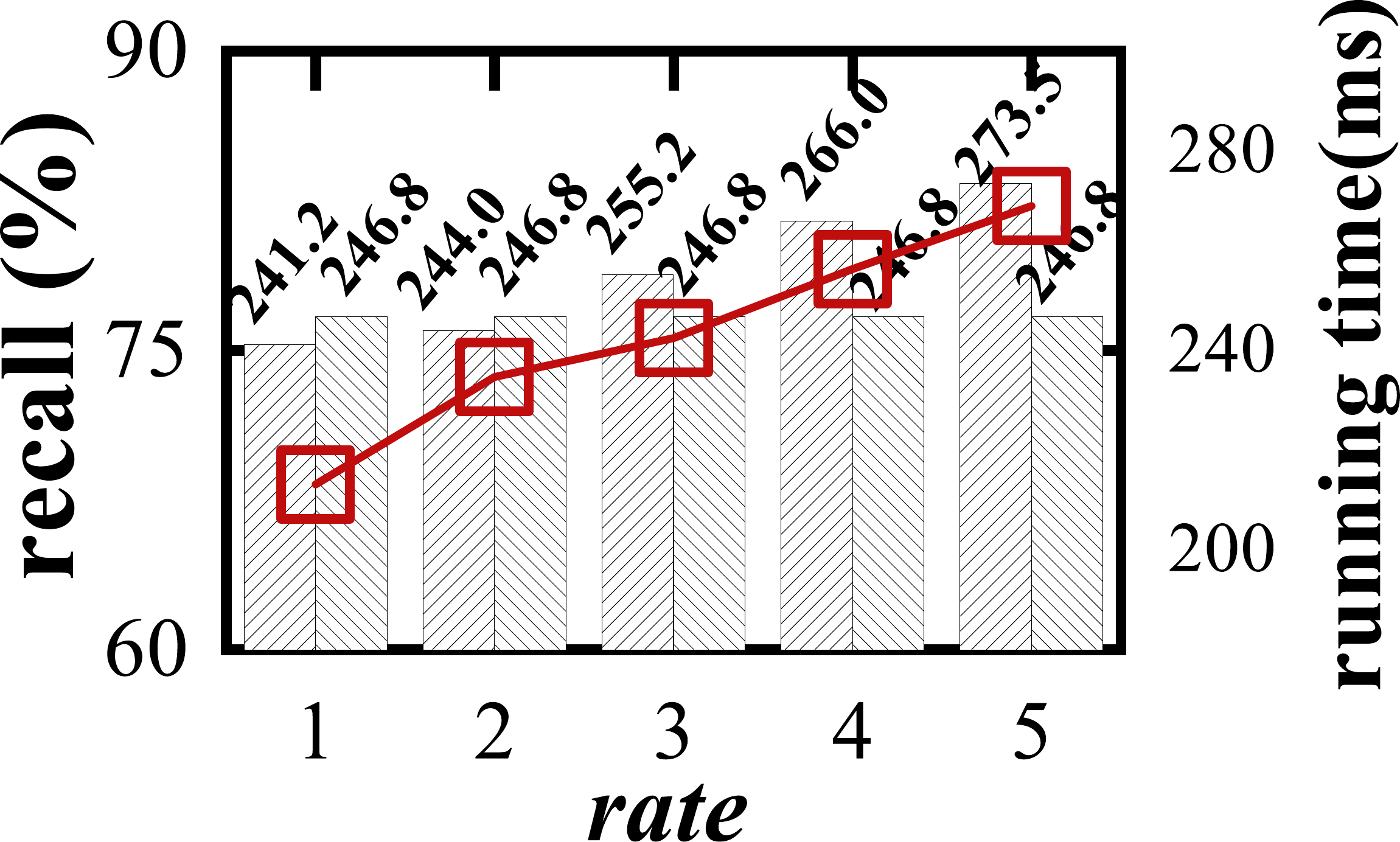}
}

\vspace{-0.3cm}
\caption{{Comparison with vector database {\sf Milvus}}}
\vspace{-8mm}
\label{fig:vectordatabase}
\end{center}
\end{figure}

\noindent\textbf{Comparison with Vector Database.} We choose the state-of-the-art vector database {\sf Milvus} for comparision, as it supports distributed multi-vector $k$NN search. It is important to note that {\sf Milvus} cannot support multi-vector range search,  and its multi-vector search involves executing search requests on multiple vector fields and merging the results using a re-ranking strategy. Therefore, both the $k$NN count for a single vector and multi-vectors need to be set, and this ratio directly affects the recall rate and efficiency of multi-vector $k$NN searches. Fig.~\ref{fig:vectordatabase} illustrates {\sf Milvus} 's search performance compared with OneDB as the query ratio varies between 1 and 5. It is observed that while {\sf Milvus} has a slightly lower runtime than OneDB at lower ratios, its recall rate is only 70\%. As recall approaches 85\% with higher query ratios, the search efficiency decreases by 10\% compared to OneDB. This indicates that in multimodal $k$NN search tasks, {\sf Milvus} 's multi-vector search struggles to balance efficiency with high recall rates.

\vspace{-0.2cm}

\subsection{Scalability Performance}

 \begin{figure}[t]
 \vspace{-1cm}
\begin{center}
\subfigtopskip=-7pt
\subfigcapskip=-3pt
\includegraphics[height=0.3cm]{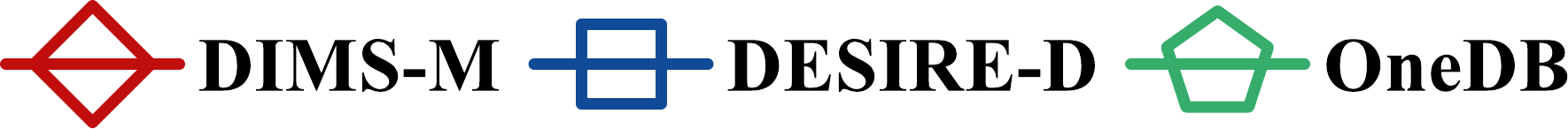}\vspace{0.1cm}
\subfigure[\textit{Food}]{
  \label{fig:rnn-02}
  \includegraphics[width=0.23\textwidth,height=0.14\textwidth]{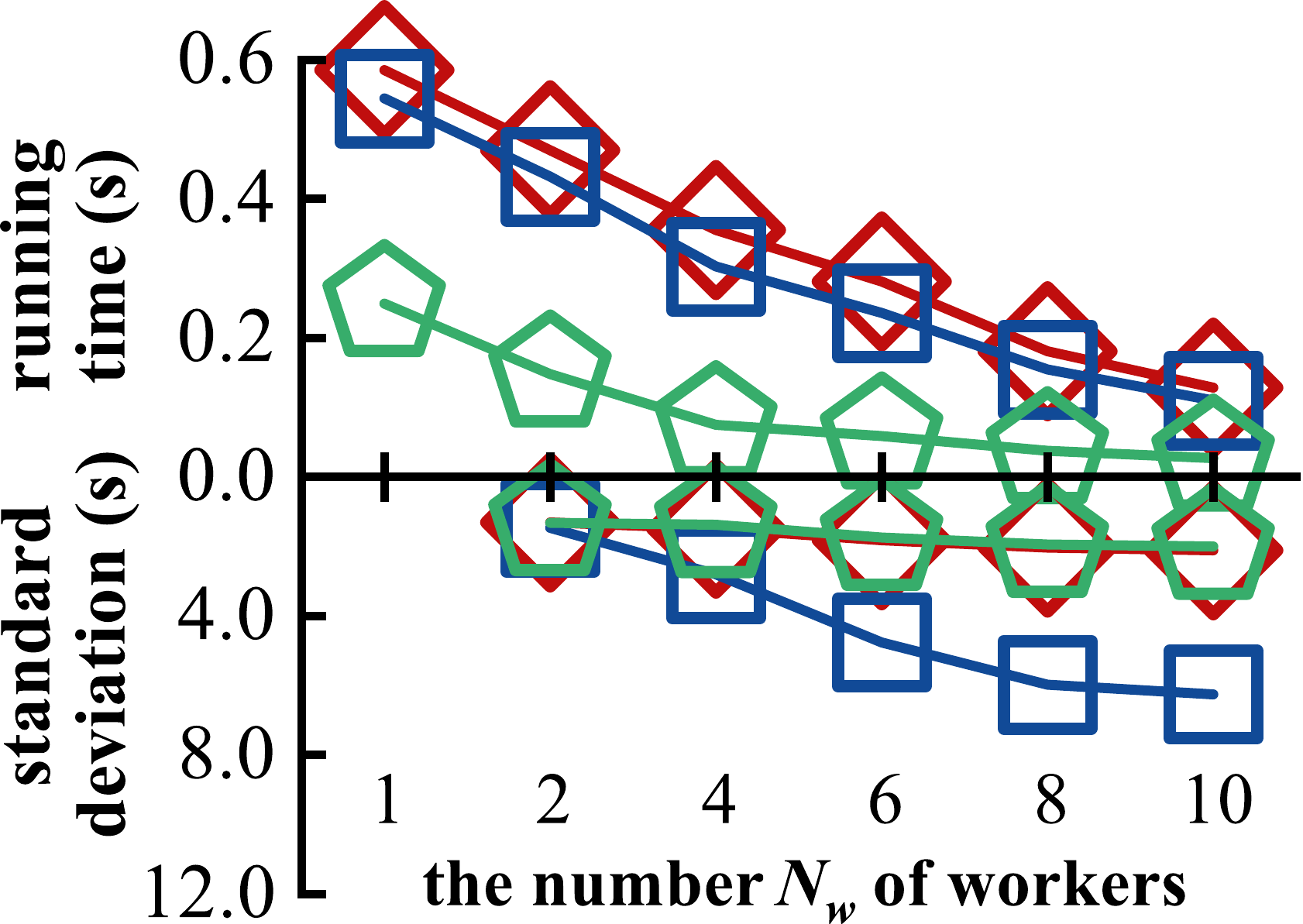}
}
\subfigure[\textit{Rental}]{
  \label{fig:rnn-02}
  \includegraphics[width=0.23\textwidth,height=0.14\textwidth]{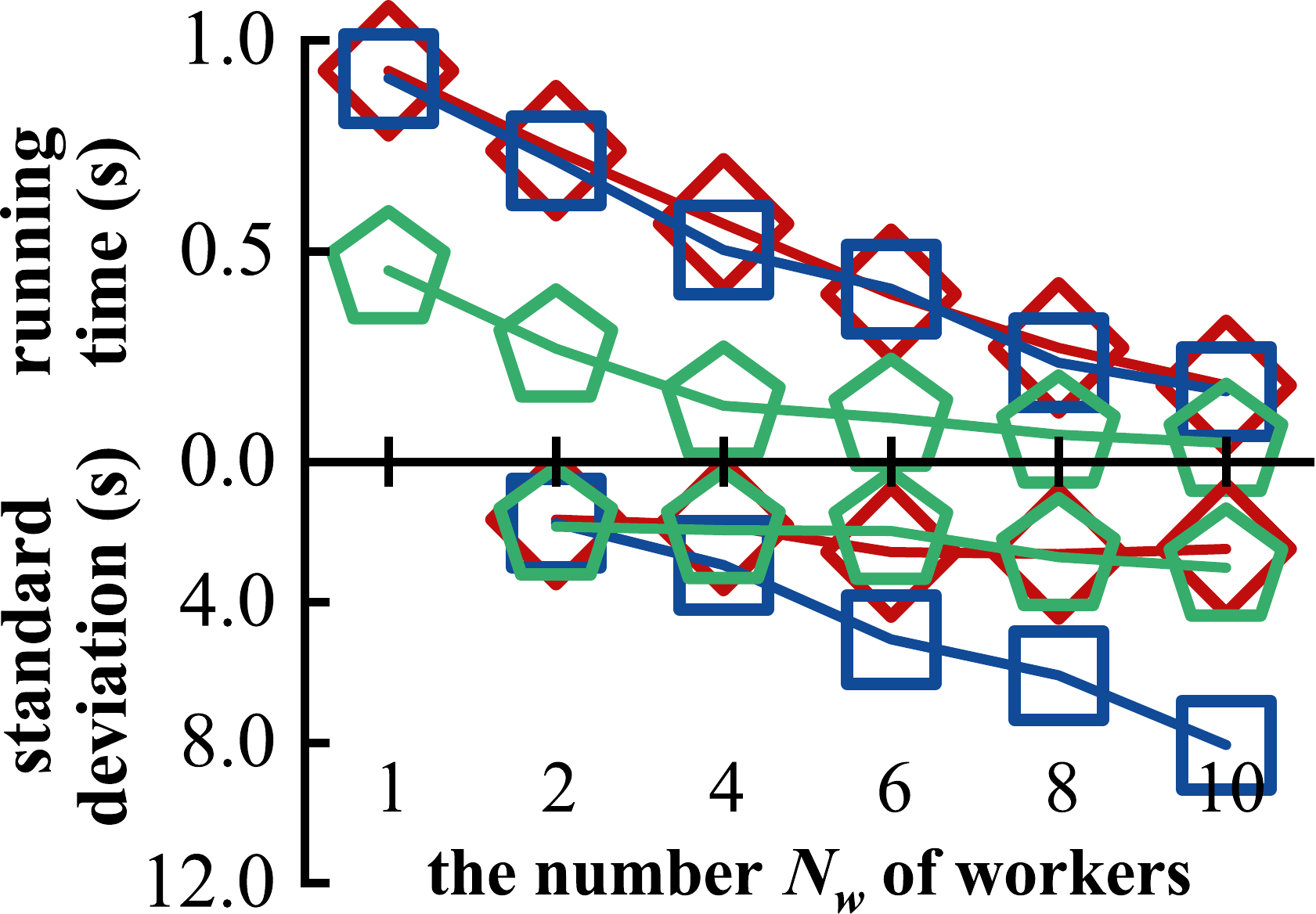}
}
\vspace{-0.2cm}
\subfigure[\textit{Air}]{
  \label{fig:rnn-02}
  \includegraphics[width=0.23\textwidth,height=0.14\textwidth]{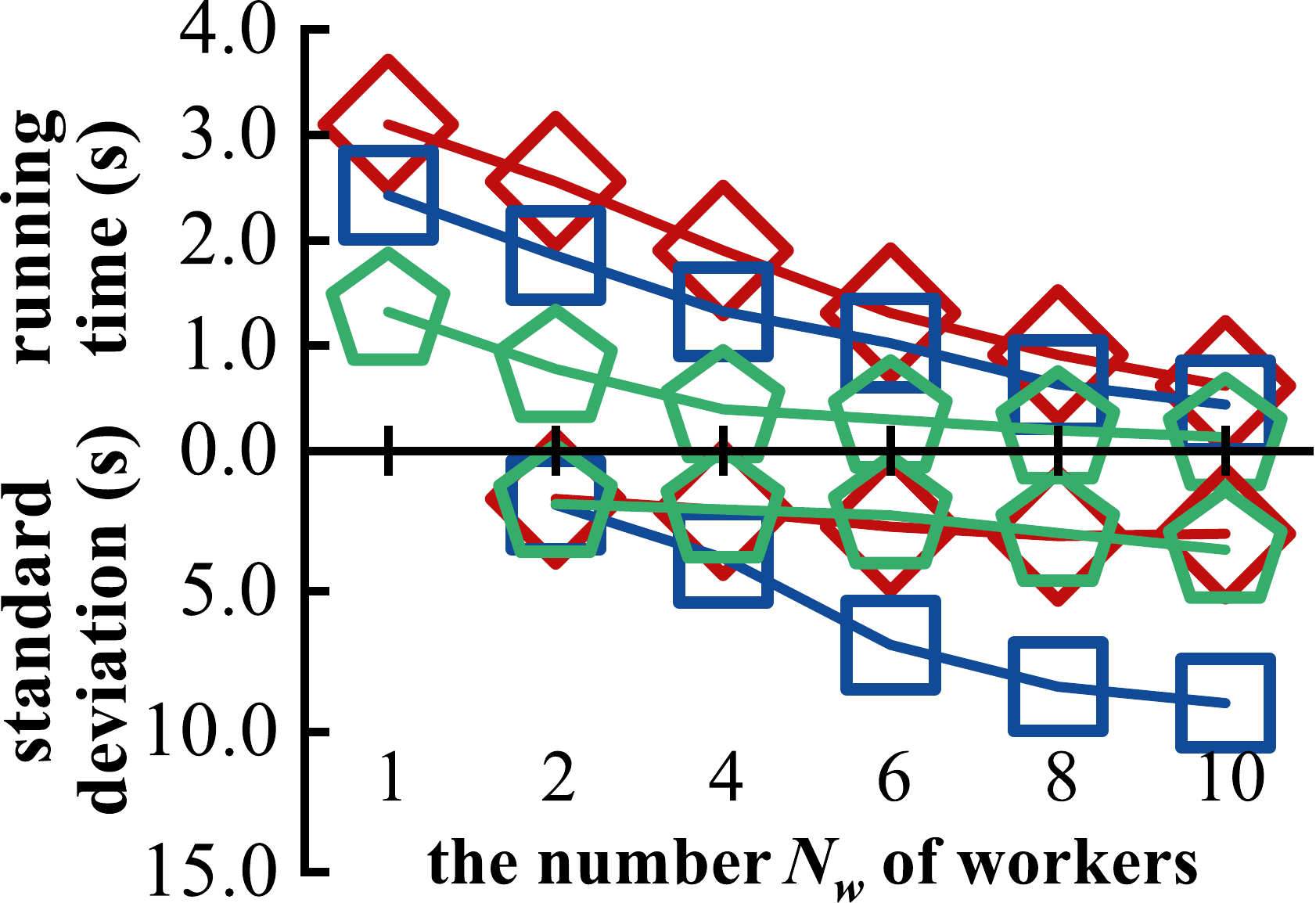}
}
\subfigure[{\textit{Synthesis II}}]{
  \label{fig:rnn-02}
  \includegraphics[width=0.23\textwidth,height=0.14\textwidth]{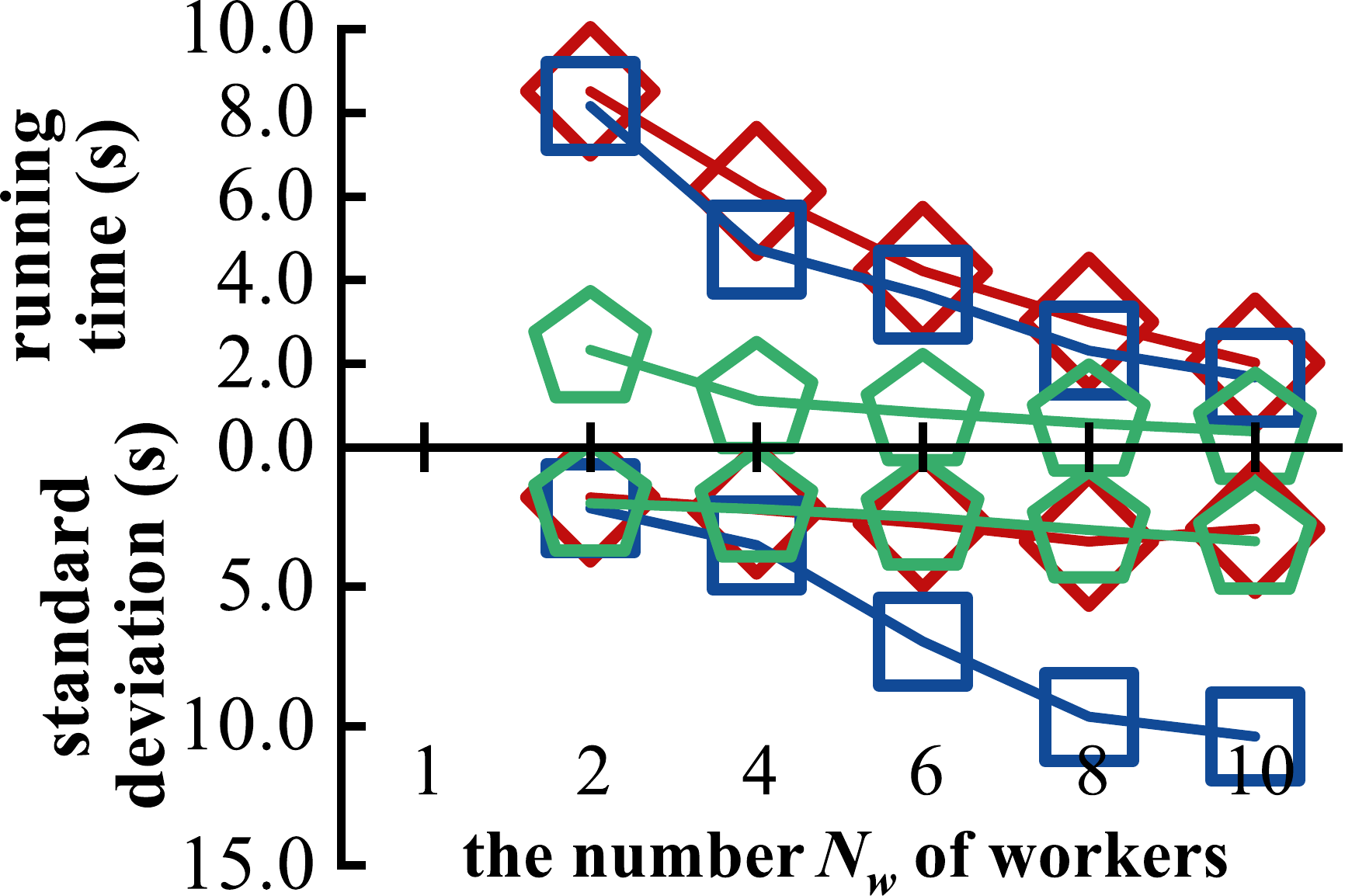}
}
\vspace{-0.3cm}
\caption{{Effect of the number $N_w$ of workers }}
\vspace{-0.8cm}
\label{fig:worker}
\end{center}
\end{figure}

\textbf{{Effects of worker number $N_w$.}} To evaluate the scalability performance of OneDB, we vary the number of workers using the MMRQ tasks. Figure~\ref{fig:worker} depicts the running time and the standard deviation of the worker's workload. It is observed that as the number of workers increases, the query time for all methods decreases, and OneDB consistently leads by a factor of 1.7x to 4.8x in performance. Simultaneously, the standard deviation of the workload increases because it is harder to distribute evenly across computing nodes when more nodes are included. OneDB and {\sf DIMS-M} significantly outperform {\sf DESIRE-D} in terms of standard deviation. Compared to {\sf DIMS-M}, OneDB has a slight advantage because a dual-layer indexing strategy is employed in the global index of OneDB. This strategy performs coarse-grained, even data partitioning at the global level. At the same time, optimized index structures are used locally for each modality, thereby achieving more efficient data access and query performance.

\begin{figure}[t]
\vspace{0.05cm}
\begin{center}
\subfigtopskip=-7pt
\subfigcapskip=-3pt
\includegraphics[height=0.3cm]{ExpFigs/icon_3.pdf}\vspace{0.1cm}
\subfigure[\textit{Air}]{
  \label{fig:rnn-02}
  \includegraphics[width=0.23\textwidth,height=0.14\textwidth]{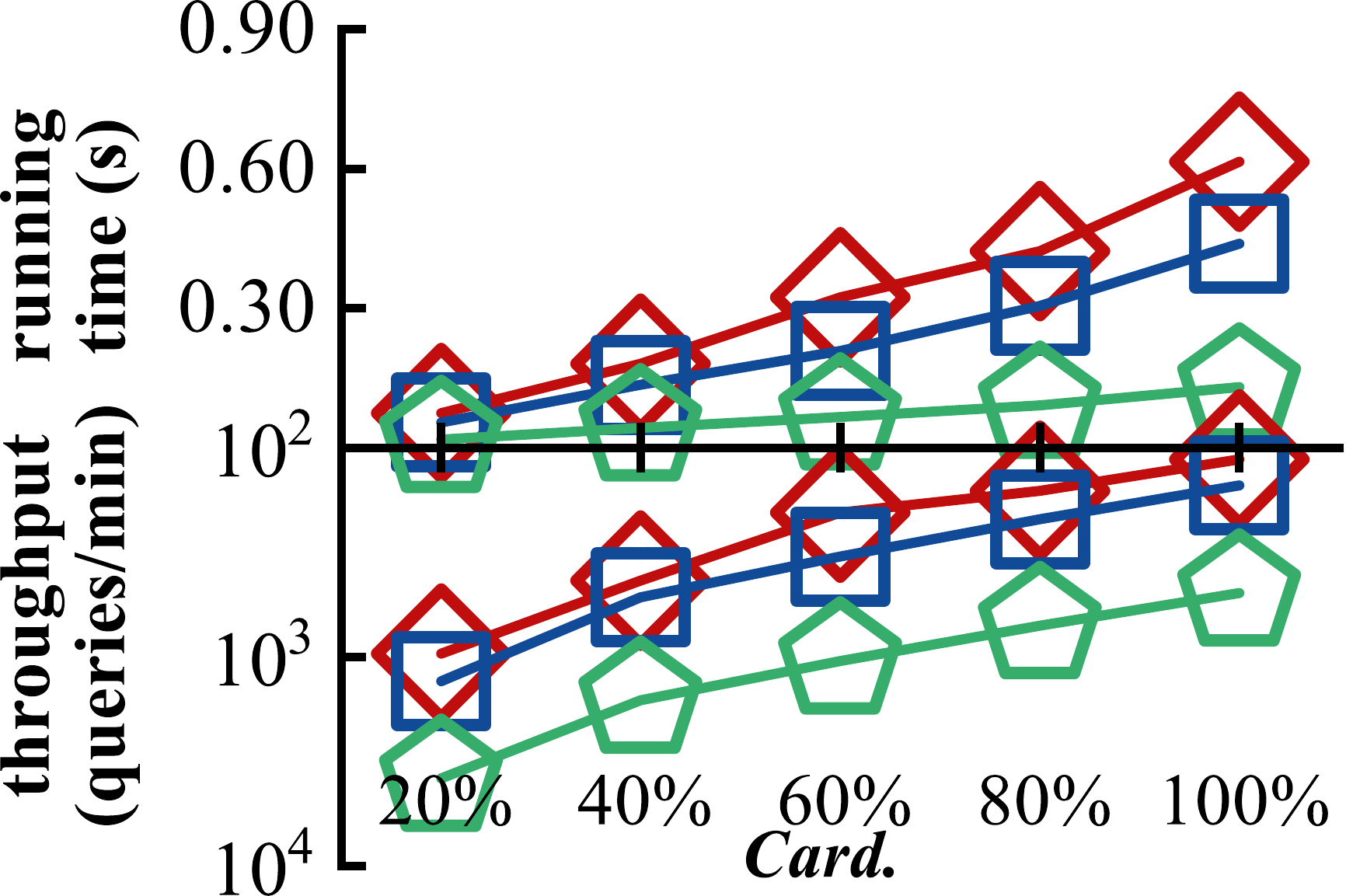}
}
\subfigure[\textit{Food}]{
  \label{fig:rnn-02}
  \includegraphics[width=0.23\textwidth,height=0.14\textwidth]{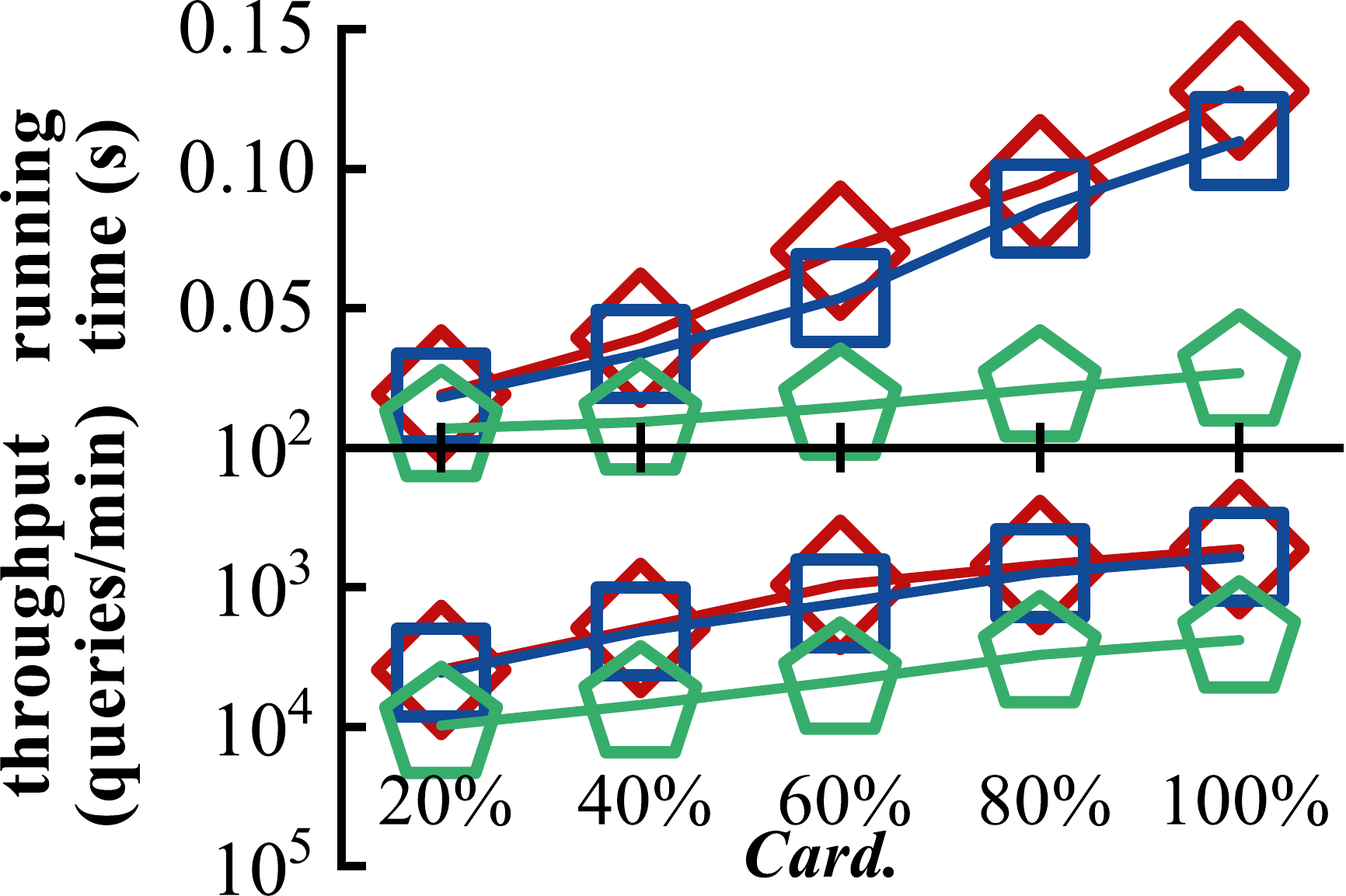}
}

\subfigure[\textit{Rental}]{
  \label{fig:rnn-02}
  \includegraphics[width=0.23\textwidth,height=0.14\textwidth]{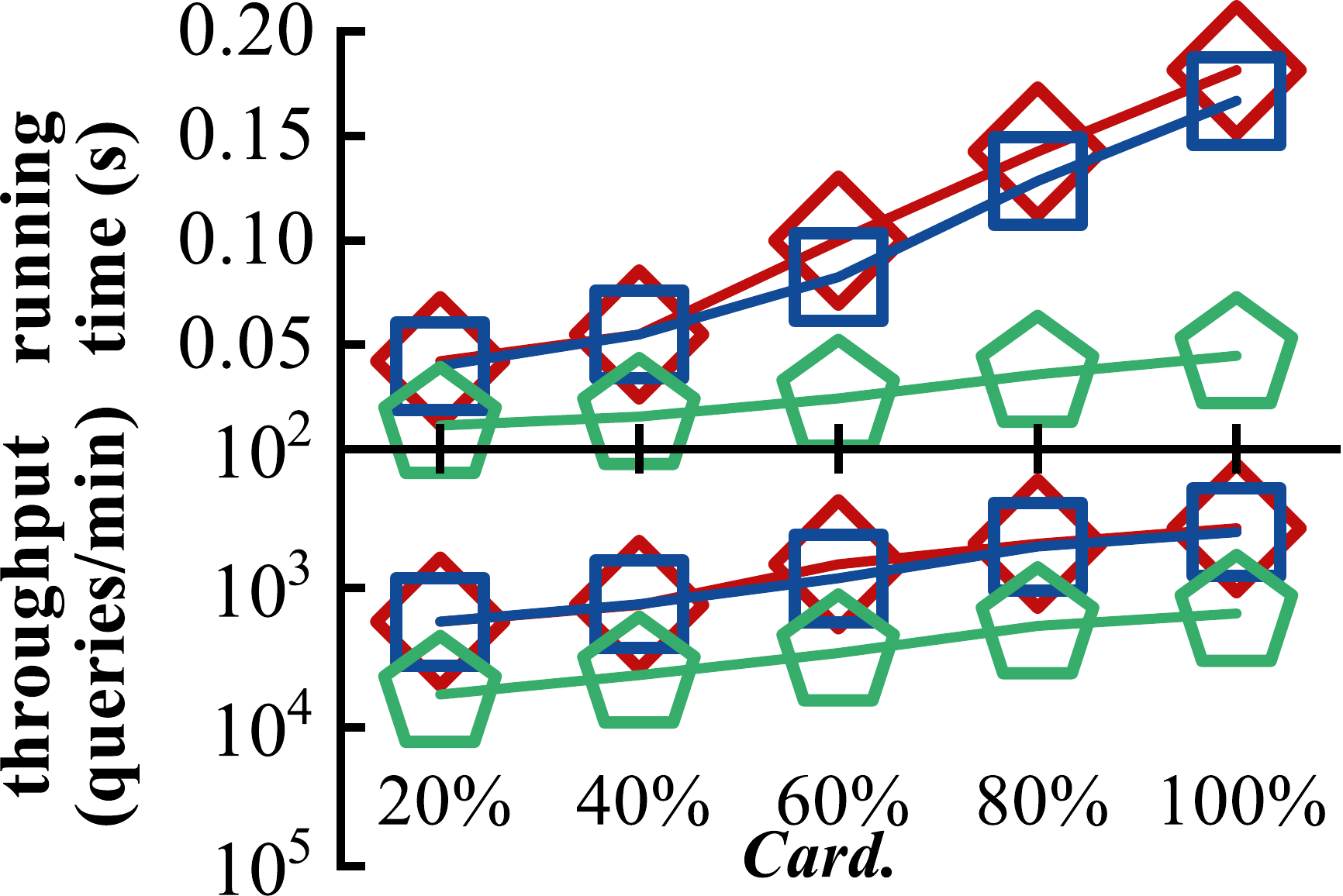}
}
\subfigure[\textit{Synthesis}]{
  \label{fig:rnn-02}
  \includegraphics[width=0.23\textwidth,height=0.14\textwidth]{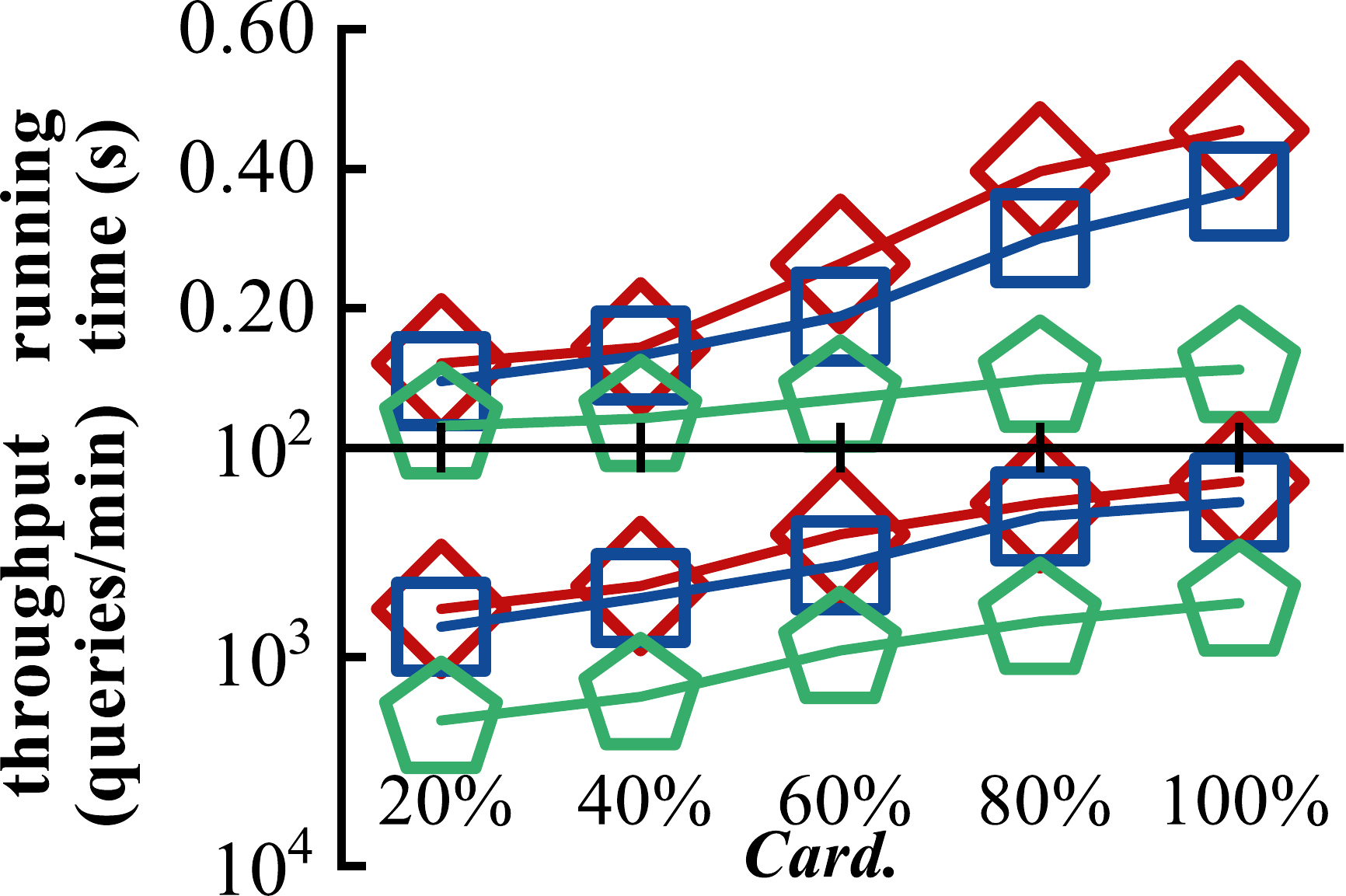}
}
\vspace{-0.3cm}
\caption{{MMRQ Performance vs. Cardinality}}
\vspace{-0.7cm}
\label{fig:mmrq_card_throughput}
\end{center}
\end{figure}

\noindent\textbf{Effect of Cardinality.} Fig.~\ref{fig:mmrq_card_throughput} illustrates the effect of dataset cardinality on  MMRQ performance, with the cardinality varying from 20\% to 100\%. The results indicate that the running time escalates significantly as the dataset size increases, while throughput exhibits a linear decline. This is because when the system encounters an increase in cardinality, it must process a significantly larger volume of data, resulting in a more complex search space and, consequently, higher computational and memory requirements. Despite these challenges, OneDB consistently delivers stable performance across various dataset sizes and superior adaptability with larger datasets. This is achieved through its effective resource distribution and optimized query processing, outperforming other methods. These advantages lie in a flexible indexing mechanism and an efficient load-balancing strategy, enabling OneDB to swiftly handle query requests even as data volumes surge while maintaining high throughput. 


\vspace{-0.15cm}
\subsection{Ablation Study}



\begin{figure}[t]
\begin{center}
\vspace{-1cm}

\subfigtopskip=-7pt
\subfigcapskip=-3pt
\includegraphics[height=0.25cm]{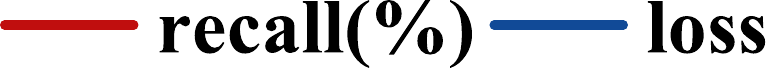}\vspace{0.15cm}
\subfigure[\textit{True Positive Food}]{
  \label{fig:rnn-02}
  \includegraphics[width=0.23\textwidth,height=0.14\textwidth]{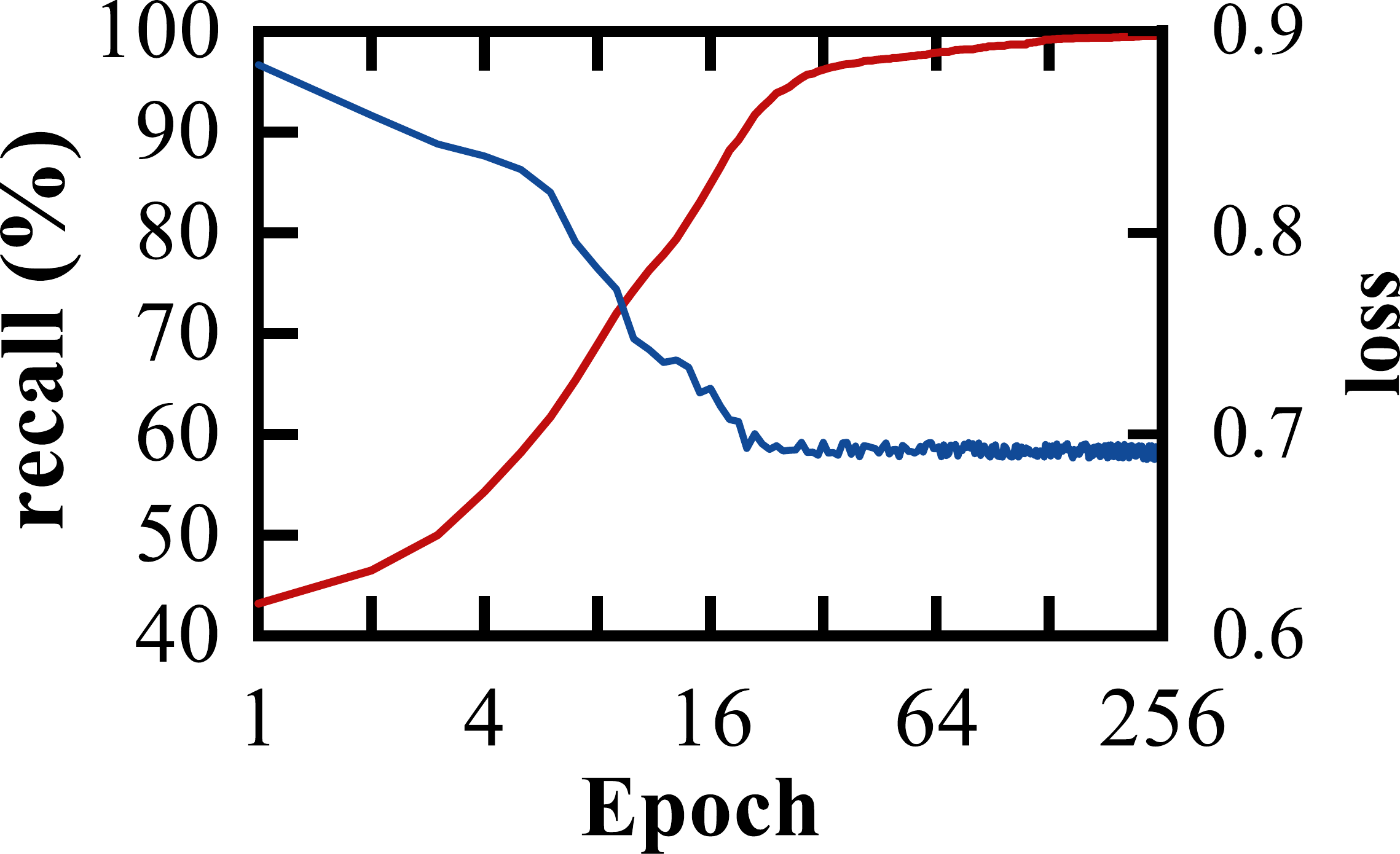}
}
\subfigure[\textit{Random Food}]{
  \label{fig:rnn-02}
  \includegraphics[width=0.23\textwidth,height=0.14\textwidth]{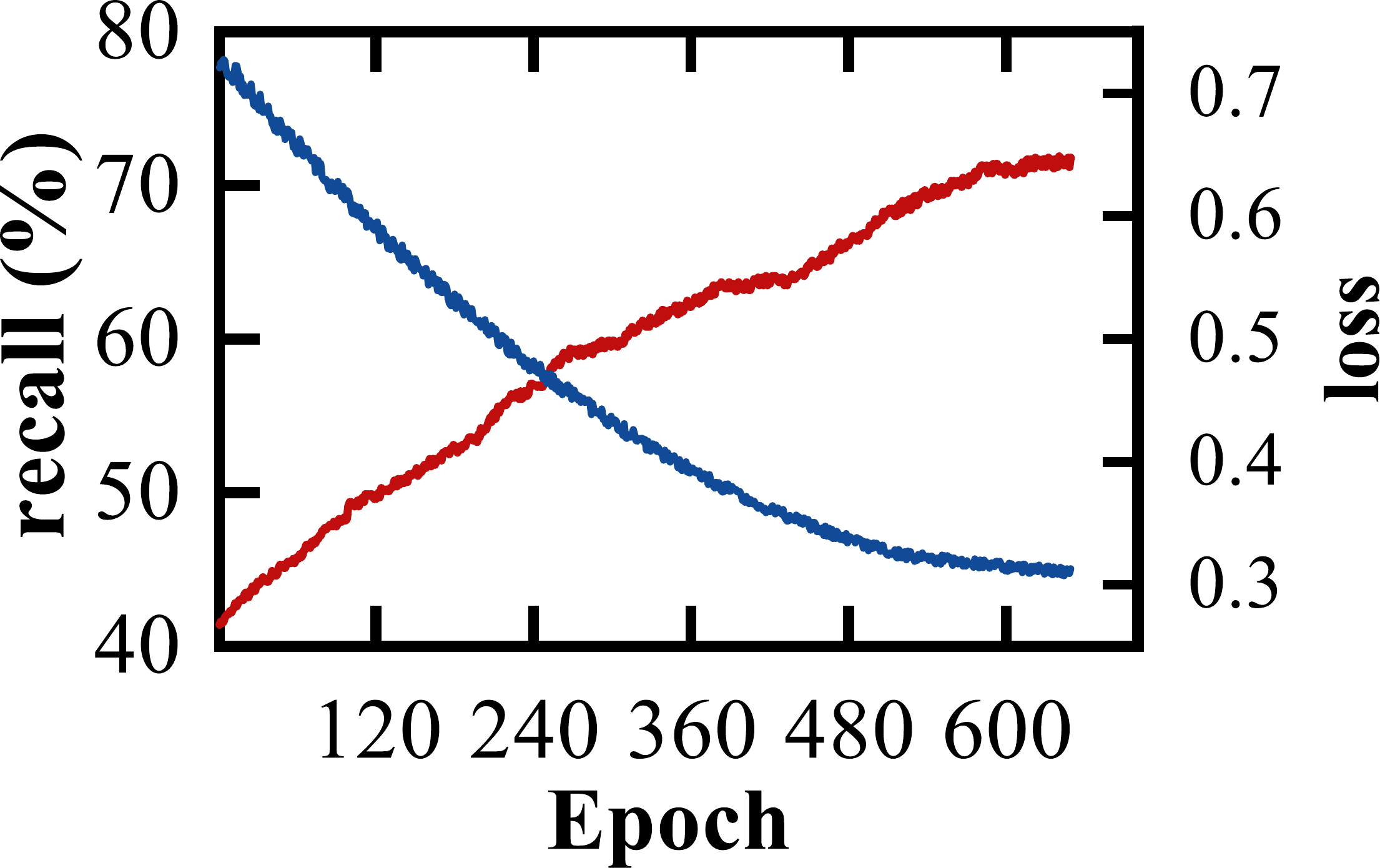}
}

\subfigure[\textit{True Positive Rental}]{
  \label{fig:rnn-02}
  \includegraphics[width=0.23\textwidth,height=0.14\textwidth]{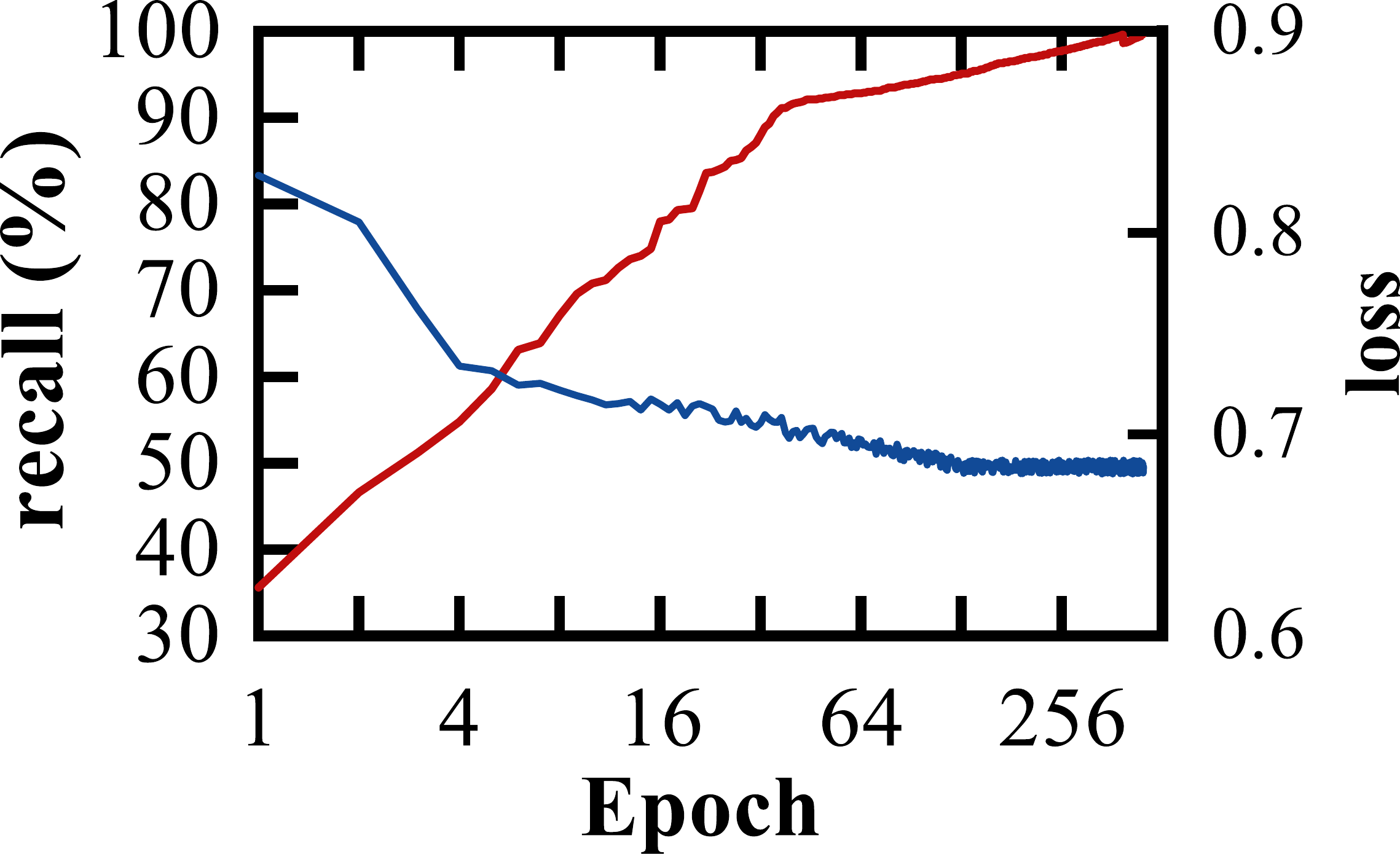}
}
\subfigure[\textit{Random Rental}]{
  \label{fig:rnn-02}
  \includegraphics[width=0.23\textwidth,height=0.14\textwidth]{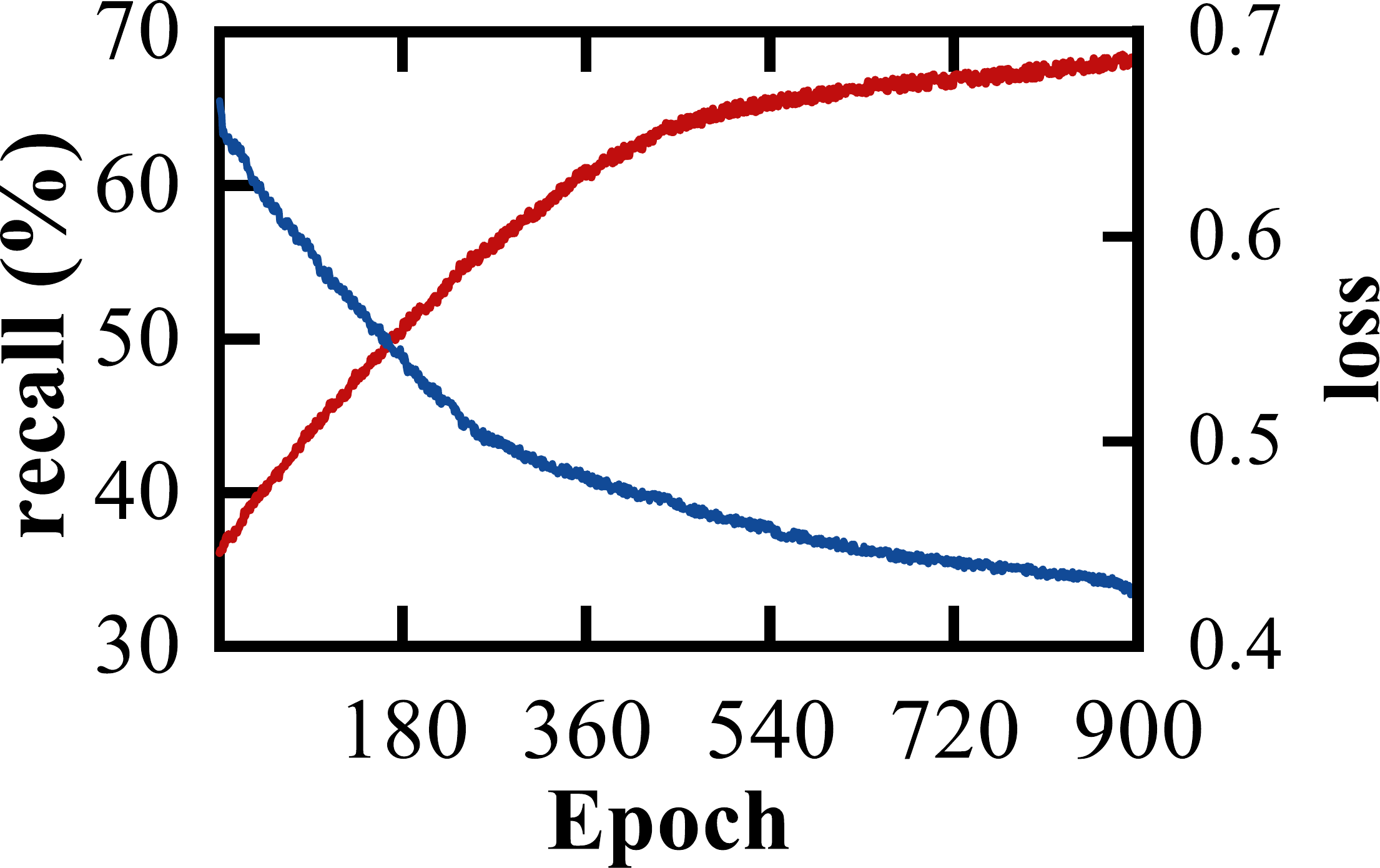}
}

\vspace{-0.3cm}
\caption{{Effect of Weight Learning}}
\vspace{-0.7cm}
\label{fig:weight}
\end{center}
\end{figure}

\textbf{Weight Learning.} We randomly select 30 historical query points, each associated with a known query weight vector and its corresponding 50 ground truth $k$-NN results. To simulate a realistic scenario where users do not explicitly specify modality weights but instead provide representative queries and desirable results, we hide the original weights during training. The learning model  simulates real-world situations where users express their retrieval intent through examples rather than numeric weight settings 
We compare the performance of our \textit{$k$-NN Positive and Negative Sample Generation Strategy} ({non-ground-truth results from the $k$-NN search in each iteration as negative samples}) with \textit{Randomly Selected Negative Samples.} The effectiveness of the weight learning process is evaluated by monitoring both the loss and recall metrics, as illustrated in Figure~\ref{fig:weight}. Compared to the random negative sample generation strategy, our model demonstrates stable convergence and a higher recall rate. This indicates that the weight adjustment process effectively refines the learning parameters over iterations. This stability suggests that the optimization strategy is robust, allowing the model to find an optimal solution without oscillations or divergence.
Furthermore, the weight learning model completes training under {100} seconds to achieve 90\% recall, exhibiting exceptional efficiency. These observations indicate that the multi-modal weight learning model is lightweight and effective.

{To further understand the interpretability of the learned weights, we conduct a case study comparing the modality-wise characteristics of the retrieved 64 $k$-NN queries. Since directly presenting multi-modal retrieval outputs is often un-intuitive, we instead compute the average value for each modality among two kinds of retrieved results with different weights, and compare them to those of manually constructed demo results. As illustrated in Figure~\ref{fig:weight}, the radar charts show that the average value among each modality of the results retrieved using the learned weights closely match the distribution of the manually constructed demo results across all modalities. In contrast, the uniform-weight baseline often over- or under-emphasizes certain modalities, leading to less aligned retrievals. This alignment suggests that the learned weights reflect the user’s intended modality emphasis more faithfully, even without explicitly specified weights.}

\begin{figure}[t]
 \vspace{-1cm}
    \centering    
    \includegraphics[width=1\linewidth,keepaspectratio]{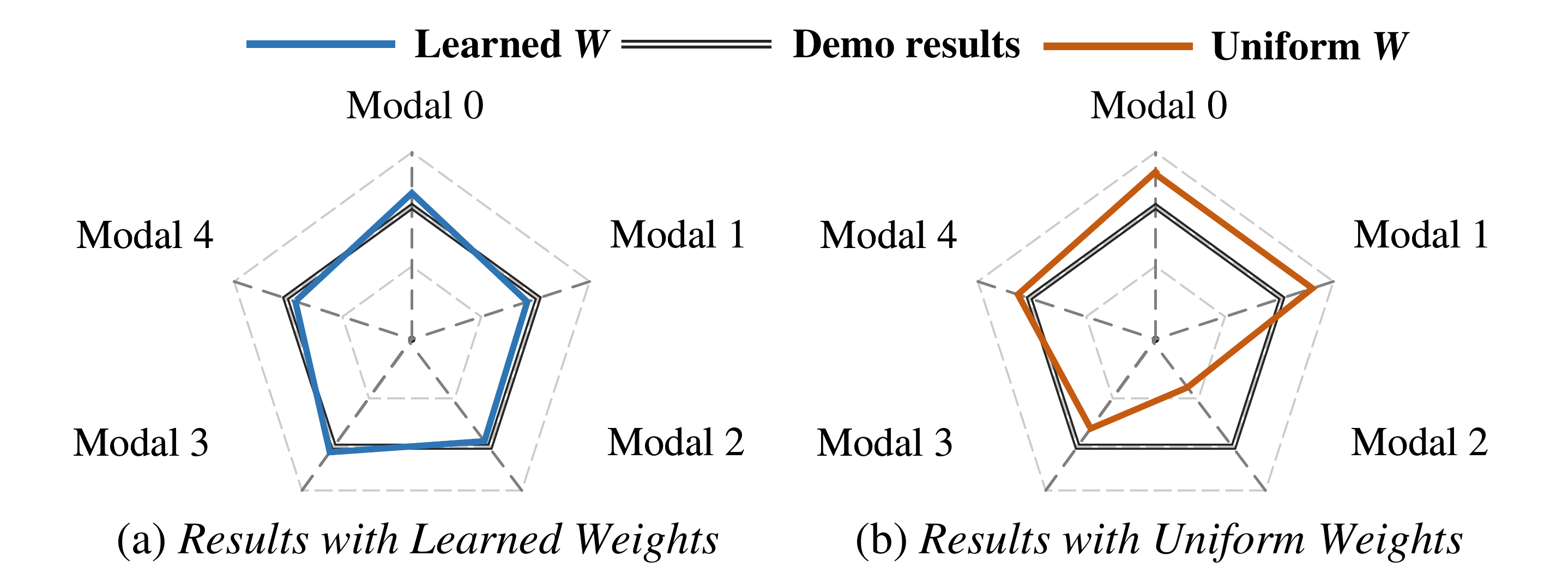} 
    \vspace{-6mm}
   \caption{Case Study}
   \label{fig:case-study}
 \vspace{-0.2cm}
\end{figure}

\begin{figure}[t]
\vspace{-0.1cm}
\begin{center}
\subfigtopskip=-7pt
\subfigcapskip=-3pt
\includegraphics[height=0.25cm]{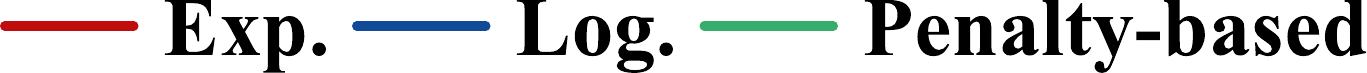}\vspace{0.15cm}

\subfigure[\textit{Air}]{
  \label{fig:rnn-02}
  \includegraphics[width=0.23\textwidth,height=0.14\textwidth]{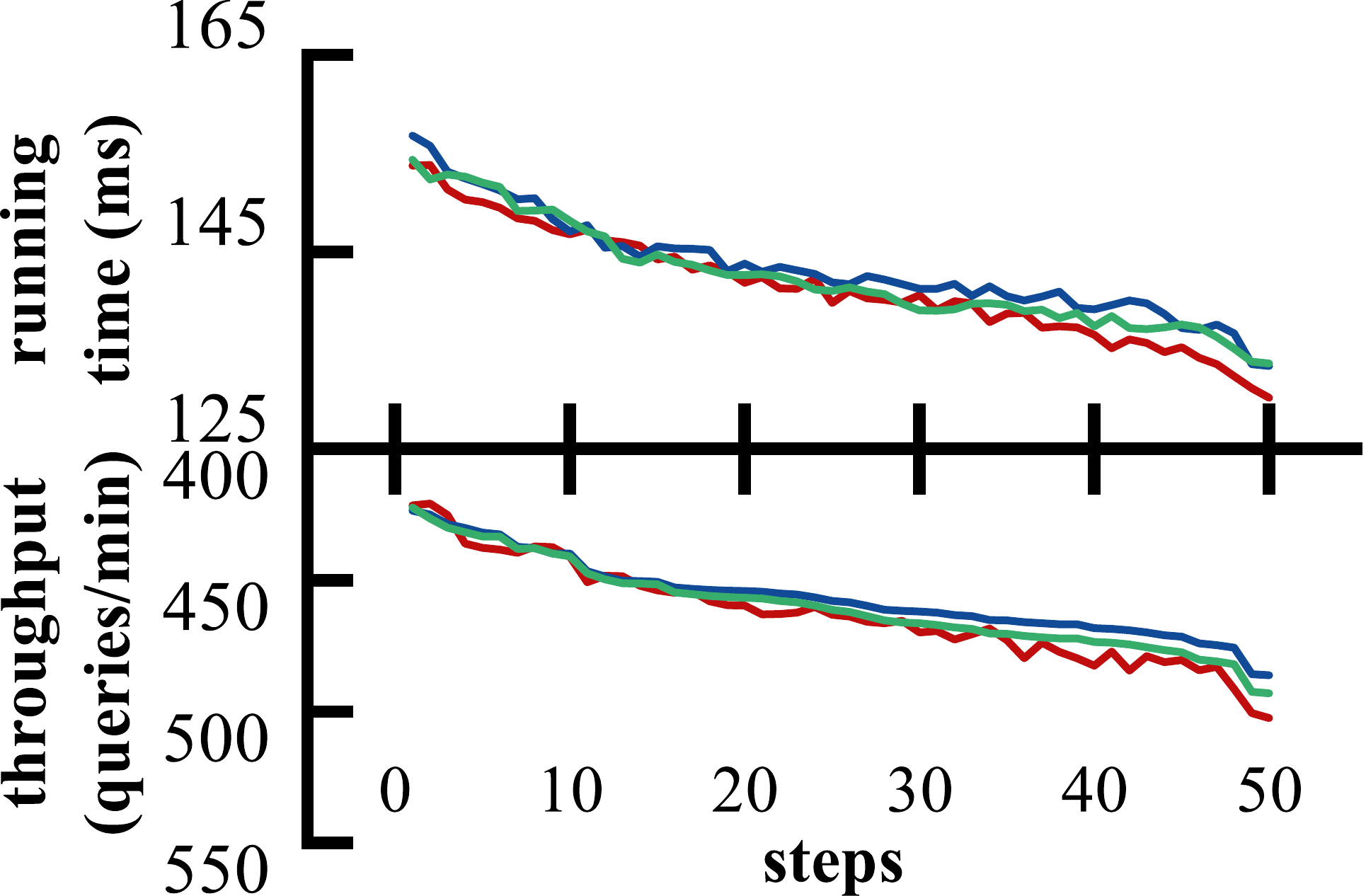}
}
\subfigure[\textit{Synthesis}]{
  \label{fig:rnn-02}
  \includegraphics[width=0.23\textwidth,height=0.14\textwidth]{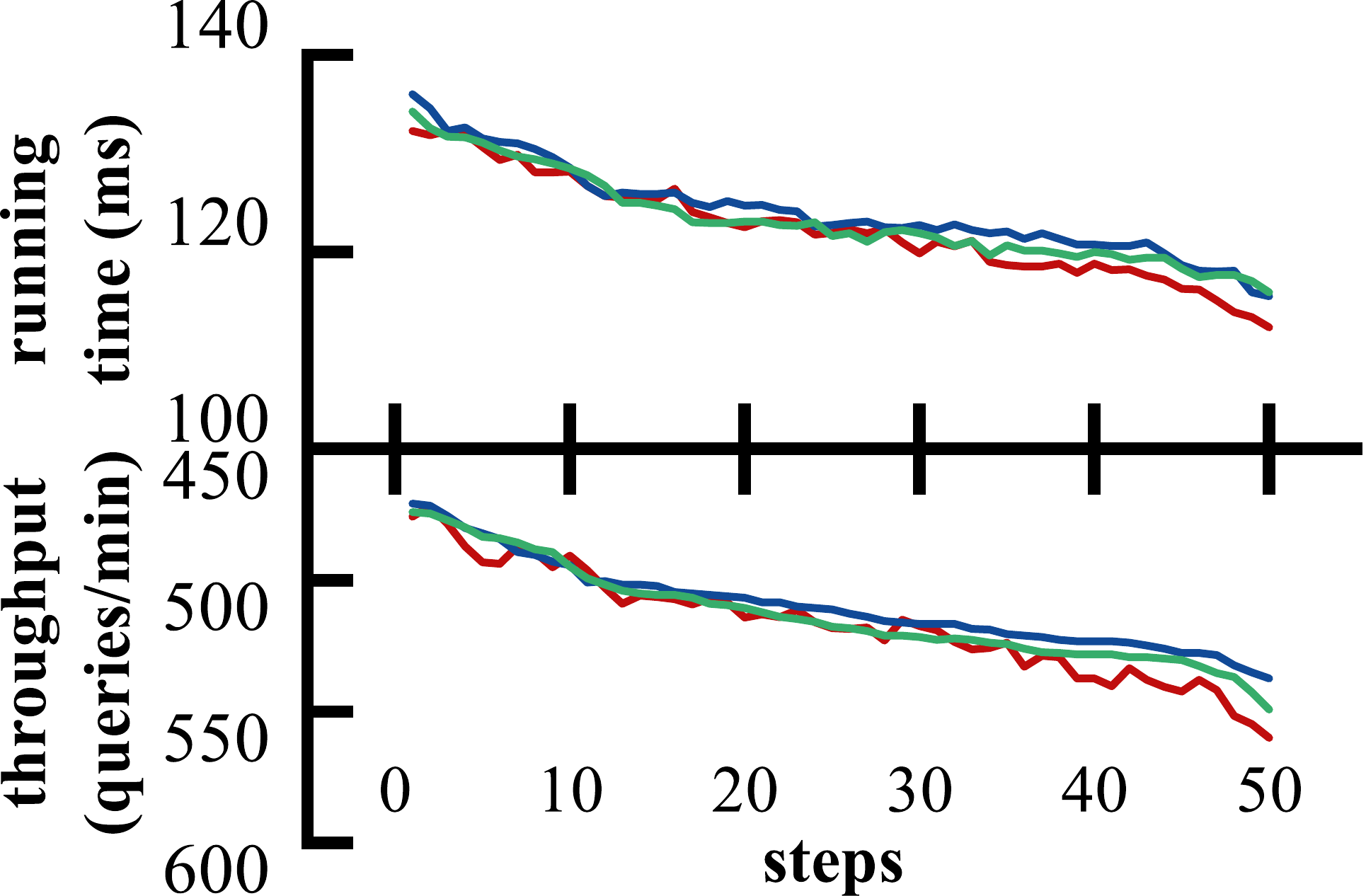}
}

\vspace{-0.2cm}
\caption{{Effect of Different Reward Function}}
\vspace{-0.7cm}
\label{fig:reward_function}
\end{center}
\end{figure}


 
\noindent\textbf{Parameter Tuning.} Fig.~\ref{fig:reward_function} shows the end-to-end parameter tuning training process. We optimized the parameters of OneDB and two other baselines through the tuning module on the "Synthesis" dataset. It can be observed that after 50 steps of training, the search performance improved by approximately 15\%-17\%, showcasing the effectiveness of the adopted tuning module. Fig.~\ref{fig:reward_function} illustrates the performance of three different reward function variants during the tuning process. Specifically, the exponentially weighted function accelerated performance improvements in the early stages by amplifying larger gains, allowing the system to reach optimal configurations quickly, but gains tapered off later. Logarithmically weighted function handled performance fluctuations smoothly, with slower initial gains but more stable and continuous improvements later. The penalty-based function strictly controls performance drops, avoiding unstable configurations and showing stable improvements in the later stages but slower early growth. 
{These results highlight the efficiency improvements achieved by the reinforcement learning-based end-to-end tuning framework in OneDB. Among the three strategies, the exponential reward function consistently achieved the best final performance in both latency and throughput, making it the most effective option in practice. While the penalty-based function provides more conservative adjustments, the exponential variant delivers faster convergence and higher overall gains, making it better for responsive systems.}

\section{Conclusion}
\label{sec:conclusions}

In this paper, we present OneDB, a distributed multi-metric data similarity search system on multi-modal datasets. OneDB integrates a lightweight metric weight learning model, a dual-layer indexing strategy, and an end-to-end parameter auto-tuning mechanism to support flexible and accurate similarity search across heterogeneous data modalities. 
Experimental results demonstrate that OneDB outperforms existing state-of-the-art solutions regarding efficiency and accuracy, particularly in handling diverse user preferences. In the future, we plan to explore dynamic data scenarios.

\bibliographystyle{IEEEtran}
\balance

\begin{thebibliography}{10}
\providecommand{\url}[1]{#1}
\csname url@samestyle\endcsname
\providecommand{\newblock}{\relax}
\providecommand{\bibinfo}[2]{#2}
\providecommand{\BIBentrySTDinterwordspacing}{\spaceskip=0pt\relax}
\providecommand{\BIBentryALTinterwordstretchfactor}{4}
\providecommand{\BIBentryALTinterwordspacing}{\spaceskip=\fontdimen2\font plus
\BIBentryALTinterwordstretchfactor\fontdimen3\font minus \fontdimen4\font\relax}
\providecommand{\BIBforeignlanguage}[2]{{%
\expandafter\ifx\csname l@#1\endcsname\relax
\typeout{** WARNING: IEEEtran.bst: No hyphenation pattern has been}%
\typeout{** loaded for the language `#1'. Using the pattern for}%
\typeout{** the default language instead.}%
\else
\language=\csname l@#1\endcsname
\fi
#2}}
\providecommand{\BIBdecl}{\relax}
\BIBdecl

\bibitem{jeong2024multimodal}
J.~Jeong, K.~Tian, A.~Li, S.~Hartung, S.~Adithan, F.~Behzadi, J.~Calle, D.~Osayande, M.~Pohlen, and P.~Rajpurkar, ``Multimodal image-text matching improves retrieval-based chest x-ray report generation,'' in \emph{Medical Imaging with Deep Learning}.\hskip 1em plus 0.5em minus 0.4em\relax PMLR, 2024, pp. 978--990.

\bibitem{palacios2019wip}
S.~Palacios, K.~Solaiman, P.~Angin, A.~Nesen, B.~Bhargava, Z.~Collins, A.~Sipser, M.~Stonebraker, and J.~Macdonald, ``Wip-skod: A framework for situational knowledge on demand,'' in \emph{Heterogeneous Data Management, Polystores, and Analytics for Healthcare: VLDB 2019 Workshops, Poly and DMAH, Los Angeles, CA, USA, August 30, 2019, Revised Selected Papers 5}.\hskip 1em plus 0.5em minus 0.4em\relax Springer, 2019, pp. 154--166.

\bibitem{afyouni2022multi}
I.~Afyouni, Z.~Al~Aghbari, and R.~A. Razack, ``Multi-feature, multi-modal, and multi-source social event detection: A comprehensive survey,'' \emph{Information Fusion}, vol.~79, pp. 279--308, 2022.

\bibitem{aguerrebere2023similarity}
C.~Aguerrebere, I.~Bhati, M.~Hildebrand, M.~Tepper, and T.~Willke, ``Similarity search in the blink of an eye with compressed indices,'' \emph{arXiv preprint arXiv:2304.04759}, 2023.

\bibitem{huang2023one}
Y.~Huang, F.~Luo, X.~Wang, Z.~Di, B.~Li, and B.~Luo, ``A one-size-fits-three representation learning framework for patient similarity search,'' \emph{Data Science and Engineering}, vol.~8, no.~3, pp. 306--317, 2023.

\bibitem{shang2018dita}
Z.~Shang, G.~Li, and Z.~Bao, ``Dita: Distributed in-memory trajectory analytics,'' in \emph{Proceedings of the 2018 International Conference on Management of Data}, 2018, pp. 725--740.

\bibitem{icde/ZhengWZZ0J21}
B.~Zheng, L.~Weng, X.~Zhao, K.~Zeng, X.~Zhou, and C.~S. Jensen, ``{REPOSE:} distributed top-k trajectory similarity search with local reference point tries,'' in \emph{ICDE}, 2021, pp. 708--719.

\bibitem{waim/ZhuSKNY12}
M.~Zhu, D.~Shen, Y.~Kou, T.~Nie, and G.~Yu, ``An adaptive distributed index for similarity queries in metric spaces,'' in \emph{{WAIM}}, vol. 7418.\hskip 1em plus 0.5em minus 0.4em\relax Springer, 2012, pp. 222--227.

\bibitem{icde/ChenGLJC15}
L.~Chen, Y.~Gao, X.~Li, C.~S. Jensen, and G.~Chen, ``Efficient metric indexing for similarity search,'' in \emph{{ICDE}}, 2015, pp. 591--602.

\bibitem{zhu2024dimsdistributedindexsimilarity}
\BIBentryALTinterwordspacing
Y.~Zhu, C.~Luo, T.~Qian, L.~Chen, Y.~Gao, and B.~Zheng, ``Dims: Distributed index for similarity search in metric spaces,'' 2024. [Online]. Available: \url{https://arxiv.org/abs/2410.05091}
\BIBentrySTDinterwordspacing

\bibitem{franzke2016indexing}
M.~Franzke, T.~Emrich, A.~Z{\"u}fle, and M.~Renz, ``Indexing multi-metric data,'' in \emph{2016 IEEE 32nd International Conference on Data Engineering (ICDE)}.\hskip 1em plus 0.5em minus 0.4em\relax IEEE, 2016, pp. 1122--1133.

\bibitem{wei2020analyticdb}
C.~Wei, B.~Wu, S.~Wang, R.~Lou, C.~Zhan, F.~Li, and Y.~Cai, ``Analyticdb-v: a hybrid analytical engine towards query fusion for structured and unstructured data,'' \emph{Proceedings of the VLDB Endowment}, vol.~13, no.~12, pp. 3152--3165, 2020.

\bibitem{yang2020pase}
W.~Yang, T.~Li, G.~Fang, and H.~Wei, ``Pase: Postgresql ultra-high-dimensional approximate nearest neighbor search extension,'' in \emph{Proceedings of the 2020 ACM SIGMOD international conference on management of data}, 2020, pp. 2241--2253.

\bibitem{elasticsearch2018elasticsearch}
B.~Elasticsearch, ``Elasticsearch,'' \emph{software], version}, vol.~6, no.~1, 2018.

\bibitem{pgvector}
``pgvector,'' \url{https://github.com/pgvector/pgvector}, accessed: 2024-10-14.

\bibitem{zhang2023vbase}
Q.~Zhang, S.~Xu, Q.~Chen, G.~Sui, J.~Xie, Z.~Cai, Y.~Chen, Y.~He, Y.~Yang, F.~Yang \emph{et~al.}, ``$\{$VBASE$\}$: Unifying online vector similarity search and relational queries via relaxed monotonicity,'' in \emph{17th USENIX Symposium on Operating Systems Design and Implementation (OSDI 23)}, 2023, pp. 377--395.

\bibitem{wang2021milvus}
J.~Wang, X.~Yi, R.~Guo, H.~Jin, P.~Xu, S.~Li, X.~Wang, X.~Guo, C.~Li, X.~Xu \emph{et~al.}, ``Milvus: A purpose-built vector data management system,'' in \emph{Proceedings of the 2021 International Conference on Management of Data}, 2021, pp. 2614--2627.

\bibitem{zhu2022desire}
Y.~Zhu, L.~Chen, Y.~Gao, B.~Zheng, and P.~Wang, ``Desire: An efficient dynamic cluster-based forest indexing for similarity search in multi-metric spaces,'' \emph{Proceedings of the VLDB Endowment}, vol.~15, no.~10, p. 2121, 2022.

\bibitem{sun2017dima}
J.~Sun, Z.~Shang, G.~Li, D.~Deng, and Z.~Bao, ``Dima: A distributed in-memory similarity-based query processing system,'' \emph{Proceedings of the VLDB Endowment}, vol.~10, no.~12, pp. 1925--1928, 2017.

\bibitem{xie2016simba}
D.~Xie, F.~Li, B.~Yao, G.~Li, L.~Zhou, and M.~Guo, ``Simba: Efficient in-memory spatial analytics,'' in \emph{Proceedings of the 2016 international conference on management of data}, 2016, pp. 1071--1085.

\bibitem{53317}
\BIBentryALTinterwordspacing
M.~Phothilimthana, S.~Kadekodi, S.~Ghodrati, S.~Moon, and M.~Maas, ``Thesios: Synthesizing accurate counterfactual i/o traces from i/o samples,'' in \emph{ASPLOS 2024}, 2024. [Online]. Available: \url{https://dl.acm.org/doi/10.1145/3620666.3651337}
\BIBentrySTDinterwordspacing

\bibitem{chen2022indexing}
L.~Chen, Y.~Gao, X.~Song, Z.~Li, Y.~Zhu, X.~Miao, and C.~S. Jensen, ``Indexing metric spaces for exact similarity search,'' \emph{ACM Computing Surveys}, vol.~55, no.~6, pp. 1--39, 2022.

\bibitem{ipl/Uhlmann91}
J.~K. Uhlmann, ``Satisfying general proximity/similarity queries with metric trees,'' \emph{Inf. Process. Lett.}, vol.~40, no.~4, pp. 175--179, 1991.

\bibitem{vldb/CiacciaPZ97}
P.~Ciaccia, M.~Patella, and P.~Zezula, ``M-tree: An efficient access method for similarity search in metric spaces,'' in \emph{VLDB}, 1997, pp. 426--435.

\bibitem{mico1994new}
M.~L. Mic{\'o}, J.~Oncina, and E.~Vidal, ``A new version of the nearest-neighbour approximating and eliminating search algorithm (aesa) with linear preprocessing time and memory requirements,'' \emph{Pattern Recognition Letters}, vol.~15, no.~1, pp. 9--17, 1994.

\bibitem{sigmod/BozkayaO97}
T.~Bozkaya and Z.~M. {\"{O}}zsoyoglu, ``Distance-based indexing for high-dimensional metric spaces,'' in \emph{{SIGMOD}}, 1997, pp. 357--368.

\bibitem{tods/BozkayaO99}
------, ``Indexing large metric spaces for similarity search queries,'' \emph{{ACM} Trans. Database Syst.}, vol.~24, no.~3, pp. 361--404, 1999.

\bibitem{vldb/Brin95}
S.~Brin, ``Near neighbor search in large metric spaces,'' in \emph{VLDB}, 1995, pp. 574--584.

\bibitem{kalantari1983data}
I.~Kalantari and G.~McDonald, ``A data structure and an algorithm for the nearest point problem,'' \emph{IEEE Transactions on Software Engineering}, no.~5, pp. 631--634, 1983.

\bibitem{tods/CiacciaP02}
P.~Ciaccia and M.~Patella, ``Searching in metric spaces with user-defined and approximate distances,'' \emph{{ACM} Trans. Database Syst.}, vol.~27, no.~4, pp. 398--437, 2002.

\bibitem{ciaccia2000m2}
------, ``The m2-tree: Processing complex multi-feature queries with just one index.'' in \emph{DELOS}, 2000.

\bibitem{bustos2012adapting}
B.~Bustos, S.~Kreft, and T.~Skopal, ``Adapting metric indexes for searching in multi-metric spaces,'' \emph{Multimedia Tools and Applications}, vol.~58, no.~3, pp. 467--496, 2012.

\bibitem{7498318}
M.~Franzke, T.~Emrich, A.~Züfle, and M.~Renz, ``Indexing multi-metric data,'' in \emph{2016 IEEE 32nd International Conference on Data Engineering (ICDE)}, 2016, pp. 1122--1133.

\bibitem{celik2006new}
C.~Celik, \emph{New approaches to similarity searching in metric spaces}.\hskip 1em plus 0.5em minus 0.4em\relax University of Maryland, College Park, 2006.

\bibitem{sac/BustosKS05}
\BIBentryALTinterwordspacing
B.~Bustos, D.~Keim, and T.~Schreck, ``A pivot-based index structure for combination of feature vectors,'' in \emph{Proceedings of the 2005 ACM Symposium on Applied Computing}, ser. SAC '05.\hskip 1em plus 0.5em minus 0.4em\relax New York, NY, USA: Association for Computing Machinery, 2005, p. 1180–1184. [Online]. Available: \url{https://doi.org/10.1145/1066677.1066945}
\BIBentrySTDinterwordspacing

\bibitem{zabot2019efficient}
G.~F. Zabot, M.~T. Cazzolato, L.~C. Scabora, A.~J. Traina, and C.~Traina, ``Efficient indexing of multiple metric spaces with spectra,'' in \emph{2019 IEEE International Symposium on Multimedia (ISM)}.\hskip 1em plus 0.5em minus 0.4em\relax IEEE, 2019, pp. 169--1697.

\bibitem{batko2005similarity}
M.~Batko, C.~Gennaro, and P.~Zezula, ``Similarity grid for searching in metric spaces,'' in \emph{Peer-to-Peer, Grid, and Service-Orientation in Digital Library Architectures: 6th Thematic Workshop of the EU Network of Excellence DELOS, Cagliari, Italy, June 24-25, 2004. Revised Selected Papers}.\hskip 1em plus 0.5em minus 0.4em\relax Springer, 2005, pp. 25--44.

\bibitem{is/NovakBZ11}
D.~Novak, M.~Batko, and P.~Zezula, ``Metric index: An efficient and scalable solution for precise and approximate similarity search,'' \emph{Inf. Syst.}, vol.~36, no.~4, pp. 721--733, 2011.

\bibitem{dpd/DoulkeridisVKV09}
C.~Doulkeridis, A.~Vlachou, Y.~Kotidis, and M.~Vazirgiannis, ``Efficient range query processing in metric spaces over highly distributed data,'' \emph{Distributed Parallel Databases}, vol.~26, no. 2-3, pp. 155--180, 2009.

\bibitem{tods/JagadishOTYZ05}
H.~V. Jagadish, B.~C. Ooi, K.~Tan, C.~Yu, and R.~Zhang, ``idistance: An adaptive b\({}^{\mbox{+}}\)-tree based indexing method for nearest neighbor search,'' \emph{{ACM} Trans. Database Syst.}, vol.~30, no.~2, pp. 364--397, 2005.

\bibitem{dase/YangDZCZG19}
K.~Yang, X.~Ding, Y.~Zhang, L.~Chen, B.~Zheng, and Y.~Gao, ``Distributed similarity queries in metric spaces,'' \emph{Data Sci. Eng.}, vol.~4, no.~2, pp. 93--108, 2019.

\bibitem{chen2017distributed}
Z.~Chen, G.~Cong, Z.~Zhang, T.~Z. Fuz, and L.~Chen, ``Distributed publish/subscribe query processing on the spatio-textual data stream,'' in \emph{2017 IEEE 33rd International Conference on Data Engineering (ICDE)}.\hskip 1em plus 0.5em minus 0.4em\relax IEEE, 2017, pp. 1095--1106.

\bibitem{sun2019balance}
J.~Sun, Z.~Shang, G.~Li, D.~Deng, and Z.~Bao, ``Balance-aware distributed string similarity-based query processing system,'' \emph{Proceedings of the VLDB Endowment}, vol.~12, no.~9, pp. 961--974, 2019.

\bibitem{mao2012pivot}
R.~Mao, W.~L. Miranker, and D.~P. Miranker, ``Pivot selection: Dimension reduction for distance-based indexing,'' \emph{Journal of Discrete Algorithms}, vol.~13, pp. 32--46, 2012.

\bibitem{zhang2021cdbtune+}
J.~Zhang, K.~Zhou, G.~Li, Y.~Liu, M.~Xie, B.~Cheng, and J.~Xing, ``Cdbtune+: An efficient deep reinforcement learning-based automatic cloud database tuning system,'' \emph{The VLDB Journal}, vol.~30, no.~6, pp. 959--987, 2021.

\end{thebibliography}

\end{document}